\pdfoutput=1
\documentclass[aps,pre,superscriptaddress,notitlepage,onecolumn]{revtex4-1}
\usepackage{natbib}
\usepackage{graphicx}
\usepackage{array}
\usepackage{caption}
\usepackage{enumerate}
\usepackage{subfig}
\usepackage{float}
\usepackage[hyperindex,breaklinks]{hyperref}
\usepackage{cleveref}
\usepackage{amssymb,amsfonts,amsmath,amsthm}
\usepackage{pseudocode}
\renewcommand{\figurename}{Figure}
\renewcommand{\tablename}{Table}
\renewcommand{\thetable}{\arabic{table}}
\newtheorem{myprop}{Proposition}

\newtheorem{mylemma}{Lemma}
\newtheorem{mythm}{Theorem}

\begin{document}

\title{Cell fate reprogramming by control of intracellular network dynamics}
\author{Jorge G. T. Za\~nudo}\email[]{jgtz@phys.psu.edu}
\affiliation{Department of Physics, The Pennsylvania State University,\\
  University Park, Pennsylvania, 16802-6300, USA.}
\author{R\'eka Albert}
\affiliation{Department of Physics, The Pennsylvania State University,\\
  University Park, Pennsylvania, 16802-6300, USA.}
\affiliation{Department of Biology, The Pennsylvania State University,\\
  University Park, Pennsylvania, 16802-5301, USA.}

\begin{abstract}
{Identifying control strategies for biological networks is paramount for practical applications that involve reprogramming a cell's fate, such as disease therapeutics and stem cell reprogramming. Here we develop a novel network control framework that integrates the structural and functional information available for intracellular networks to predict control targets. Formulated in a logical dynamic scheme, our approach drives any initial state to the target state with 100\% effectiveness and needs to be applied only transiently for the network to reach and stay in the desired state. We illustrate our method's potential to find intervention targets for cancer treatment and cell differentiation by applying it to a leukemia signaling network and to the network controlling the differentiation of helper T cells. We find that the predicted control targets are effective in a broad dynamic framework. Moreover, several of the predicted interventions are supported by experiments.}
\end{abstract}
\maketitle

\section*{Author Summary} \label{sec:0}

Practical applications in modern molecular and systems biology such as the search for new therapeutic targets for diseases and stem cell reprogramming have generated a great interest in controlling the internal dynamics of a cell. Here we present a network control approach that integrates the structural and functional information of the network. We show that stabilizing the expression or activity of a few select components can drive the cell towards a desired fate or away from an undesired fate. We demonstrate our method's effectiveness by applying it to a type of blood cell cancer and to the differentiation of a type of immune cell. Overall, our approach provides new insights into how to control the dynamics of intracellular networks.

\section*{Introduction} \label{sec:1}

An important task of modern molecular and systems biology is to achieve an understanding of the dynamics of the network of macromolecular interactions that underlies the functioning of cells. Practical applications such as stem cell reprogramming \cite{StemCellsTakahashi,CellReprogReview1,CellReprogReview2} and the search for new therapeutic targets for diseases \cite{SystemsMedicine,NetworkMedicine,
SystemsBiologyMedicine} have also motivated a great interest in the general task of cell fate reprogramming, i.e., controlling the internal state of a cell so that it is driven from an initial state to a final target state (see references \cite{BarabasiControllability,MullerSchuppertReply,BarabasiObservability,NodalDynamics,MotterControl,FVS1,FVS2}).

Theoretically derived control methods are based on simplified models of the interactions and/or the dynamics of cellular constituents such as proteins or mRNAs. Some of these models only include information on which cell components (e.g. molecules or proteins) interact among each other, i.e., the structure of the underlying interaction network. Other models, known as dynamic models, include the structure of the interaction network and also an equation for each component, which describes how the state of this component changes in time due to the influence of other cell components (e.g. how the concentration of a molecule changes in time due to the reactions the molecule participates in).

Although the topic of network controllability has a long history in control and systems theory (see, for example, \cite{Kalman,Luenberger,Slotine,Lin}), most of this work is not directly applicable to large intracellular networks. There are several reasons for this: (i) combinatorial complexity and the size of the matrices involved makes control theory applicable to small networks only, (ii) linear functions are used for the regulatory functions and it is unclear how the switch-like behavior of many biochemical processes \cite{TysonDynamics1,TysonDynamics2} will affect these results, and (iii) the notion of controllability in control theory, i.e. control of the full set of states \cite{Kalman,Luenberger,Slotine} or \textit{complete controllability}, is different from that in the biological sense, which commonly encompasses only the \textit{biologically admissible states} \cite{MullerSchuppertReply}.

In recent work on network controllability \cite{BarabasiControllability,BarabasiObservability,NodalDynamics,MotterControl,FVS1,FVS2,Akutsu,Cheng,Tamura} some of the limitations of standard control theory approaches are addressed. For example, Akutusu, Cheng, Tamura et al. \cite{Akutsu,Cheng,Tamura} extend the framework of control theory to systems with Boolean (switch-like) dynamics and provide some formal results in this setting. In the work of Liu et al. \cite{BarabasiControllability} the size limitation of linear control theory is overcome by using a maximal matching approach to identify the minimal number of nodes needed to control a variety of real-world large scale networks. Specifically, for some gene regulatory
networks, Liu et al. find that control of roughly $80\%$ of the nodes is needed to fully control the dynamics of these networks \cite{BarabasiControllability}. In contrast, experimental work in stem cell reprogramming suggests that for biologically admissible states the number of nodes required for control is drastically lower (five or fewer genes \cite{MullerSchuppertReply,StemCellsTakahashi,CellReprogReview1,CellReprogReview2}). Fiedler, Mochizuki et al. \cite{FVS1,FVS2} use the concept of the feedback vertex set, a subset of nodes in a directed network whose removal leaves the graph without directed cycles (i.e. without feedback loops). They show that, for a broad class of regulatory functions, controlling any feedback vertex set is enough to guide the dynamics of the system to any target trajectory of the uncontrolled network \cite{FVS1,FVS2}. As one of their examples, the authors use a signal transduction network with 113 elements and show that the minimal feedback vertex set is composed of only 5 elements.

Since systems whose interaction networks and dynamics are known equally well are rare, current control strategies are based on either the network structure \cite{BarabasiControllability,BarabasiObservability,NodalDynamics,FVS1,FVS2} or its dynamics (function) \cite{MotterControl,Akutsu,Cheng,Tamura}. Yet, as manipulating the activity of even a single intracellular component is a long, difficult, and expensive experimental task, it is crucial to reduce as much as possible the number of nodes that need to be controlled. We hypothesize that integrating network structure with qualitative information on the regulatory functions or on the target states of interest could yield control strategies with a small number of control targets. Qualitative information about the regulatory functions is commonly known (e.g. positive/negative regulation, cooperativity among regulators, etc.), and relative qualitative information on the desired/undesired states also exists (e.g. upregulation or downregulation of mRNA levels in a disease state with respect to a healthy state). Thus, we choose a logical dynamic framework as our modeling method \cite{LessIsMore}. This framework is well suited for modeling intracellular networks: discrete dynamic models have been shown to reproduce the qualitative dynamics of a multitude of cellular systems while requiring only the combinatorial activating or inhibiting nature of the interactions, and not the kinetic details \cite{MiskovTCell,ArabidopsisRoot,SaezRodriguezCancer,SocolarCellCycle,TLGLPNAS,SorgerReview,PhysBioReview}.

Logical dynamic network models \cite{KauffmanOriginal,GlassKauffman,GlassAsynchronous,ThomasReview,Chaves,AssiehJTB,Socolar,Laubenbacher} consist of a set of binary variables $\{\sigma_i\}$, $i=1, 2, \ldots, N$, each of which denotes the state of a node (also referred to as node state). The state ON (or 1) commonly refers to above a certain threshold level, while the state OFF (or 0) refers to below the same threshold level. The vector formed by the state of all nodes $(\sigma_1, \sigma_2, \ldots, \sigma_N)$ denotes the state of the system (or system/network state). To each node $v_i$ one assigns a Boolean function $f_i$ which contains the biological information on how node $v_i$'s inputs influence $\sigma_i$; these functions are used to evolve in time the state of each element. We use the general asynchronous updating scheme \cite{GlassAsynchronous,ThomasReview,AssiehJTB} (see \hyperref[Methods]{Methods}), a stochastic scheme which takes into consideration the variety of timescales present in intracellular processes and our incomplete knowledge of the rates of these processes.

\begin{figure}[t]
\centerline{\includegraphics[width=0.55\textwidth]{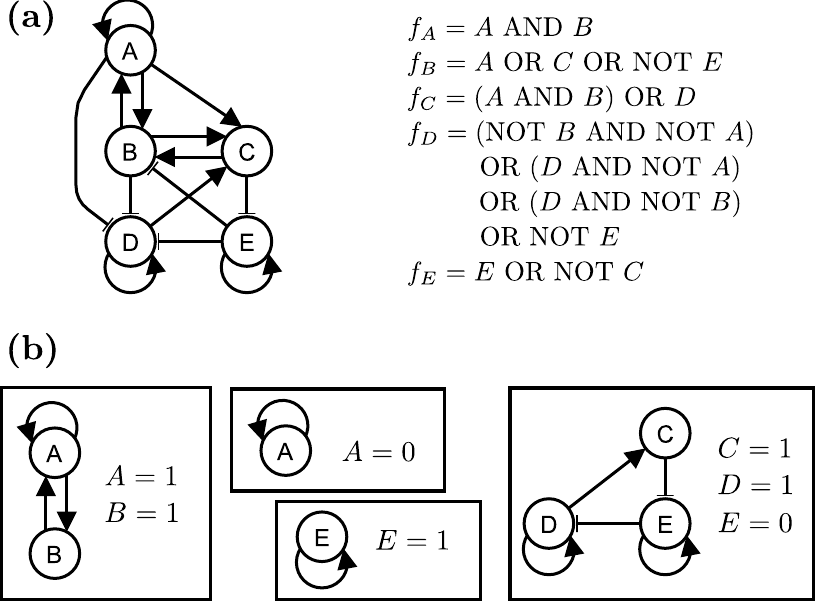}}
\caption{Stable motifs of a logical (Boolean) network. (a) An example of a logical network indicating the regulatory relationships and the logical update function of each node. (b) The four stable motifs of the logical network in (a) and their corresponding node states. These stable motifs are strongly connected components and partial fixed points of the logical network.}
\label{fig:NetworkExample}
\end{figure}

\begin{figure}[t]
\centerline{\includegraphics[width=0.6\textwidth]{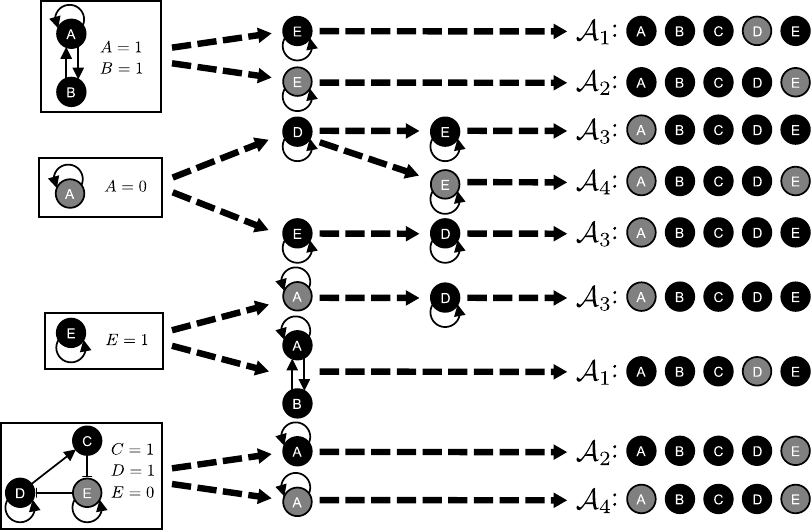}} \
\caption{Stable motif succession diagram for the example in Fig. \ref{fig:NetworkExample}. The stable motif succession diagram shows the stable motifs obtained successively during the attractor finding process and the attractors they finally lead to. A more detailed representation of the first steps of the attractor finding method is shown in Fig. \hyperref[fig:FigS1]{S1}. Nodes are colored based on their respective node states in the motifs or the attractors: gray for 0 and black for 1. The four stable motifs of the original logical network and their matching node states are shown in the leftmost part of the figure. The attractors obtained for each possible sequence of stable motifs are shown in the rightmost part of the figure. The result of applying network reduction using a stable motif is represented by each dashed arrow. If network reduction due to a stable motif leads to a simplified network with at least one stable motif, then the dashed arrows point from the stable motif being considered to the stable motifs of the simplified network. Otherwise, network reduction leads directly to an attractor and the dashed arrow points towards the attractor.}
\label{fig:ReductionMethod}
\end{figure}

In a logical (Boolean) model, every temporal trajectory must eventually reach a set of system states in which it settles down, known as an attractor. The attractors of intracellular networks have been found to be identifiable with different cell fates, cell behaviors, and stable patterns of cell activity \cite{MiskovTCell,ArabidopsisRoot,SaezRodriguezCancer,SocolarCellCycle,TLGLPNAS,SorgerReview,PhysBioReview,Huang1,HuangCancerAttrs}. In general, the task of finding Boolean network attractors is limited by combinatorial complexity; the size of the state space grows exponentially with the number of nodes $N$. To address this, we recently proposed an alternative approach to find the attractors of a Boolean network which allowed us to identify the attractors of networks for which a full search of the state space is not feasible \cite{ReductionChaos}. This attractor-finding method is based on identifying certain function-dependent network components, referred to as \textit{stable motifs}, that must stabilize in a fixed state. A stable motif is defined as a set of nodes and their corresponding states which are such that the nodes form a minimal strongly connected component (e.g. a feedback loop) and their states form a partial fixed point of the Boolean model. (A partial fixed point is a subset of nodes and a respective state for each of these nodes such that updating any node in the subset leaves its state unchanged, regardless of the state of the nodes outside the subset.) It is noteworthy that stable motifs are preserved for other updating schemes because of their dynamical property of being partial fixed points. For more details on the attractor-finding method and the identification of the stable motifs see Text \hyperref[sec:S1]{S1} and ref. \cite{ReductionChaos}; for a more formal and mathematical discussion see Text \hyperref[sec:S2]{S2} section \hyperref[sec:S2A]{A} or Appendix A of ref. \cite{ReductionChaos}.

Once a network's stable motifs and their corresponding fixed states are identified, a network reduction technique \cite{AssiehJTB,DecimationProcess,ReductionNadil,ReductionVeliz} is used for each stable motif by tracing the downstream effect of the stable motif on the rest of the network (see Text \hyperref[sec:S1]{S1}). Repeating this procedure iteratively for each separate stable motif until no new stable motifs are found yields the attractors of the logical model. Formally, the result is a set of network states called quasi-attractors, which capture steady states exactly and are a compressed representation of complex attractors \cite{ReductionChaos}. The network control method we propose here builds on the concept of stable motifs and its relation to (quasi-)attractors \cite{ReductionChaos} and takes it much further by connecting stable motifs with a way to identify targets whose manipulation (upregulation or downregulation) ensures the convergence of the system to an attractor of interest. The use of quasi-attractors in our method does not compromise its general applicability, but it does require that certain networks with special types of complex attractors are treated with care when our method is applied. None of the networks we discuss in this work nor any intracellular network models we are aware of fall in this category; for more details see Text \hyperref[sec:S1]{S1}, Text \hyperref[sec:S2]{S2}, and ref \cite{ReductionChaos}.)

As an illustration, consider the logical network shown in Fig. \ref{fig:NetworkExample}(a). This logical network has four stable motifs (Fig. \ref{fig:NetworkExample}(b)): (i) $\left\{\right.$A=1, B=1$\left.\right\}$, (ii) $\left\{\right.$A=0$\left.\right\}$, (iii) $\left\{\right.$E=1$\left.\right\}$, and (iv) $\left\{\right.$C=1, D=1, E=0$\left.\right\}$. Network reduction for each of these stable motif yields four reduced networks, each of which has its own stable motifs, all of which are shown in Fig. \hyperref[fig:FigS1]{S1}. For example, the reduced logical network obtained from the first stable motif consists of two nodes (D and E) and has two stable motifs: $\left\{\right.$E=1$\left.\right\}$ and $\left\{\right.$E=0$\left.\right\}$. The stable motifs of the remaining three reduced logical networks are, respectively: $\left\{\right.$E=1$\left.\right\}$ and $\left\{\right.$D=1$\left.\right\}$; $\left\{\right.$A=1, B=1$\left.\right\}$ and $\left\{\right.$A=0$\left.\right\}$; $\left\{\right.$A=1$\left.\right\}$ and $\left\{\right.$A=0$\left.\right\}$. Repeating the same network reduction procedure with each of the new stable motifs leads to either a new reduced network or one of four attractors ($\mathcal{A}_i, i=1, \ldots, 4$). The stable motifs obtained from the original network and from each reduced network, and the attractors they lead to are shown in Fig. \ref{fig:ReductionMethod}. This diagram is a compressed representation of the successive steps of the attractor finding process, which include the original network, the stable motifs of the original network, the reduced networks obtained for each stable motif, the stable motifs of these reduced networks, and so on (see Fig. \hyperref[fig:FigS1]{S1}). We refer to such a diagram as a \textit{stable motif succession diagram}, and we note that it is closely analogous to a cell fate decision diagram. We propose to use this stable motif succession diagram to guide the system to an attractor of interest.

\section*{Results} \label{sec:2}

\subsection*{Stable motif control implies network control} \label{sec:2.01}

The stable motifs' states are partial fixed points of the logical model, and as such, they act as ``points of no return" in the dynamics. Normally, the sequence of stable motifs is chosen autonomously by the system based on the initial conditions and timing. We propose to use our knowledge of the sequence of stable motifs to guide the system to an attractor of interest. We refer to this network control method as \textit{stable motif control}.

The basis of the stable motif control approach is that a sequence of motifs from a stable motif succession diagram like Fig. \ref{fig:ReductionMethod} uniquely determines an attractor, so controlling each motif in the sequence must prod the system towards this attractor. We give the proof of this statement in Lemma 4 and Proposition 6 of Text \hyperref[sec:S2]{S2} section \hyperref[sec:S2B]{B}. The number of nodes that need to be controlled can be minimized by removing motifs that do not need to be controlled and by finding a subset of nodes in a motif which can fix the whole motif's state. A step by step description of the stable motif control algorithm is given in \hyperref[Methods]{Methods}. For more details on the motif-removal step involved in minimizing the number of control nodes, see Text \hyperref[sec:S1]{S1}; for a justification of the steps involved in minimizing the number of control nodes, see Text \hyperref[sec:S2]{S2}. Text \hyperref[sec:S3]{S3} presents a discussion of the complexity of our methods and mitigation techniques for the most time consuming parts of our methods.

As an example, consider the network in Fig. \ref{fig:NetworkExample}(a) and choose $\mathcal{A}_2$ in Fig. \ref{fig:ReductionMethod} as our target attractor. There are two sequences of stable motifs that lead to $\mathcal{A}_2$: $\left(\right.\left\{\right.$C=1, D=1, E=0$\left.\right\},\left\{\right.$A=1$\left.\right\}\left.\right)$ and $\left(\right.\left\{\right.$A=1, B=1$\left.\right\},\left\{\right.$E=0$\left.\right\}\left.\right)$. For motif $\left\{\right.$C=1, D=1, E=0$\left.\right\}$ in the first sequence, fixing E=0 is enough to fix the whole motif's state; for motif $\left\{\right.$A=1$\left.\right\}$ in the same sequence there is only one node, so the only choice is to fix A=1. The control set obtained from the first sequence is then $\left\{\right.$E=0, A=1$\left.\right\}$. For the second sequence, a similar reasoning leads to the same control set, $\left\{\right.$E=0, A=1$\left.\right\}$ (E=0 from $\left\{\right.$E=0$\left.\right\}$, and A=1 from $\left\{\right.$A=1, B=1$\left.\right\}$). The result is a single set of network control interventions for attractor $\mathcal{A}_2$, $C_{\mathcal{A}_2}=\left\{\right.\left\{\right.$A=1, E=0$\left.\right\}\left.\right\}$. For a step by step description of the stable motif control algorithm applied to this example see Text \hyperref[sec:S1]{S1}.

Using our approach with each of the remaining attractors we obtain the following network control interventions: $C_{\mathcal{A}_1}=\left\{\right.\left\{\right.$A=1, E=1$\left.\right\}\left.\right\}$, $C_{\mathcal{A}_2}=\left\{\right.\left\{\right.$A=1, E=0$\left.\right\}\left.\right\}$, $C_{\mathcal{A}_3}=\left\{\right.\left\{\right.$A=0, E=1$\left.\right\}\left.\right\}$, and $C_{\mathcal{A}_4}=\left\{\right.\left\{\right.$A=0, E=0$\left.\right\}\left.\right\}$. Inspecting these network control interventions we conclude that controlling nodes A and E is enough to guide the system to each of the four possible attractors, with the exact combination being given by the $C_{\mathcal{A}_i}$'s.

In order to gauge the potential improvement in the control set's size brought about by our method, we compare our network control set with the feedback vertex set, the subset of nodes whose removal leaves the network without directed cycles. This set was demonstrated to be an effective control target and set an upper limit in the size of the control set in references \cite{FVS1,FVS2}. Because removing the feedback vertex set from the network must destroy all cycles, including self-loops, there are two possible minimal feedback vertex sets, $\left\{\right.$A, B, D, E$\left.\right\}$ and $\left\{\right.$A, C, D, E$\left.\right\}$. The number of nodes that need to be controlled in our method is half of the size of the feedback vertex set, a substantial improvement. It should be noted that our method does not guarantee that the resulting control sets are small nor that the control sets are the smallest possible, though our case studies suggest that the resulting control sets tend to be relatively small (between one and five nodes out of more than fifty, see Tables \ref{tab:TableTLGL} and \ref{tab:TableTh}, and ref \cite{EMTModel}).

\subsection*{Blocking stable motifs may obstruct specific attractors} \label{sec:2.02}

In many situations the main interest is to prevent the system from reaching an unwanted state (e.g. the proliferative cell state encountered in tumors). Based on the motif-sequence point of view provided by the stable motif succession diagram (Fig. \ref{fig:ReductionMethod}), we hypothesize that blocking the stable motifs that lead to an attractor will either prevent or make it less likely for the system to reach this attractor. We refer to this network control method as \textit{stable motif blocking}. The algorithm for the method is given in \hyperref[Methods]{Methods}.

The interventions obtained from this method are negations of node states of the target attractor, and as such, have the property of eliminating the intended attractor. However, new attractors can arise that are similar to the destroyed attractor. In biological situations (like in our test cases) one commonly has certain molecular markers of cell fate which specify the attractor to a large degree but not at the level of every node. Thus the final state obtained after stable motif blocking may still be consistent with the biological specification of the undesired attractor, making the intervention unsuccessful. We also adopt a stricter definition for a successful intervention: if a long-term but not permanent intervention (i.e. a transient intervention) reduces the number of network states or trajectories that lead to the unwanted attractor, then the intervention is considered to be \textit{long-term successful}. The best-case scenario would be that the manipulated network has only the desired attractors of the original network (i.e., any but the unwanted attractors), in which case the network will stay in these attractors even if the intervention is stopped.

Consider, for example, the network in Fig. \ref{fig:NetworkExample}(a) and the attractor $\mathcal{A}_3$ in Fig. \ref{fig:ReductionMethod}. From the stable motif succession diagram (Fig. \ref{fig:ReductionMethod}), the stable motifs involved in the sequences that lead to $\mathcal{A}_3$ are $\{$A=0$\}$, $\{$D=1$\}$, and $\{$E=1$\}$. Our approach proposes blocking these motifs to obstruct the system from reaching $\mathcal{A}_3$, that is, it provides $\mathcal{B}_{\mathcal{A}_3}=\left\{\{\right.$A=1$\}, \{$E=0$\}, \{$D=0$\}\left.\right\}$ or a combination of these node states as intervention candidates.

To verify the effectiveness of the interventions, we analyze the dynamics of the manipulated network with each individual intervention. The first intervention (A=1) causes the system to have $\mathcal{A}_1$ and $\mathcal{A}_2$ as its only attractors, and thus, the network is driven towards these attractors and away from the unwanted attractor $\mathcal{A}_3$. Furthermore, the network stays in those attractors even after the intervention is stopped, as they are also attractors of the original network, so the intervention is long-term successful. Similarly, the second intervention (E=0) causes the system to have $\mathcal{A}_2$ and $\mathcal{A}_4$ as its sole attractors, so it is also a long-term successful intervention. The third intervention (D=0) only leaves attractor $\mathcal{A}_1$ intact, and also gives rise to two new attractors. To evaluate if this intervention is long-term successful we compare the probabilities that an arbitrary initial condition ends in $\mathcal{A}_3$ with and without the intervention. For the intervened case, we set D=0 for a long time, then stop the intervention and wait for the network to reach an attractor. We find that the intervention makes it more likely for an arbitrary initial condition to reach $\mathcal{A}_3$, so this intervention is not long-term successful.

\begin{figure*}[h]
\centerline{\includegraphics[width=0.75\textwidth]{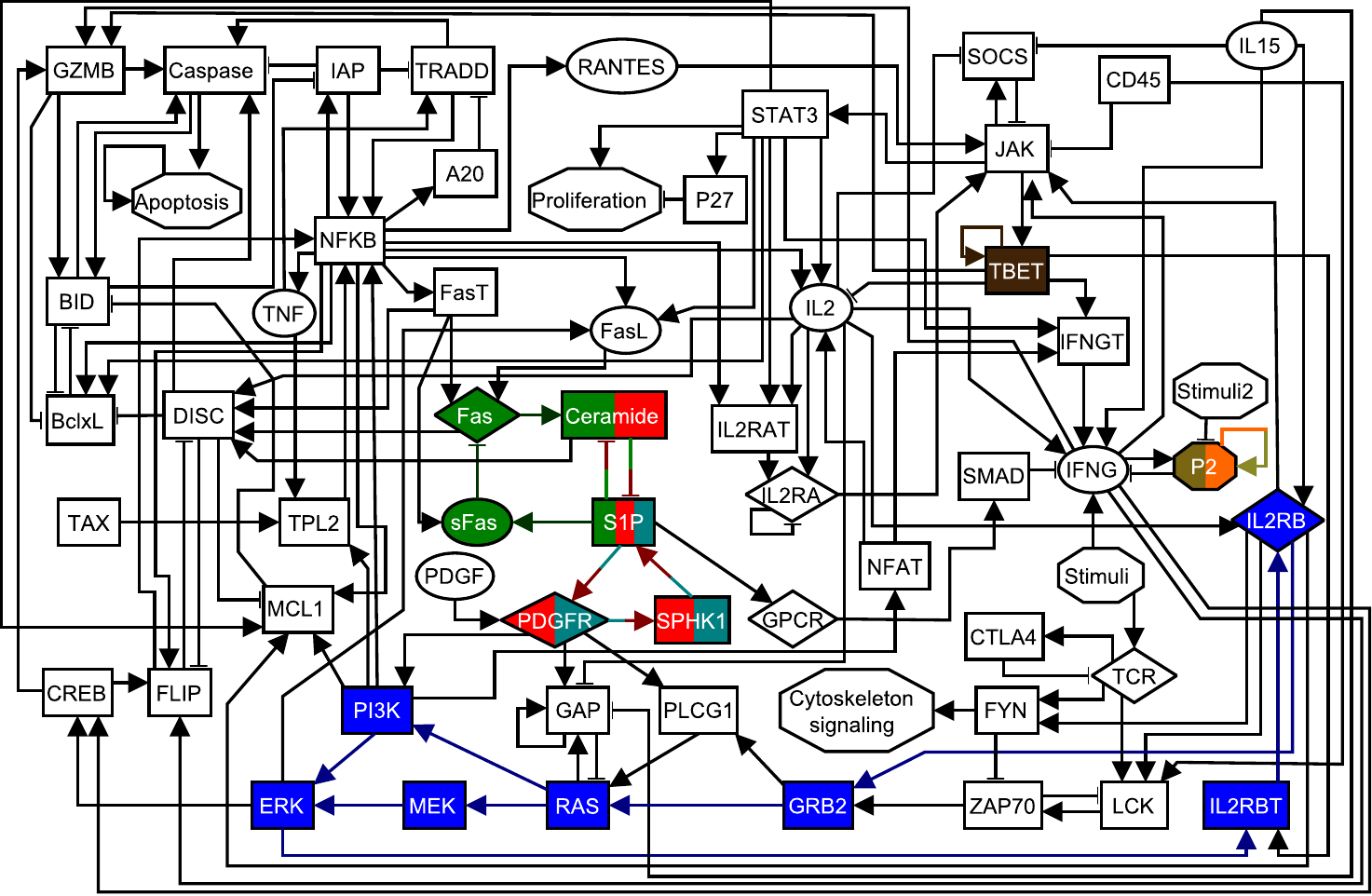}}
\caption{The T-LGL leukemia survival signaling network. The shape of the nodes indicates the cellular location or the type of nodes: rectangles indicate intracellular components, ellipses indicate extracellular components, diamonds indicate receptors, and hexagons represent conceptual nodes (Stimuli, Stimuli2, P2, Cytoskeleton signaling, Proliferation, and Apoptosis). Node colors are used to denote the different stable motifs of the network in the presence of the external signals Stimuli and IL15. Nodes and edges with multiple colors are part of several stable motifs. An arrowhead or a short perpendicular bar at the end of an edge indicates activation or inhibition, respectively. This figure and its caption are adapted from \cite{AssiehPCB}.}
\label{fig:TLGLnetwork}
\end{figure*}

\subsection*{Verification of the method's effectiveness in test cases} \label{sec:2.1}

The network control framework we propose is applicable to any cell fate reprogramming process for which a logical dynamical model can be constructed. This is a broad and increasing domain of application: refs. \cite{MiskovTCell,ArabidopsisRoot,SaezRodriguezCancer,SocolarCellCycle,TLGLPNAS} are examples of recent logical models that had experimentally validated predictions, while other examples can be found in the review articles \cite{SorgerReview,PhysBioReview}.

To demonstrate the potential of our framework, we choose two types of cell fate reprogramming processes: disease therapeutics and cell differentiation. More specifically, we use our network control framework to predict network control interventions on previously developed logical dynamic models for a leukemia signaling network and for the network controlling the differentiation of helper T cells. We confirm the effectiveness of the predicted stable motif control interventions using dynamic simulations, an independent verification  of the result we prove in Text \hyperref[sec:S2]{S2}. For the case of stable motif blocking interventions, whose effectiveness is not guaranteed, we use dynamic simulations to test the effectiveness of the predicted interventions.

\subsubsection*{T Cell Large Granular Lymphocyte Leukemia Network} \label{sec:2.1a}

Cytotoxic T cells are a central part of the immune system's response to infection. These T cells detect antigens in infected cells and, in response, induce the self-destruction of the infected cells. After fighting infection normal cytotoxic T cells undergo activation-induced cell death (apoptosis), but in T-cell large granular lymphocyte (T-LGL) leukemia  cytotoxic T cells avoid cell death and survive, which eventually leads to diseases such as autoimmune disorders.

A Boolean network model of cytotoxic T cell signaling that reproduces the known experimental results of these T cells in the context of T-LGL leukemia was previously constructed by Zhang et al. \cite{TLGLPNAS}. This network model consists of 60 nodes and 142 regulatory edges, with the nodes representing genes, proteins, receptors, small molecules, external signals (e.g. Stimuli), or biological functions (e.g. Apoptosis). The T-LGL network is shown in Fig. \ref{fig:TLGLnetwork} and its logical functions are reproduced in Text \hyperref[sec:S4]{S4}. Previous work by Zhang et al. \cite{TLGLPNAS} and Saadatpour et al. \cite{AssiehPCB} has shown that in the sustained presence of the external signals IL15, PDGF, and Stimuli (antigen presentation) the system has two attractors: one that recapitulates the survival phenotype and node deregulations seen in T-LGL leukemia, and a second one that corresponds to self-programmed cell death (apoptosis) (see Text \hyperref[sec:S4]{S4} for more details about attractor specification).

\begin{figure*}[!t]
\centerline{\includegraphics[width=0.8\textwidth]{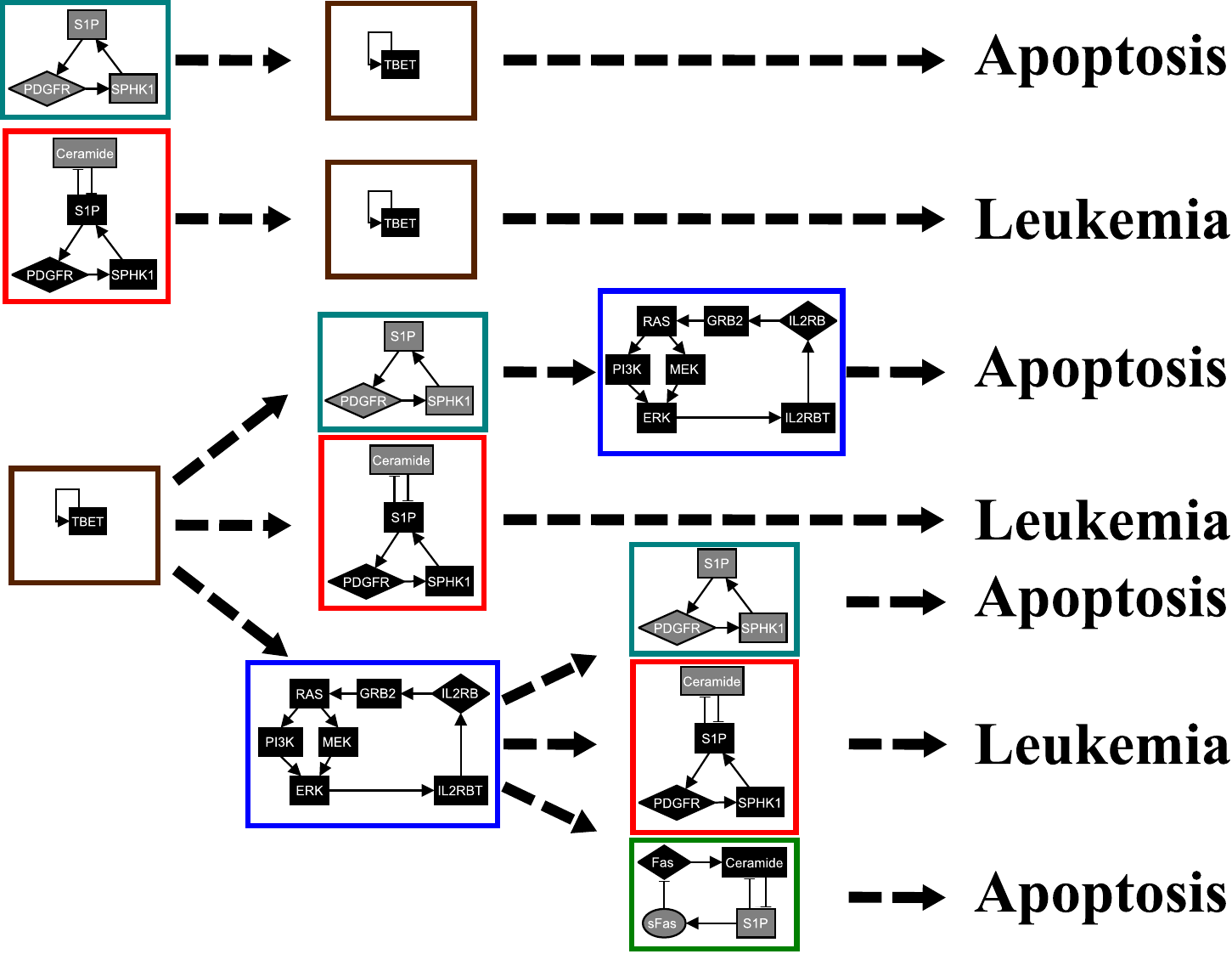}}
\caption{Stable motif succession diagram for the T-LGL leukemia network. The color of the nodes denotes their respective node states in the stable motifs: gray for 0 and black for 1. The colored rectangle surrounding each stable motif corresponds to the respective color of the motif in Fig. \ref{fig:TLGLnetwork}. There are two possible attractors for the system: the normal state of self-programmed cell death (apotosis) and the diseased state (T-LGL leukemia). The attractors obtained for each possible sequence of stable motifs are shown in the rightmost part of the figure.}
\label{fig:ReductionTLGL}
\end{figure*}

We first use our attractor-finding method on the T-LGL leukemia network in the presence of the external signals Stimuli and IL15 to obtain the stable motifs and the succession diagram. The result is 7 different stable motifs, each of which is shown in Fig. \ref{fig:TLGLnetwork} with a different node/edge color (nodes and edges with multiple colors are part of several stable motifs). The stable motif succession diagram for the T-LGL network is shown in Fig. \ref{fig:ReductionTLGL}. For simplicity we do not include the motifs associated with the node P2 in the succession diagram, as these motifs require the other stable motifs to influence the resulting attractor in the succession diagram.

The succession diagram in Fig. \ref{fig:ReductionTLGL} suggests a simple picture for the cell fate determination process: the activation of any of the three S1P-related motifs is enough to drive the system to either apoptosis (either the teal or the green stable motif in Figs. \ref{fig:TLGLnetwork} and \ref{fig:ReductionTLGL}) or T-LGL leukemia (the red stable motif in Figs. \ref{fig:TLGLnetwork} and \ref{fig:ReductionTLGL}). This result agrees with previous studies of T-LGL leukemia, in which it was found that blocking S1P signaling induced apoptosis in leukemic T-LGL cells \cite{TLGLPNAS,PDGFRTLGL}, a result reproduced by the network model when the state of S1P was set to OFF \cite{ReductionChaos,AssiehPCB}.

Next, we use the stable motif diagram in Fig. \ref{fig:ReductionTLGL} and our two control strategies to find intervention targets for the T-LGL leukemia network. The obtained intervention targets for each control strategy are shown in Table \ref{tab:TableTLGL}. Note that some intervention targets may be present in both control strategies (e.g. \{S1P=OFF\} is a target both for apoptosis control and T-LGL attractor blocking). For the case of stable motif blocking one may have the same intervention for blocking two different attractors (e.g. \{TBET=OFF\}), which means that this intervention could block either attractor.

\begin{table}[!t]
    \caption{Intervention targets for each control strategy in the T-LGL leukemia network model}
    \label{tab:TableTLGL}
    \begin{tabular}{@{\vrule height 7pt depth4pt  width0pt}c}
    \hline
    T-LGL leukemia stable motif control interventions ($C_{TLGL}$)\\
    \hline
    \{S1P=ON\}, \{Ceramide=OFF, SPHK1=ON\}, \\
    \{Ceramide=OFF,PDGFR=ON\} \\
    \hline
    Apoptosis stable motif control interventions ($C_{Apoptosis}$) \\
    \hline
    \{S1P=OFF\}, \{PDGFR=OFF\}, \{SPHK1=OFF\},\\
    \{TBET=ON, Ceramide=ON, RAS=ON\} \\
    \{TBET=ON, Ceramide=ON, GRB2=ON\}, \\
    \{TBET=ON, Ceramide=ON, IL2RB=ON\}, \\
    \{TBET=ON, Ceramide=ON, IL2RBT=ON\}, \\
    \{TBET=ON, Ceramide=ON, ERK=ON\}, \\
    \{TBET=ON, Ceramide=ON, MEK=ON, PI3K=ON\} \\
    \hline
    T-LGL leukemia stable motif blocking interventions ($B_{TLGL}$) \\
    \hline
    \{S1P=OFF\}, \{PDGFR=OFF\},\{SPHK1=OFF\}, \{Ceramide=ON\},\\
    \{TBET=OFF\},\{PI3K=OFF\},\{RAS=OFF\}, \{GRB2=OFF\},\\
    \{MEK=OFF\},\{ERK=OFF\}, \{IL2RBT=OFF\},\{IL2RB=OFF\}\\
    \hline
    Apoptosis stable motif blocking interventions ($B_{Apoptosis}$) \\
    \hline
    \{S1P=ON\}, \{PDGFR=ON\},\{SPHK1=ON\}, \{Ceramide=OFF\},\\
    \{sFas=ON\}, \{Fas=OFF\}, \{TBET=OFF\}, \{PI3K=OFF\}, \\
    \{RAS=OFF\}, \{GRB2=OFF\}, \{MEK=OFF\},\{ERK=OFF\}, \\
    \{IL2RBT=OFF\}, \{IL2RB=OFF\}  \\ \hline
    \end{tabular}
\end{table}

To validate an intervention target, we compare the probabilities that an arbitrary initial condition ends in the target attractor with and without the intervention (see \hyperref[Methods]{Methods}). The results of the intervention target validation are summarized in Table S1. For all the stable motif control interventions we obtain 100\% effectiveness in reaching the desired state, both for the case in which the intervention is permanent and for the case in which it is not. This means that all stable motif control interventions are \textit{long-term successful}, in agreement with our formal results in Text \hyperref[sec:S2]{S2}. For example, when fixing S1P=OFF the apoptosis attractor is reached for all the initial conditions, indicating that the T-LGL attractor is unreachable. For the case of the stable motif blocking interventions we find that each of them but one (GRB2=OFF) is successful in blocking its target attractor or one of its target attractors, though not always with 100\% effectiveness. For example, for TBET=OFF the apoptosis attractor is reached from 10\% of the initial conditions, which is a substantial reduction from the baseline of 62\% in the case of no intervention, indicating that this interventions is effective as an apoptosis blocking strategy. We also find that most of the stable motif blocking interventions are effective when the intervention is permanent, but only a few of them are effective when the intervention is temporary.

Single interventions are the most commonly used therapeutic strategies for treating diseases. Thus, we evaluate the success of each single intervention from control sets with more than one node (see Table S1). We find that one of the 12 single node interventions, Ceramide=ON, is 100\% effective and long-term successful. Of the remaining 11 single node interventions only a few are successful (Ceramide=OFF, SPHK1=ON, and PDGFR=ON) and/or long-term successful (SPHK1=ON and PDGFR=ON) but none of them are 100\% effective. This result illustrates the benefit of combinatorial interventions over single interventions.

\subsubsection*{Helper T Cell Differentiation Network} \label{sec:2.1b}

Helper T cells are crucial in the regulation of the immune response in mammals. These T cells release specific cytokines that alter how the immune system responds to external agents, for example, by recruiting specific immune system cells to fight infection, promoting antibody production, or inhibiting the activation and proliferation of other cells. Various subtypes of helper T cells are known, such as Th1, Th2, Th17 and Treg, which are distinguished by a differential expression of specific transcription factors and cytokines.

\begin{figure*}[t]
\centerline{\includegraphics[width=0.7\textwidth]{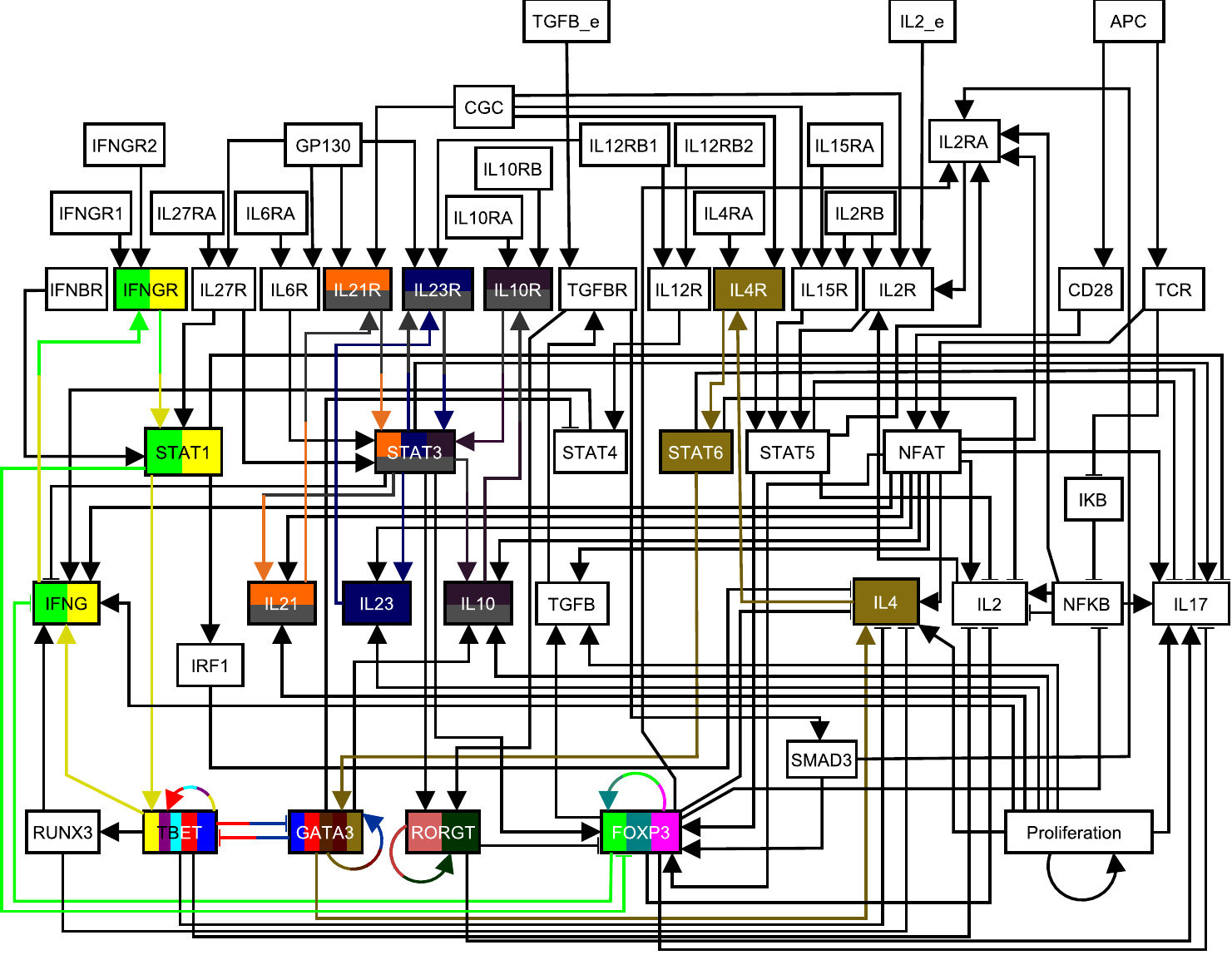}}
\caption{The helper T cell differentiation network. The nodes that encode the environmental conditions (APC=ON, TGFB\_e=ON, IL2\_e=ON) are located in the upper part of the network diagram. Node colors are used to denote the different stable motifs of the network in the used environmental conditions. Nodes and edges with multiple colors are part of several stable motifs. An arrowhead or a short perpendicular bar at the end of an edge indicates activation or inhibition, respectively. This figure is adapted from \cite{ThCellDifferentiation}.}
\label{fig:Thnetwork}
\end{figure*}

A logical network model of the regulatory and signaling pathways controlling helper T cell activation and differentiation was constructed by Naldi et al. \cite{ThCellDifferentiation}. This network model has several attractors, which correspond to the known canonical helper T cell subtypes, and also to some hybrid cell types (see \cite{ThCellDifferentiation} and Text \hyperref[sec:S5]{S5}). The reachability of each attractor depends on the presence of several external environmental signals (either cytokines or antigen), which are represented as input nodes in the network. For our study we use one of the environmental conditions studied by Naldi et al. (TGFB\_e=ON, IL2\_e=ON, and APC=ON) \cite{ThCellDifferentiation} because it allows us to explore control targets for all T cell subtypes. The helper T cell differentiation network under the selected environmental conditions consists of 55 nodes and 121 edges and is shown in Fig. \ref{fig:Thnetwork}. Its corresponding logical functions are reproduced in Text \hyperref[sec:S5]{S5}.

\begin{figure*}[!h]
\centerline{\includegraphics[width=0.6\textwidth]{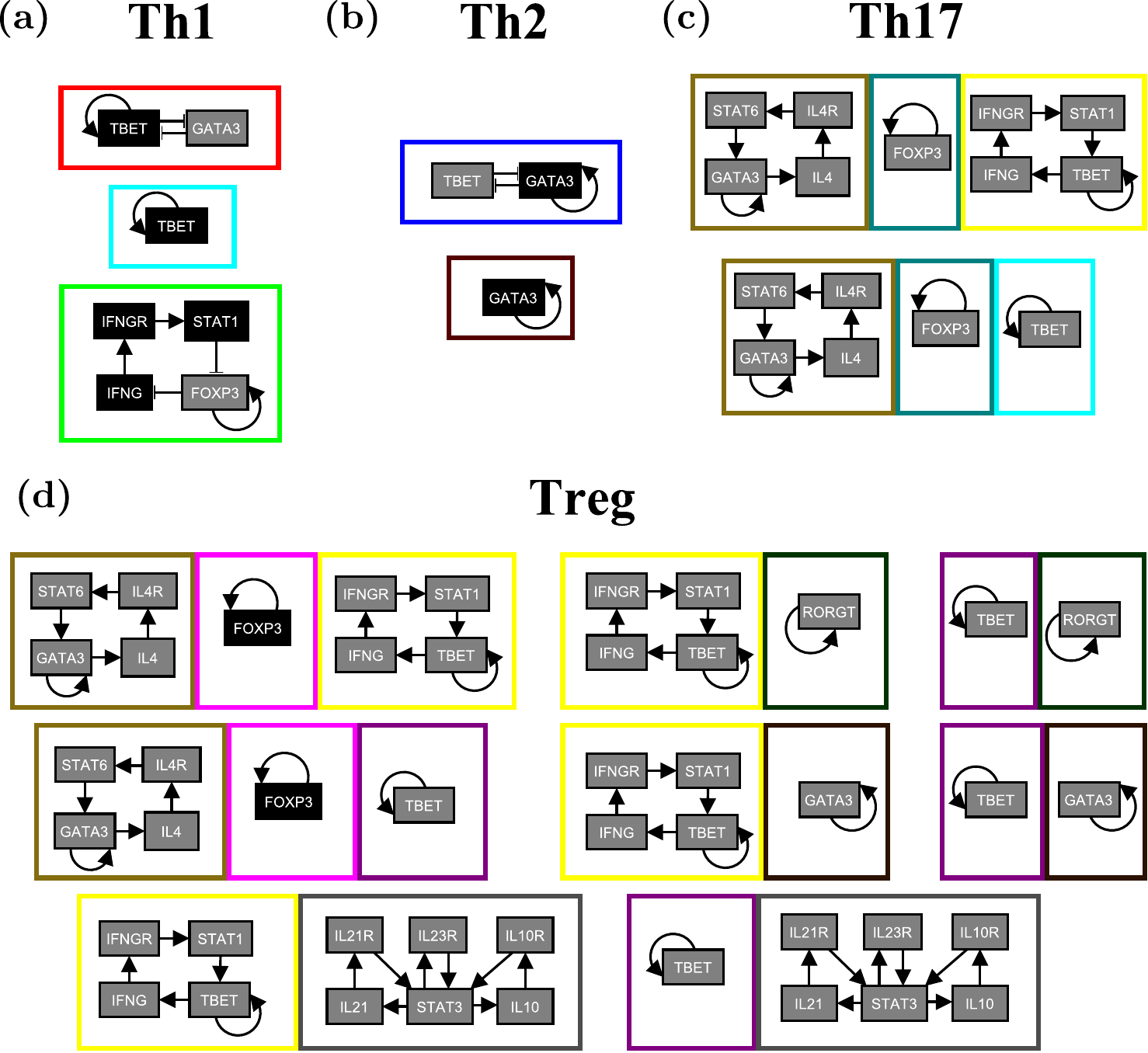}}
\caption{Minimal subsets of stable motifs associated to each helper T cell subtype. Each stable motif is enclosed by a colored rectangle, and motifs which are part of the same minimal subset have their enclosing rectangles touching each other. The node colors denotes their respective node states in the stable motifs: gray for 0 and black for 1. The color of the rectangle enclosing each stable motif corresponds to the respective color of that motif in Fig. \ref{fig:Thnetwork}.}
\label{fig:ReductionTh}
\end{figure*}

We obtain 17 stable motifs, each of which is shown in Fig. \ref{fig:Thnetwork} with a different node/edge color, and a stable motif succession diagram composed of 697 sequences. Despite the large size of the succession diagram, a closer look at it gives a simple interpretation: the stable motifs associated with each attractor regulate the characteristic transcription factor of each helper T cell subtype (see Text \hyperref[sec:S5]{S5}). We use the stable motif succession diagram and our stable motif control and stable motif blocking strategies to find intervention targets for each helper T cell subtype (see Table \ref{tab:TableTh}).

To validate the proposed intervention targets we use the same procedure as in the T-LGL leukemia network case (see \hyperref[Methods]{Methods}). We also look at the effect of single node interventions for control sets with more than one node. The results of the intervention targets for the stable motif control, stable motif blocking strategies, and single node interventions are summarized in Table S2. We find that (i) there is a 100\% effectiveness in reaching the desired state for all the stable motif control interventions, (ii) most of the stable motif blocking interventions are successful in blocking their target attractor or one of their target attractors, though not always with 100\% effectiveness, and (iii) some single interventions are successful, but none of them are 100\% effective.

\subsection*{The control targets transcend the logical modeling framework}

The network control approach we propose is formulated in a Boolean framework, which brings up the question of whether the control targets identified are dependent on the logical modeling scheme. To address this, we translate the studied Boolean network models into ordinary differential equation (ODE) models using the method described by Wittmann et al. \cite{BooleantoODE}. In the ODE models the node state variables $\widetilde\sigma_i$ can take values in the range $\left[0,1\right]$; the differential equations of the translated model have the form $\dot{\widetilde\sigma}_i=(1/\tau_i)[\widetilde{f}_i(\widetilde\sigma_{i_1}, \dots, \widetilde\sigma_{i_{k_i}})-\widetilde\sigma_i]$, where $\widetilde{f}_i$ is a smooth Hill-type function parameterized by Hill coefficients and threshold parameters, and $\tau_i$ is a time-scale parameter. The function $\widetilde{f}_i$ is such that it matches the Boolean function $f_i$ whenever its inputs $\widetilde\sigma_{i_1}, \dots, \widetilde\sigma_{i_{k_i}}$ are either 0 or 1. Thus, the fixed point attractors of the Boolean model are preserved in the ODE model.

We test the effectiveness of the stable motif control interventions in the translated ODE models by comparing the probability for an uniformly chosen initial condition to reach the target attractor with and without the intervention (see Text \hyperref[sec:S6]{S6}). We find that the stable motif control interventions are still 100\% effective or very close for both permanent and transient interventions (Tables S3 and S4). We also find that the effectiveness of the interventions is mostly unchanged by varying the Hill coefficients (Table S5), varying the the time-scale parameters $\tau_i$ and thresholds (Table S6), or fixing the intervened node variables close to but not exactly at the intervention-prescribed values (Table S7). We finally test single interventions and find that they still underperform combinatorial interventions (Tables S3 and S4).

To further validate the successful control targets we identified, we searched the literature for experimental support for these targets. We find that several of the single interventions predicted to be successful in inducing apoptosis of leukemic T cells or in inducing specific T cell types were found to be successful experimentally. The control targets for which experimental support was found, the attractors they lead to, and the references are shown in Table \ref{tab:ExperimentalSupport}. Collectively, these results strongly suggest that the control targets identified by our approach transcend the logical framework.

\section*{Discussion} \label{sec:3}

Identifying control targets for intracellular networks is of crucial importance for practical applications such as disease treatment and stem cell reprogramming. Despite recent advances in network controllability approaches, most of them rely solely on the topology \cite{BarabasiControllability,BarabasiObservability,NodalDynamics,FVS1,FVS2} or the dynamics \cite{MotterControl,Akutsu,Cheng,Tamura} of the network. Thus, potentially important effects that depend on the interplay between structure (topology) and function (dynamics), such as combinatorial interactions, are not considered. In this work we proposed a network control approach that combines the structural and functional information of a discrete (logical) dynamic network model to identify control targets. The method builds on the concept of stable motif and its relation to finding attractors \cite{ReductionChaos}, and takes it much further by connecting stable motifs with a way to identify targets whose manipulation (upregulation or downregulation) ensures the convergence of the system to an attractor of interest. We illustrated our method's potential to find intervention targets for cancer treatment and cell differentiation by applying it to network models of T-LGL leukemia and helper T cell differentiation.

The control interventions identified by our method have many desirable characteristics. For example, stable motif control interventions are guaranteed to drive an initial state to the target attractor state with 100\% effectiveness, regardless of the initial state, a general result which we prove in Text \hyperref[sec:S2]{S2} and corroborate in our test cases (see Tables S1 and S2). They are also long-term successful, meaning that the intervention only needs to be applied transiently for the network to reach and stay in the desired state, a general result which we also verify in our test cases (see Tables S1 and S2). We attribute these properties to the use of the natural (autonomous) dynamics of the network to control its dynamics.

Another noteworthy characteristic of our stable motif control method is the combinatorial nature of the multi-target interventions. As shown in Tables S1 and S2, only one single-node intervention (namely, Ceramide=ON in the T-LGL leukemia network) was able to match the 100\% effectiveness of the multi-target interventions. This agrees with recent clinical studies on the advantages of combinatorial over single target interventions \cite{CombTherapyClinical1,CombTherapyModel,CombinatorialTherapy}. Finally, the stable motif control interventions for our case studies target only a few nodes (between one and five out of more than fifty), which matches what is expected from stem cell reprogramming experiments \cite{StemCellsTakahashi,CellReprogReview1,CellReprogReview2,MullerSchuppertReply}.

The framework presented in this work is formulated and applied in the context of logical network modeling of cell fate reprogramming processes but its applicability is not restricted to it. Indeed, our control approach is applicable to any dynamic process that can be captured qualitatively by a Boolean dynamic network model such as ecological community dynamics \cite{PlantPolinator}, social dynamics \cite{SocialReview,VoterModels}, or disease spreading \cite{EpidemicVerispignani,EpidemicKlemm}. The validity of the control targets on the translated ODE models of our two case studies and the experimental support found for several of these targets demonstrates the broader, potentially model-independent reach of our method. Further work is needed to address exactly how to extend the concept of stable motif and our network control approach to continuous models; formalizing our framework to admit an arbitrary number of discrete states and other updating schemes may prove a valuable step in this direction.

Taken together, our results provide a novel framework for the control of the dynamics of intracellular networks that combines realistically obtainable structural and functional information of the network of interest. As such, we expect this framework to be significant to a variety of practical applications and to also provide a new avenue to better understand how the complex behaviors of cells in living organisms emerges from the underlying network of biochemical interactions.

\begin{table}[h]
\caption{Intervention targets for each control strategy in the helper T cell network.}
\label{tab:TableTh}
\begin{tabular}{@{\vrule height 7pt depth4pt  width0pt}c}
\hline
Th1 stable motif control interventions ($C_{Th1}$) \\
\hline
\{TBET=ON\}\\
\hline
Th2 stable motif control interventions ($C_{Th2}$) \\
\hline
\{GATA3=ON\} \\
\hline
Th17 stable motif control interventions ($C_{Th17}$) \\
\hline
\{GATA3=OFF, FOXP3=OFF, TBET=OFF, STAT3=ON\}, \{GATA3=OFF, FOXP3=OFF, TBET=OFF, IL10=ON\},\\
\{GATA3=OFF, FOXP3=OFF, TBET=OFF, IL10R=ON\}, \{GATA3=OFF, FOXP3=OFF, TBET=OFF, IL21=ON\},\\
\{GATA3=OFF, FOXP3=OFF, TBET=OFF, IL21R=ON\}, \\
\{GATA3=OFF, FOXP3=OFF, TBET=OFF, IL23R=ON, RORGT=ON\}\\
\hline
Treg stable motif control interventions ($C_{Treg}$) \\
\hline
\{GATA3=OFF, FOXP3=ON, TBET=OFF\}, \{GATA3=OFF, TBET=OFF, STAT3=OFF\},\\
\{GATA3=OFF, TBET=OFF, IL23R=OFF, IL10R=OFF, IL21R=OFF\},\\
\{GATA3=OFF, TBET=OFF, IL23R=OFF, IL10=OFF, IL21R=OFF\},\\
\{GATA3=OFF, TBET=OFF, IL23R=OFF, IL10R=OFF, IL21=OFF\},\\
\hline
Th1 stable motif blocking interventions ($B_{Th1}$) \\
\hline
\{GATA3=ON\}, \{TBET=OFF\}, \{IL4=ON\}, \{IL4R\_2=ON\}, \{STAT6=ON\}, \{STAT1=OFF\}, \{IFNG=OFF\}, \{IFNGR=OFF\},\\
\{IL23=OFF\}, \{IL10=ON,OFF\}, \{IL10R=ON,OFF\}, \{IL21=ON,OFF\}, \{IL21R=ON,OFF\}, \{STAT3=ON,OFF\}, \\
\{IL23R=ON,OFF\}, \{RORGT=ON,OFF\}, \{FOXP3=ON,OFF\}\\
\hline
Th2 stable motif blocking interventions ($B_{Th2}$) \\
\hline
\{GATA3=OFF\}, \{TBET=ON\}, \{STAT1=ON\}, \{IFNG=ON\}, \{IFNGR=ON\},\{IL23=OFF\}, \{IL23R=OFF\}, \{STAT3=OFF\},\\
\{IL10=OFF\}, \{IL10R=OFF\},\{RORGT=ON\}, \{FOXP3=ON,OFF\}\\
\hline
Th17 stable motif blocking interventions ($B_{Th17}$) \\
\hline
\{GATA3=ON\}, \{TBET=ON\}, \{IL4=ON\}, \{IL4R\_2=ON\}, \{STAT6=ON\}, \{STAT1=ON\},\{IFNG=ON\}, \{IFNGR=ON\},\\
\{STAT3=OFF\}, \{FOXP3=ON\}, \{RORGT=OFF\},\{IL21=OFF\}, \{IL21R=OFF\}, \{IL23=OFF\}, \{IL23R=OFF\}, \\
\{IL10=OFF\}, \{IL10R=OFF\}\\
\hline
Treg stable motif blocking interventions ($B_{Treg}$) \\
\hline
\{GATA3=ON\}, \{TBET=ON\}, \{IL4=ON\}, \{IL4R\_2=ON\}, \{STAT6=ON\}, \{STAT1=ON\},\{IFNG=ON\}, \{IFNGR=ON\},\\
\{STAT3=ON,OFF\},\{FOXP3=OFF\}, \{RORGT=ON,OFF\},\{IL21=ON,OFF\}, \{IL21R=ON,OFF\}, \{IL23=OFF\}, \\
\{IL23R=ON,OFF\}, \{IL10=ON,OFF\}, \{IL10R=ON,OFF\}\\
\hline
\end{tabular}
\end{table}

\begin{table}[!h]
\caption{Experimental support for successful control targets in Tables \ref{tab:TableTLGL} and \ref{tab:TableTh}.}
\label{tab:ExperimentalSupport}
\begin{tabular}{@{\vrule height 7pt depth4pt  width0pt}ccc}
\hline
Intervention & Target attractor & Reference \\
\hline
\multicolumn{3}{c}{T-LGL leukemia} \\
\hline
\{S1P=OFF\}  & Apoptosis &  \cite{PDGFRTLGL} \\
\{SPHK1=OFF\}  & Apoptosis &  \cite{TLGLPNAS} \\
\{PDGFR=OFF\}  & Apoptosis &  \cite{TLGLPNAS,PDGFPI3KTLGL}\\
\{Ceramide=ON\}  & Apoptosis &  \cite{CeramideTLGL}\\
\{RAS=OFF\}  & Apoptosis &  \cite{ERKTLGL}\\
\{MEK=OFF\}  & Apoptosis &  \cite{ERKTLGL}\\
\{ERK=OFF\}  & Apoptosis &  \cite{ERKTLGL}\\
\{PI3K=OFF\}  & Apoptosis &  \cite{PDGFPI3KTLGL,PI3KTLGL}\\
\hline
\multicolumn{3}{c}{Helper T cell differentiation} \\
\hline
\{TBET=ON\}  & Th1 & \cite{Th1Th2Review,Th1TBET}\\
\{GATA3=ON\}  & Th2 & \cite{Th1Th2Review,Th2GATA3} \\
\{IL21=ON\}  & Th17 & \cite{Th17IL21RIL23R} \\
\{IL21R=ON\}  & Th17 & \cite{Th17IL21RIL23R} \\
\{IL23R=ON\}  & Th17 & \cite{Th17IL21RIL23R} \\
\{FOXP3=ON\}  & Treg & \cite{TregFOXP3} \\
\hline
\end{tabular}
\end{table}

\section*{Methods} \label{Methods}
\setcounter{subsection}{0}
\subsection*{Computational methods}
The simulations of the logical model, the attractor-finding method, and the analysis of the stable motif succession diagrams were performed using a custom Java code, which is available per request to the interested reader. The generation of the ODE model from the logical model was done using the MATLAB implementation of the method of Wittman et al. \cite{BooleantoODE,Odefy}; the numerical integration of the ODE models was performed using MATLAB's ode45 function (see Text \hyperref[sec:S6]{S6} for more details). The networks in all figures were created using the yEd graph editor (http://www.yworks.com/).

\subsection*{General asynchronous updating scheme}\label{Methods:Asynchr}In the general asynchronous scheme, the state of the nodes is updated at discrete time steps starting from an initial condition at $t=0$. At every time step, one of the variables is chosen randomly (uniformly) and is updated using its respective function and the state of its regulators at the previous time step
\begin{equation} \label{eq:asynchronous1}
  \sigma_j(t+1)= f_j\left(\sigma_{j_1}(t), \sigma_{j_2}(t), \dots, \sigma_{j_{k_j}}(t)\right),
\end{equation}
while the rest of the variables retain their state. In this way, every possible update order is allowed, and thus, all relative timescales of the processes involved are sampled.

\subsection*{Stable motif control algorithm}\label{Methods:Control1}For an attractor of interest $\mathcal{A}$, the steps of the stable motif network control method are the following:
\begin{itemize}
\item[-] \textit{Step }\textit{1}: Identify the sequences of stable motifs that lead to $\mathcal{A}$. These can be obtained from the stable motif succession diagram (see Fig. \ref{fig:ReductionMethod}) by choosing the attractor of interest in the right-most part and selecting all of the attractor's predecessors in the succession diagram.
\item[-] \textit{Step }\textit{2}: Shorten each sequence $\mathcal{S}$ by identifying the minimum number of motifs in $\mathcal{S}$ required for reaching $\mathcal{A}$ and removing the remaining motifs from the sequence. This minimum number of motifs can be identified from the stable motif succession diagram (Fig. \ref{fig:ReductionMethod}); they are the motifs after which all consequent motif choices lead to the same attractor $\mathcal{A}$.
\item[-] \textit{Step 3}: For each stable motif's state $\mathcal{M}=\left(\sigma_{\mathcal{M}_1}, \sigma_{\mathcal{M}_2}, \ldots, \sigma_{\mathcal{M}_m}\right)$, find the subsets of stable motif's states $O=\left\{M_i\right\}, M_i \subseteq \mathcal{M}$ that, when fixed, are enough to force the state of every node in the motif into $\mathcal{M}$. At worst, there will only be one subset, which will equal the whole stable motif's state $\mathcal{M}$. If any of these subsets is fully contained in another subset, remove the larger of the subsets. In each stable motif sequence $\mathcal{S}=\left(\mathcal{M}_1, \ldots, \mathcal{M}_L\right)$, substitute every stable motif $\mathcal{M}_j$ with the subsets of the stable motif's states obtained, that is, $\mathcal{S}=\left(O_1, \ldots, O_L\right)$.
\item[-] \textit{Step 4}: For each sequence $\mathcal{S}=\left(O_1, \ldots, O_L\right)$ create a set of states $\mathcal{C}$ by choosing one of the subsets of stable motif's states $M_{k_j}$ in each $O_j$ and taking their union, that is, $\mathcal{C}=M_{k_1} \cup \cdots \cup M_{k_L}, M_{k_j} \in O_j$. The network control set for attractor $\mathcal{A}$ is the set of node states $C_\mathcal{A}=\left\{\mathcal{C}_i\right\}$ obtained from all possible combinations of subsets of stable motif's states $M_{k_j}$'s for every sequence $\mathcal{S}$. To avoid any redundancy, we additionally prune $C_\mathcal{A}$ of duplicates and remove each set of node states $\mathcal{C}_i$ which is a superset of any of the other sets of node states $\mathcal{C}_j$ (i.e. $\mathcal{C}_j \subset \mathcal{C}_i$).
\end{itemize}
For a pseudocode of each step of the stable motif control algorithm see Text \hyperref[sec:S7]{S7}.

\subsection*{Stable motif blocking algorithm}\label{Methods:Control2}Given an attractor $\mathcal{A}$ one is interested in obstructing, the steps to identify potential interventions are the following:
\begin{itemize}
\item[-] \textit{Step }\textit{1}: Identify the sequences of stable motifs that lead to $\mathcal{A}$. This step is the same as the first step in the stable motif control algorithm, and can be obtained from the stable motif succession diagram (Fig. \ref{fig:ReductionMethod}).
\item[-] \textit{Step }\textit{2}: Take each stable motif's state $\mathcal{M}_i$ in the sequences obtained in the previous step. Create a new set $\mathbf{M}_{\mathcal{A}}$ with all of these stable motif states, $\mathbf{M}_{\mathcal{A}}=\left\{\mathcal{M}_i\right\}$.
\item[-] \textit{Step }\textit{3}: Take each node state $\sigma_j \subset \mathcal{M}_i$ of the stable motif's states $\mathcal{M}_i$ in $\mathbf{M}_{\mathcal{A}}$. Create a new set $\mathcal{B}_{\mathcal{A}}$ with the negation of each node state, $\mathcal{B}_{\mathcal{A}}=\left\{\overline{\sigma}_j\right\}$. The node states in $\mathcal{B}_{\mathcal{A}}$ and any combination of them are identified as potential interventions to block attractor $\mathcal{A}$.
\end{itemize}
For a pseudocode of each step of the stable motif blocking algorithm see Text \hyperref[sec:S7]{S7}.

\subsection*{Intervention target validation}
To validate an intervention target, we fix the node states prescribed by the intervention, choose a random (uniformly chosen) initial condition, and evolve the system using the general asynchronous updating scheme for a sufficiently large number of time steps (50,000) so that the system reaches an attractor. We repeat this for a large number of initial conditions (100,000) and calculate the probability of reaching each attractor from an arbitrary (uniformly chosen) initial condition. We also look at the probability of reaching each attractor when the intervention is not permanent, that is, we fix the prescribed node states for a large number of time steps, then stop fixing these states and wait for the system to reach an attractor. For this case we use 100,000 uniformly chosen initial conditions and 50,000 time steps both before and after stopping the intervention. The number of initial conditions we use is enough to estimate the probabilities $p_{Attr}$ of reaching the attractor of interest with an error (standard deviation of the estimated probability $p_{Attr}$) of $3\cdot10^{-3}\left[p_{Attr} (1-p_{Attr})\right]^{1/2}$. Equivalently, if $p_{Attr}$ is expressed as a percentage (which we denote as  $\%p_{Attr}$ for clarity), the error in it is estimated as $3\cdot10^{-3}\left[\%p_{Attr}(100\%-\%p_{Attr})\right]^{1/2}\%$ (e.g. 0.03\% for a $\%p_{Attr}$ of 1\%, and 0.15\% for a $\%p_{Attr}$ of 50\%). The number of time steps we use is enough to show no changes in $p_{Attr}$ beyond what is expected from the standard deviation of the estimated probability $p_{Attr}$, and is also found to be enough for the initial conditions to reach the attractors when no interventions are applied.

\section*{Acknowledgments}
We would like thank Steven N Steinway for fruitful discussion and the three anonymous reviewers for their suggestions.

\clearpage

\setcounter{table}{0}
\renewcommand{\thepseudocode}{\arabic{pseudocode}}
\renewcommand{\figurename}{Figure}
\renewcommand{\tablename}{Table}
\renewcommand{\thetable}{S\arabic{table}}
\renewcommand{\thefigure}{S1.\arabic{figure}}
\setcounter{figure}{0}

\section*{Text S1. Details and examples of the attractor finding method and stable motif control algorithm} \label{sec:S1}
\setcounter{subsection}{0}
\subsection{Attractor-finding method} \label{sec:S1A}
The task of finding attractors for a Boolean network is limited by the exponential growth of the state space with the number of number of nodes $N$. As a consequence, a full search of the state space to find the attractors is viable only for small networks ($N \lesssim 20$). To overcome this problem when dealing with intracellular networks (or, more generally, with sparse networks), we recently proposed an alternative approach to find the attractors of a Boolean network model \cite{ReductionChaos}. This approach successfully found the attractors of a previously developed biological network model composed of 60 nodes \cite{TLGLPNAS,ReductionChaos}, and of an ensemble of random Boolean networks composed of up to hundreds of nodes \cite{ReductionChaos}. It has also been proven to find both fixed point and complex attractors. More formally, the result of the attractor-finding method are the so-called quasi-attractors, each of which has a corresponding system attractor (see Text \hyperref[sec:S2]{S2} section \hyperref[sec:S2A]{A} or ref. \cite{ReductionChaos} for details). A quasi-attractor is a set of network states in which each node state is either fixed (0 or 1) or is not specified, in which case it is expected to oscillate. The difference between an attractor and a quasi-attractor is that an attractor includes the nodes that oscillate and the precise network states they can take, while the quasi-attractor does not specify the precise network states that the oscillating nodes take. For a more detailed explanation and the step-by-step algorithm of the attractor-finding method see Text \hyperref[sec:S2]{S2} section \hyperref[sec:S2A]{A} or ref. \cite{ReductionChaos}.

\subsection{Stable motif identification} \label{sec:S1B}

Stable motifs are function-dependent network components (subnetworks) in a Boolean model that must stabilize in a fixed state. These network components and their respective fixed states are identified with a certain type of strongly connected component (or SCC, a subgraph in a directed network for which all node pairs are connected by paths in both directions) in an expanded representation of the Boolean network \cite{ReductionChaos,ESMRuiSheng}. The expanded network representation explicitly incorporates the combinatorial nature and the sign of the interactions. This is achieved by introducing complementary nodes for every node, which are used to indicate negative regulation in a Boolean function (NOT relationship), as well as introducing a composite node to denote a conditional dependence (AND relationship) among two or more inputs in a Boolean function. A detailed explanation of the expanded network representation can be found in Text \hyperref[sec:S2]{S2} section \hyperref[sec:S2A1]{A.1} and ref. \cite{ReductionChaos}.

As an example, consider node $C$ in the example network in Fig. \hyperref[fig:NetworkExample]{1}. The expanded network representation of $C$ and its complementary node $\overline{C}$ is shown in Fig. \hyperref[fig:FigS2]{S2}(a). The function $f_C=\left(A\hbox{ AND }B\right)\hbox{ OR }D$ contains an AND relationship between the state of node $A$ and the state of node $B$, so a composite node $AB$ is added when expanding the network. Node $A$ and $B$ are connected by directed edges to the composite node $AB$, and an edge from $AB$ to $C$ is also present. Since the state of node $D$ is OR-separated from the $\left(A\hbox{ AND }B\right)$ term, an edge from $D$ to $A$ is part of the expanded network. A complementary node $\overline{C}$ is also added in the expanded network, with an associated Boolean function $f_{\overline{C}}=\overline{f}_C= \left(\hbox{NOT } A\hbox{ AND NOT } D\right)\hbox{ OR }\left(\hbox{NOT } B\hbox{ AND NOT }D\right)$. The expanded network will contain the composite nodes $\overline{A}\overline{D}$ and $\overline{B}\overline{D}$, directed edges from $\overline{A}$ and $\overline{D}$ to $\overline{A}\overline{D}$, directed edges from $\overline{B}$ and $\overline{D}$ to $\overline{B}\overline{D}$, and directed edges from $\overline{A}\overline{D}$ and $\overline{B}\overline{D}$ to $\overline{C}$. As another example, the expanded network representation of $B$ and $\overline{B}$ is shown in Fig. \hyperref[fig:FigS2]{S2}(b).

In the expanded network representation, stable motifs correspond to minimal strongly connected components that satisfy two properties: (1) the strongly connected component does not contain both a node and its complementary node, and (2) if the strongly connected component contains a composite node, all of its input nodes must also be part of the strongly connected component. A more detailed explanation of the method for identifying stable motifs is given in Text \hyperref[sec:S2]{S2} and ref. \cite{ReductionChaos}. The main point is that a stable motif can be identified with a set of nodes that form a minimal strongly connected component, and that a stable motif's corresponding states are such that they form a partial fixed point of the Boolean model (for a Boolean model with node variables $\{\sigma_i\}$ and associated functions $\{f_i\}$, a partial fixed point is a set of node states $P=\{\sigma_{p_1}=s_{p_1}, \sigma_{p_2}=s_{p_2},\ldots, \sigma_{p_l}=s_{p_l}\}$ such that if $\Sigma_P$ is any network state in which $\sigma_{p_k}=s_{p_k} \forall p_k \in \left\{p_1, p_2, \ldots, p_l\right\}$, then $f_{p_j}(\Sigma_P)=s_{p_j}$.).

As an example, consider the logical network in Fig. \hyperref[fig:NetworkExample]{1}(a) in the main text, and its associated stable motifs in Fig. \hyperref[fig:NetworkExample]{1}(b). The expanded network representation of these stable motifs is shown in the leftmost column of Fig. \hyperref[fig:FigS5]{S5} and their corresponding node states are shown in the middle column of Fig. \hyperref[fig:FigS5]{S5}.

\subsection{Network reduction} \label{sec:S1C}
Network reduction techniques \cite{AssiehJTB,DecimationProcess,ReductionNadil,ReductionVeliz} are used to simplify a network when a node's state is known to be fixed, for example, in the case of a sustained signal. The downstream effect of this fixed state is evaluated by setting the fixed node state of interest in the Boolean function of its target nodes. As a consequence, a target node's modified Boolean function may only have one possible outcome, which means the target node's state is fixed. The whole procedure is repeated iteratively until no new fixed node states are obtained. These fixed-state nodes and their edges can be eliminated from the network.

In our work, this reduction method is used to evaluate the effect of each separate stable motif on the rest of the network \cite{ReductionChaos}. This is done by applying network reduction separately for each stable motif of the network, using the stable motif's corresponding states as the initial fixed node states. The result is a set of simplified Boolean networks, each of which corresponds to a separate stable motif, and a set of node states for each simplified network, with the latter being the node states that stay fixed due to their respective stable motifs.

\subsection{Dependence of stable motifs and attractors on the logical functions} \label{sec:S1D}
An important question related to the attractor-finding method (and, thus, to the stable motif control algorithm) is how stable motifs and attractors depend on the logical functions of the logical network in consideration. The attractor-finding  method takes as an input a given logical network, which includes both the topology of the network and the associated logic functions. Given that any topological or functional change in the logical network gives rise to a different logical network (potentially similar or potentially very different, depending on the extent of the change), the attractor-finding method needs to be applied again to the modified logical network to fully assess if the change impacts the stable motifs and/or the attractor landscape.

Even though the task of assessing the change that an arbitrary change in a logical functions brings about on a stable motif and/or the attractors is a complicated problem, it is possible to identify sufficient conditions for a target stable motif and/or attractor to be conserved after a change in the functions or topology of the logical network. For the case of a stable motif, this is done by identifying the terms of the logical functions associated with the stable motif; these terms are part of the formal definition of stable motifs in the expanded network representation of the network.

As an example, consider the logical network in Fig. \hyperref[fig:NetworkExample]{1}(a) in the main text, and its associated stable motifs in Fig. \hyperref[fig:NetworkExample]{1}(b). These stable motifs are shown in their expanded network representation in Fig. \hyperref[fig:FigS5]{S5}, together with the associated terms of the logical function. The sufficient condition for the preservation of a stable motif is that the terms associated with it stay the same. For the case of an attractor, one needs to identify the terms related to the stable motifs in each sequence associated with the attractor, and also consider the terms in the logical functions responsible for the node states that get fixed during the network reduction portion of the attractor-finding method. The whole process can become quite convoluted and is beyond the scope of this work.

\subsection{Quasi-attractors, oscillations, and the stable motif control algorithm} \label{sec:S1E}

The stable motif control algorithm uses as a starting point the stable motif succession diagram obtained from the attractor-finding method in ref. \cite{ReductionChaos}. As discussed in section A, the output of the attractor-finding method is, formally, not the system's attractors but its quasi-attractors, each of which is a network state which captures a steady state exactly and is a compressed representation of a complex (oscillating) attractor. A consequence of the relation between quasi-attractors and attractors is that certain networks with special types of complex attractors need to be treated with care when our method is applied. These special types of attractors were called unstable oscillations and incomplete oscillations in ref. \cite{ReductionChaos}.

Unstable and incomplete oscillations denote the dynamical behavior of the node state of a group of nodes that form a special type of SCC in the expanded network representation described in section A. In unstable oscillations the node state of the nodes forming the SCC oscillate in an attractor, yet are fixed in another attractor that differs only in the state of these nodes (and, potentially, on the state of nodes affected by the state of nodes in the SCC). In incomplete oscillations the node state of the nodes forming the SCC oscillate in an attractor, but do not visit all possible states of their sub-state-space in the attractor. Incomplete oscillations are the reason why undetermined states in a quasi-attractor do not necessarily oscillate.

These special types of attractors pose a challenge to the attractor-finding method, in the sense that one needs to go beyond identifying stable motifs to also identify potential unstable oscillations and incomplete oscillations. Our method can identify when a given network has the potential to have this special type of complex attractor by an extra step of analysis involving what we called oscillating components, and may in some cases involve an exploration of the sub-state-space associated with the potentially-oscillating network components. For more details, see Text \hyperref[sec:S2]{S2} section \hyperref[sec:S2A]{A} or ref. \cite{ReductionChaos}.

In some cases, the sub-state-space associated with the potentially-oscillating network components is too large to fully enumerate. In these cases, the stable motif succession diagram we can obtain without exploring this sub-state-space has an outgoing arrow which may not exist in the full stable motif succession diagram, after which there may be an attractor not found in the rest of the motif succession diagram. As a consequence, we only have partial knowledge of the full stable motif succession diagram. As we discuss in Text \hyperref[sec:S3]{S3} section \hyperref[sec:S3B3]{B.3}, partial knowledge of the stable motif succession diagram does not compromise the effectiveness of the stable motif control algorithm for the attractors in the part of the motif succession diagram we have knowledge of, but it does require us to use a modified stable motif control algorithm in which step 2 of the original algorithm is skipped.

As an example of unstable oscillations, consider the Boolean network shown in Fig. \hyperref[fig:FigS3]{S3}, which is the simplest example (up to a relabeling of node states) of unstable oscillations. The network and logical functions are given in Fig. \hyperref[fig:FigS3]{S3}(a), the state space of the system under asynchronous updating is given in Fig. \hyperref[fig:FigS3]{S3}(b), and the stable motif succession diagram is given in Fig. \hyperref[fig:FigS3]{S3}(c). Note that the states of nodes $A$ and $B$ oscillate between three network states in an attractor ($\left\{(A=1, B=0), (A=0, B=0), (A=0, B=1)\right\}$), while they are fixed in another attractor ($\left\{A=1, B=1\right\}$). Applying the attractor-finding method to this network, we find a stable motif $\{A=1, B=1\}$ and find that the set of nodes $\{A, B\}$ satisfy the necessary conditions to display unstable oscillations. Since $A$ and $B$ satisfy the necessary conditions to display unstable oscillations, one needs to search the state space spanned by this set of nodes, which in this case corresponds to the whole state space. Doing so, one finds that there is an unstable oscillation between the network states $\{A=0, B=1\}$, $\{A=0, B=0\}$, and $\{A=1, B=0\}$. The motif succession diagram in this case has the stable motif $\{A=1, B=1\}$ and the oscillating motif ($\{A=0, B=1\}$, $\{A=0, B=0\}$, $\{A=1, B=0\}$), as shown in in Fig. \hyperref[fig:FigS3]{S3}(c).

As an example of incomplete oscillations, consider the Boolean network shown in Fig. \hyperref[fig:FigS4]{S4}. The network and logical functions are given in Fig. \hyperref[fig:FigS4]{S4}(a), the state space of the system under asynchronous updating is given in Fig. \hyperref[fig:FigS4]{S4}(b), and the stable motif succession diagram is given in Fig. \hyperref[fig:FigS4]{S4}(c). Note that the states of nodes $A$ and $B$ oscillate between three subnetwork states in the attractors ($\left\{(A=1, B=0), (A=0, B=0), (A=0, B=1)\right\}$), and thus, $A$ and $B$ do not visit all possible states of their sub-state-space in each attractor. The result of applying the attractor-finding method to this network is a stable motif $\{C=0\}$ and that the set of nodes $\{A, B\}$ satisfies the conditions to display incomplete oscillations. Since $A$ and $B$ satisfy the necessary conditions to display incomplete oscillations, one needs to search the state space spanned by $A$ and $B$. Doing so, one finds that there is an incomplete oscillation between the states ($\{A=0, B=1\}$, $\{A=0, B=0\}$, $\{A=1, B=0\}$). The stable motif succession diagram for this Boolean network has the stable motif $\{C=0\}$ and the oscillating motif ($\{A=0, B=1\}$, $\{A=0, B=0\}$, $\{A=1, B=0\}$).

\subsection{Rationale and example of step 2 of the stable motif control algorithm} \label{sec:S1F}

The aim of step 2 of the stable motif control algorithm (see Methods) is to simplify the sequences of stable motifs so that the number of nodes that need to be controlled is minimized. This is done by identifying motifs after which all consequent motifs lead to the same attractor and then removing these consequent motifs from the sequence. To illustrate this, consider the stable motif succession diagram shown in Fig. \hyperref[fig:NetworkExampleS]{S1.1}. Since every possible motif after motif 1 leads to attractor 1, fixing the node states associated to motif 1 is enough to prod the system towards attractor 1. Step 2 makes sure that motifs 2 - 4 are removed from the sequences of stable motifs associated to attractor 1, since they are not necessary for the system to reach attractor 1.
\\
\begin{figure}[h]
\centerline{\includegraphics[width=0.8\textwidth]{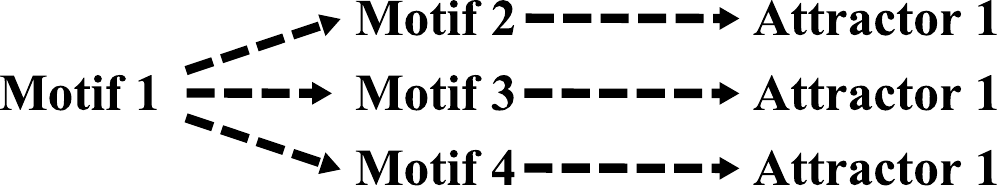}}
\caption{Example stable motifs succession diagram illustrating the simplification brought about by step 2 of the stable motif control algorithm. Step 2 removes motifs 2 - 4 from the sequences of stable motifs associated to attractor 1.}
\label{fig:NetworkExamples}
\end{figure}

\subsection{Step by step description of the stable motif control algorithm applied to the network in Fig. \hyperref[fig:NetworkExample]{1}(a)} \label{sec:S1G}
Consider the network in Fig. \hyperref[fig:NetworkExample]{1}(a) and choose $\mathcal{A}_2$ in Fig. \hyperref[fig:ReductionMethod]{2} as our target attractor. Following step 1 and using the stable motif succession diagram (Fig. \hyperref[fig:ReductionMethod]{2}), we obtain two sequences of stable motifs that lead to $\mathcal{A}_2$: $\mathcal{S}_1=\left(\right.\left\{\right.$A=1, B=1$\left.\right\},\left\{\right.$E=0$\left.\right\}\left.\right)$ and $\mathcal{S}_2=\left(\right.\left\{\right.$C=1, D=1, E=0$\left.\right\},\left\{\right.$A=1$\left.\right\}\left.\right)$. For these sequences, step 2 provides no simplification. Following step 3, the four stable motifs involved give only one subset of motif states per motif. For the first sequence, these subsets of states are $M_1=\left\{\right.$A=1, B=1$\left.\right\}$ for $\mathcal{M}_1=\left\{\right.$A=1, B=1$\left.\right\}$ and $M_2=\left\{\right.$E=0$\left.\right\}$ for $\mathcal{M}_2=\left\{\right.$E=0$\left.\right\}$. For the second sequence, the states are
$M_3=\left\{\right.$E=0$\left.\right\}$ for $\mathcal{M}_3=\left\{\right.$C=1, D=1, E=0$\left.\right\}$ and $M_4=\left\{\right.$A=1$\left.\right\}$ for $\mathcal{M}_4=\left\{\right.$A=1$\left.\right\}$. The result of step 3 are the sequences $\mathcal{S}_1=\left(O_1,O_2\right)$, where $O_1=\{$A=1$\}$ and $O_2=\{$E=0$\}$, and $\mathcal{S}_2=\left(O_3,O_4\right)$, where $O_3=\{$E=0$\}$ and $O_4=\{$A=1$\}$. Since each $O_i$ contains a single state, step 4 gives one set of states for each sequence: $\mathcal{C}_1=\left\{\right.$A=1, E=0$\left.\right\}$ for $\mathcal{S}_1$ and $\mathcal{C}_2=\left\{\right.$E=0, A=1$\left.\right\}$ for $\mathcal{S}_2$. Since both states are the same, the network control target for attractor $\mathcal{A}_2$ contains a single set of states, $C_{\mathcal{A}_2}=\left\{\right.\left\{\right.$A=1, E=0$\left.\right\}\left.\right\}$.

\clearpage

\section*{Text S2. Mathematical foundations of the attractor-finding method and of the stable motif control approach} \label{sec:S2}

In this part we describe the methods used in our work in a formal way. Part of the text in section \hyperref[sec:A]{A} is adapted from our previous work (ref. \cite{ReductionChaos}). For the propositions, lemmas, and theorems in section \hyperref[sec:A]{A}, which we proved in our previous work \cite{ReductionChaos}, we restrict ourselves to reproducing their statements and explaining their meaning, and refer the reader to our previous work \cite{ReductionChaos} for the proof.

In the following we use $V=\left( v_1, v_2, \ldots, v_N \right)$ to represent the $N$ nodes of the Boolean network, $\sigma_i, i=1, 2, \ldots, N$ to represent the state of node $v_i$, $\Sigma=\left(\sigma_1, \sigma_2, \ldots, \sigma_N \right)$ to represent the states of all nodes (also called a network state), $f_i, i=1, 2, \ldots, N$ to represent the Boolean function of node $v_i$, and $F=\left(f_1, f_2 , \ldots, f_N \right)$ to represent all the Boolean functions. We use $f(\Sigma)$ to denote a Boolean function $f$ evaluated at a network state $\Sigma$, and $f|_P$ to denote a Boolean function where only the state of a subset of nodes $P=\left\{\sigma_{p_1}, \sigma_{p_2}, \ldots, \sigma_{p_l}\right\}$ is evaluated. We commonly use $b_i$ to indicate that a specific value for node state $\sigma_i$ is chosen, that is, $\sigma_i=b_i$.

We assume, for convenience, that the Boolean functions $f_i, i=1, 2, \ldots, N$ satisfy these properties:

\begin{enumerate}
  \item The $f_i$'s do not take constant values (i.e. $f_i\neq0$ and $f_i\neq1$).
  \item If $f_i$ depends on the state of node $v_j$, $\sigma_j$, then there must be at least one pair of network states $\Sigma^{(1)}$ and $\Sigma^{(2)}$ with $\sigma_j^{(1)} \neq \sigma_j^{(2)}$, and $\sigma_k^{(1)} = \sigma_k^{(2)}$ for all $k \neq j$, such that $f_i(\Sigma^{(1)}) \neq f_i(\Sigma^{(2)})$.
  \item The $f_i$'s are written in a disjunctive normal form:
  \begin{align*}
  f_i =&\left(s_{1} \hbox{ AND } s_{2} \hbox{ AND } \cdots \hbox{ AND } s_{k} \right)\hbox{ OR }\left(s_{k+1} \hbox{ AND } s_{k+2} \hbox{ AND } \cdots s_{l} \right) \\
  & \hbox{ OR } \cdots  \hbox{ OR } \left(s_{m} \hbox{ AND } s_{m+1} \hbox{ AND } \cdots \hbox{ AND } s_{n} \right),
    \end{align*}
    where the $s_{j}$'s are either the states of one of the input nodes of $f_i$, or one of these states' negations.
  \item If for $M$, denoting a state of a subset of the inputs of $f_i$, one has $f_i|_M=1$ (regardless of the states of the remaining inputs), then the disjunctive form of $f_i$ must have at least one of its conjunctive clauses equal to 1 when evaluated at the state $M$ of this subset of nodes.
\end{enumerate}

The first property makes sure we have no source nodes. For our purposes this can be assumed without loss of generality, because even if that is not the case, we can use the reduction method of Saadatpour et al. \cite{AssiehJTB,AssiehPCB} and remove all source nodes while preserving all attractors \cite{AssiehMath}. The second property can also be assumed without any loss of generality; it is just a way of stating that we consider $f_i$ to depend on $\sigma_j$ only if it explicitly depends on $\sigma_j$ for at least a pair of network states. The third and fourth property are also general, since one can construct the respective disjunctive normal form from the truth table of the Boolean function.

The dynamics of a Boolean network $(V, \Sigma, F)$ is determined using a stochastic updating scheme known as the general asynchronous scheme \cite{GlassAsynchronous,ThomasReview,AssiehJTB}. In the general asynchronous scheme, the state of the nodes is updated at discrete time steps starting from an initial condition. At every time step, one of the variables ($\sigma_j$) is chosen randomly (uniformly) and is updated using its respective function and the state of its regulators at the previous time step
\begin{equation*}
  \sigma_j(t+1)= f_j\left(\sigma_{j_1}(t), \sigma_{j_2}(t), \dots, \sigma_{j_{k_j}}(t)\right),
\end{equation*}
while the rest of the variables retain their state.

\setcounter{subsection}{0}
\subsection{Expanded network/network reduction attractor-finding method of ref. \cite{ReductionChaos}} \label{sec:S2A}

\subsubsection{The expanded network representation} \label{sec:S2A1}

In order to identify the stable motifs of a Boolean network, we use a representation that incorporates explicitly the update functions $f_i$. Previous work \cite{ESMRuiSheng,ReductionChaos} has shown that a useful representation for this purpose is the so-called expanded network representation.

The creation of the expanded network consists of two basic operations. First, we introduce a complementary node $\overline{v}_i$ for every node $v_i$ in the network and assign to each $\overline{v}_i$ an update function $\overline{f}_i$ which is the Boolean negation of $v_i$'s update function $f_i$. The update functions $\overline{f}_i$ are assumed to satisfy the same four properties as the update functions $f_i$ without loss of generality. Second, to incorporate the combinatorial nature of the update functions, we introduce a composite node $v^{(comp)}$ for each set of synergistic interactions (that is, for each conjunctive clause) in the Boolean functions $f_i$, with a Boolean function $f^{(comp)}$ given by its respective conjunctive clause. In the following we commonly refer to nodes that are not complementary nor composite as normal nodes.

The expanded network $G_{exp}=(V_{exp}, E_{exp}, F_{exp})$ consists of a directed graph $(V_{exp}, E_{exp})$ and a set of Boolean functions $F_{exp}$, where $V_{exp}$ is the union of all normal nodes $V$, all complementary nodes $\overline{V}$, and all composite nodes $V^{(comp)}$, $E_{exp}$ is to be defined soon, and $F_{exp}$ is the union of the Boolean functions of all normal nodes $F$, the Boolean functions of all complementary nodes $\overline{F}$, and the Boolean functions of all composite nodes $F^{(comp)}$. The connectivity $E_{exp}$ of the expanded network is defined by the type of node considered. If $v \in V_{exp}$ is a composite node, then its associated Boolean function has the form $f=s_1 \hbox{ AND } s_2 \hbox{ AND } \cdots$, and $v$  has an input for each $s_j$; if $s_j$ is a node state (the negation of a node state), the input is a normal node (complementary node). If $v \in V_{exp}$ is not a composite node, then its Boolean function has the form $f=S_1 \hbox{ OR } S_2 \hbox{ OR } \cdots$, and $v$ has an input for each $S_j$; the input is a composite node if $S_j$ is a conjunctive clause, and a normal node (complementary node) if it is a node state (the negation of a node state).

\subsubsection{Identifying stable motifs from the expanded network} \label{sec:S2A2}

We define a stable motif $M$ in the expanded network as any of the smallest strongly connected components (SCCs) in the expanded network representation which satisfy these two properties:

\begin{enumerate}
\item If $M$ contains a normal node $v_i$ (complementary node $\overline{v}_i$) then $M$ does not contain its corresponding complementary node $\overline{v}_i$ (normal node $v_i$).
\item If $M$ contains a composite node $v^{(comp)}$, then all input nodes of $v^{(comp)}$ are elements of $M$.
\end{enumerate}

The first condition makes sure that there is no contradiction between the SCCs found and a state in the original Boolean network, wherein every node can either take the value 0 or 1. The second condition is a consequence of the synergistic nature of composite nodes, which means that a composite node and all of its inputs form an irreducible unit. (By smallest SCC, we mean any SCC that does not contain another SCC with the specified properties, but that, otherwise, is arbitrary in size.)

The composition of the stable motif $M$ directly determines the states $M_{state}$ of a set of nodes $M_{nodes}$ in the original Boolean network: for every normal node $v_i$ of the stable motif the corresponding node $v_i$ of the Boolean network adopts the state 1 ($\sigma_i=1$), and for every complementary nodes $\overline{v}_i$ in the stable motif the corresponding node $v_i$ of the Boolean network adopts the state 0 ($\sigma_i=0$). This set of nodes $M_{nodes}$ of the Boolean network and their corresponding states $M_{state}$ is what we defined as a stable motif in the main text, and they are such that the nodes form a minimal strongly connected component and their states form a partial fixed point of the Boolean model.

\subsubsection{Network reduction} \label{sec:S2A3}

Once the stable motifs of the network are identified, the next step is to determine the influence of these nodes on the rest of the network. More specifically, for each stable motif found, we want to find the nodes in the network whose state is fixed by the influence of this stable component. We adapt the method previously developed by Saadatpour et al. \cite{AssiehJTB,AssiehPCB} to simplify the network, which has been shown to preserve both the fixed points  and the complex attractors of the system \cite{AssiehMath}. It consists of two steps:
\begin{enumerate}
\item Identify a set of nodes $\left\{v_{p_1}, v_{p_2}, \ldots, v_{p_l}\right\}$ whose state is fixed during the dynamics, which we refer to as source nodes; for the attractor-finding method these initially correspond to the nodes in the stable motif being considered.
\item Modify the Boolean functions of the nodes downstream of the source nodes by setting the state of the source nodes to their fixed values $P=\left\{\sigma_{p_1}=b_{p_1}, \sigma_{p_2}=b_{p_2}, \ldots, \sigma_{p_l}=b_{p_l}\right\}$, that is, the modified function is given by $f|_P$. If a downstream node's modified function can only have one possible outcome, then this node can be used as a source node itself.
\end{enumerate}
For each separate stable motif found in the expanded network, these two steps are repeated recursively until neither of them can be applied anymore.

\subsubsection{The attractor-finding method algorithm and quasi-attractors} \label{sec:S2A4}

After network reduction, we obtain a set of states for each stable component, each of which corresponds to the states of the nodes in the stable motif and the states of other nodes which are fixed as a consequence of the stable motif. For each of these sets of states, there is also a reduced network that contains the nodes whose state we still do not know. On each of these reduced networks the whole method is applied again and iteratively until there are no more nodes with unknown states or no new stable motifs are found. This version of the attractor-finding method algorithm does not consider oscillatory behavior such as the one shown in Fig. \hyperref[fig:FigS3]{S3}; we come back to these cases in Text \hyperref[sec:S2]{S2} subsections \hyperref[sec:S2A5]{A.5} and \hyperref[sec:S2A6]{A.6}.

The attractor-finding method algorithm is summarized below.

\begin{enumerate}
\item Take the original Boolean network as the starting set of Boolean networks.
\item Create the expanded network representation for each of the Boolean networks.
\item Search the expanded network for stable motifs.
\item For every separate stable motif create a copy of the current network. On each of the networks created use the states of the corresponding stable motif as inputs and apply the two steps of the network reduction recursively until neither of them can be applied anymore.
\item Repeat 2-4 iteratively until there are no more nodes with unknown states or no new stable motifs are found.
\end{enumerate}

For the case where there are no more nodes with unknown states, a fixed point attractor of the system is obtained directly from the state of the nodes of the fixed-state components. For the cases in which there are no new stable motifs in the final reduced networks, the state of the nodes making up said networks is still unknown. Since our method is based on identifying nodes that are fixed in a specific steady state, the expectation is that these leftover nodes oscillate in an attractor of the system, while in that same attractor the rest of the nodes take the steady state value found during the simplification process that leads to the reduced network in consideration. We refer to the final output of our method, consisting of a set of fixed-state nodes (and their states) and a (potentially empty) set of nodes with undetermined states as a \emph{quasi-attractor}.

Quasi-attractors are closely related to the attractors of a network, both fixed points and complex attractors. For example, if the set of undetermined states in a quasi-attractor is empty, then the states of the fixed-state nodes correspond to the node states in a fixed point attractor, thus, this quasi-attractor is in fact a fixed point. More generally, for every attractor of the system there exists a quasi-attractor associated to it; this quasi-attractor is such that every node whose state is fixed in the quasi-attractor also has its state fixed in the same state in the attractor it is associated to. The proof of this is statement is given in Theorem 1 in ref. \cite{ReductionChaos}, which is reproduced in Text \hyperref[sec:S2]{S2} subsection \hyperref[sec:S2A7]{A.7}.

\subsubsection{Oscillating components and oscillations} \label{sec:S2A5}

The expanded network representation can be used to identify nodes that form an SCC in the original network, whose node states are not fixed in a complex attractor (i.e. their state oscillates). We refer to these nodes as \emph{oscillating motifs} or \emph{oscillating components}. To find the oscillating components $O$ using the expanded network representation, we search for the largest SCCs that satisfy these properties:

\begin{enumerate}
\item If $O$ contains a normal node $v_i$ then $O$ also contains its corresponding complementary node $\overline{v}_i$, and vice versa.
\item If $O$ contains a composite node $v^{(comp)}$, then all input nodes of $v^{(comp)}$ are elements of $O$.
\end{enumerate}

The first of these conditions makes sure that all nodes oscillate, by having both states of every node as part of the SCC. The second condition is a consequence of a composite node and all of its inputs forming an irreducible unit. In this case we look for the largest SCCs because we want to find all the nodes that feed back to each other in the oscillation.

These properties are necessary but not sufficient conditions for a group of node states to oscillate. We have found that there is a third condition that, if also satisfied, is sufficient (though not necessary) for a group of node states to oscillate, which is that (3) the oscillating component cannot contain stable motifs composed only of normal and complementary nodes. This extra condition is related to the possibility of the coexistence of a steady state and a complex attractor in the sub-state-space. The simplest example that shows this kind of behavior, which we denote \emph{unstable oscillation}, is shown in Fig. \hyperref[fig:FigS3]{S3}. In general, during the reduction process, we need to find the components that could display unstable oscillations (that is, that satisfy (1) and (2), but not (3)) to make sure that we preserve all attractors. As a consequence, we obtain a group of quasi-attractors that may not have a corresponding attractor; we refer to these quasi-attractors as marked quasi-attractors in the step-by-step algorithm in Text \hyperref[sec:S2]{S2} subsection \hyperref[sec:S2A6]{A.6}.

Another type of dynamical behavior of the oscillating components that needs to be considered is when the nodes of the oscillating components do not visit all possible states of their sub-state-space in an attractor, which we refer to as an \emph{incomplete oscillation}. As shown in Lemma 3 in ref. \cite{ReductionChaos} (reproduced in subsection \hyperref[sec:A7]{A.7}), nodes whose state is undetermined in a quasi-attractor are downstream of the nodes whose state oscillates in the attractor corresponding to the considered quasi-attractor. Incomplete oscillations are important because a node that is downstream of an oscillating component that displays incomplete oscillations may reach a steady state as a consequence of the nodes of the component only visiting part of their sub-state-space. Incomplete oscillations are the reason why undetermined states in a quasi-attractor do not necessarily oscillate.

\subsubsection{The full algorithm of the expanded network/network reduction attractor-finding method} \label{sec:S2A6}

In the following we describe the full algorithm of the attractor-finding method. Unlike the algorithm introduced in Text \hyperref[sec:S2]{S2} subsection \hyperref[sec:S2A4]{A.4}, the following algorithm considers the so-called  unstable oscillations, such as the one shown in Fig. \hyperref[sec:S3]{S3}. During the description of the algorithm we refer the reader to the subsections in Text \hyperref[sec:S2]{S2} where each of these steps are described in more detail.

\begin{enumerate}
\item For every combination of the states of the source nodes (nodes with no upstream components) apply the two steps of network reduction method described in Text \hyperref[sec:S2]{S2} subsection \hyperref[sec:S2A3]{A.3} recursively until neither of them can be applied anymore.
\item Create the expanded network representation for each of the resulting networks (Text \hyperref[sec:S2]{S2} subsection \hyperref[sec:S2A1]{A.1}).
\item Search the expanded network for stable motifs (Text \hyperref[sec:S2]{S2} subsection \hyperref[sec:S2A2]{A.2}) and oscillating components (Text \hyperref[sec:S2]{S2} subsection \hyperref[sec:S2A5]{A.5}).
\item For every separate stable motif create a copy of the current network. On each of the networks created use the states of the corresponding stable motif as inputs and apply the two steps of the network reduction described in Text \hyperref[sec:S2]{S2} subsection \hyperref[sec:S2A3]{A.3} recursively until neither of them can be applied anymore.
\item For every oscillating component of more than two nodes (i.e., every oscillating component that could display incomplete oscillations) create a copy of the current network. On each of the networks created, the nodes in the corresponding oscillating component and the nodes downstream of this component are marked. The marked nodes cannot be reduced at any later step of the algorithm (i.e, they have their state undetermined in the quasi-attractors that are derived from these networks).
\item For the oscillating components of two nodes (i.e, only one normal node and its corresponding complementary node), check if any node downstream of these oscillating motifs participates in a stable motif with no composite nodes. If any of them do, go to step 7; otherwise, check if there are any stable motifs that are downstream of these oscillating components (these stable motifs would necessarily have a composite node). If there are not, go to step 7; if there are, check if any of them is downstream of a stable motif that is itself not downstream of any of these oscillating components. If this is the case, go to step 7; if this is not the case, then create one copy of the current network and mark the nodes in the oscillating motifs considered in this step and the nodes downstream of them. The marked nodes cannot be reduced at any later step of the algorithm (i.e, they have their state undetermined in the quasi-attractors that are derived from these networks).
\item Repeat 2-6 for each of the networks iteratively until no more stable motifs are found. The result, a set of fixed state nodes and their node state, and a set of nodes with undetermined states with their reduced Boolean functions, is the set of quasi-attractors (Text \hyperref[sec:S2]{S2} subsection \hyperref[sec:S2A4]{A.4}).
\item Prune the set of quasi-attractors of duplicates (two quasi-attractors are the same if they have the same set of fixed state nodes and the same node state for these fixed-state nodes; if two quasi-attractors are the same, except that one of them has some nodes marked while the other one does not, remove the one that has the marked nodes).
\end{enumerate}

Some of the resulting quasi-attractors have marked nodes while others do not. For every unmarked quasi-attractor there necessarily is a corresponding attractor in the Boolean network. For a marked quasi-attractor there may not be a corresponding attractor in the Boolean network; only by knowing the specific states visited during oscillations by the undetermined nodes in the quasi-attractor's reduced network can one confirm whether there is a corresponding attractor (this is a consequence of incomplete oscillations and unstable oscillations, see Text \hyperref[sec:S2]{S2} subsection \hyperref[sec:S2A5]{A.5}).

\subsubsection{Conservation of attractors by the expanded network/network reduction attractor-finding method} \label{sec:S2A7}

The first proposition states that the stable motifs found from the expanded network are such that the corresponding states of these motifs are partial fixed points of the Boolean rules of the nodes involved.
\begin{myprop} \label{PartFixedPoints}
Let $M=\left(V_{m_1}, V_{m_2}, \ldots, V_{m_l},\right.$$\left.V_{m_{l+1}}, V_{m_{l+2}}, \ldots, V_{m_L}\right)$ be a stable motif in the expanded network representation, where $V_{m_1}, V_{m_2}, \ldots, V_{m_l}$ can either be a normal node or a complementary node, and where $V_{m_{l+1}}, V_{m_{l+2}}, \ldots, V_{m_L}$ are composite nodes. We denote $M_{state}=\left(\sigma_{m_1}=b_{m_1},\right.$$\left.\sigma_{m_2}=b_{m_2}, \ldots, \sigma_{m_l}=b_{m_l} \right)$, with $b_{m_j} \in \left\{0, 1\right\}$ as the corresponding state of $M$ in the network state $\Sigma$: $b_{m_j}=1$ if it is a normal node, and $b_{m_j}=0$ if it is a complementary node. Then, for any normal node $v_{m_j}$ or complementary node $\overline{v}_{m_j}$ in $M$ and for any network state $\Sigma_M$ such that $\sigma_{m_k}=b_{m_k} \forall m_k \in \left\{m_1, m_2, \ldots, m_l\right\}$, we have $f_{m_j}(\Sigma_M)=b_{m_j}$.
\end{myprop}

The reverse of this proposition is also true, that is, if for a given set of node states updating any of the states in the set gives back the same state, regardless of the state of any node outside of the set, then this set of states correspond to a set of stable motifs in the expanded network representation:

\begin{myprop} \label{Stablemotifs-fixedpoints}
Let $M_{state}=\left(\sigma_{m_1}=b_{m_1}, \sigma_{m_2}=b_{m_2},\right.$ $\left.\ldots, \sigma_{m_l}=b_{m_l} \right)$ be the state of a set of nodes such that if $\Sigma_M$ is any network state in which $\sigma_{m_k}=b_{m_k} \forall m_k \in \left\{m_1, m_2, \ldots, m_l\right\}$, then $f_{m_j}(\Sigma_M)=b_{m_j}$. Then (i) there is a set of stable motifs $\left\{M_n\right\}$ in the expanded network representation such that each of the $M_n$'s contain only normal nodes or complementary nodes of the nodes whose state is specified in $M_{state}$ (normal nodes if $b_{m_k}=1$, and complementary nodes if $b_{m_k}=0$) and in which all other nodes in the $M_n$'s (if any) are composite nodes made up of the normal nodes or complementary nodes in the corresponding $M_n$, and (ii) the nodes whose state is specified in $M_{state}$ but that are not included in the set of stable motifs $\left\{M_n\right\}$ are downstream of the nodes in at least one of the stable motifs.
\end{myprop}

For the next propositions we need certain properties of the attractors of the general asynchronous updating scheme, in which the state of one randomly (uniformly) chosen node is updated at every discrete time step (see Methods). For any attractor $\mathcal{A}$, we can divide the $N$ nodes into two classes: those that take the same value in all network states of $\mathcal{A}$ (i.e, either 0 or 1), and those that take more than one value in the different network states of $\mathcal{A}$ (i.e, both 0 and 1). We refer to the former as \emph{stabilized or fixed-state nodes}, and to the latter as \emph{oscillating nodes}. The following propositions state that fixed-state nodes can have inputs from fixed-state nodes or oscillating nodes (Proposition \ref{Stablemotifs-inputs}), while oscillating nodes must have at least one oscillating node as an input (Proposition \ref{Oscmotifs-inputs}).

\begin{myprop} \label{Stablemotifs-inputs}
Let $\mathcal{A}$ be an attractor of the Boolean network $(V,\Sigma,F)$ under the general asynchronous updating scheme, and let $\mathcal{S}$ and $\mathcal{O}$ be the set of the fixed-state and oscillating nodes in the attractor, respectively. If $v_{s} \in \mathcal{S}$, and $b_s$ is the fixed-state state of node $v_{s}$, then one of the following two cases holds: (i) one of the conjunctive clauses of $f_s$ (if $b_s=1$) or $\overline{f}_s$ (if $b_s=0$) depends only on the specific state of the nodes of $\mathcal{S}$ in $\mathcal{A}$. If (i) is not true, then (ii) for both $f_s$ and $\overline{f}_s$ at least one conjunctive clause depends on the state of one or more nodes in $\mathcal{O}$ and, if the clause depends on any more states, they have to be the state of the nodes of $\mathcal{S}$ in $\mathcal{A}$.
\end{myprop}

\begin{myprop} \label{Oscmotifs-inputs}
Let $\mathcal{A}$ be an attractor of the Boolean network $(V,\Sigma,F)$ under the general asynchronous updating scheme, and let $\mathcal{S}$ and $\mathcal{O}$ be the set of the fixed-state and oscillating nodes, respectively. If $v_{o} \in \mathcal{O}$ then (i) neither $f_o$ nor $\overline{f}_o$ can have any conjunctive clauses that depend only on the state of the nodes of $\mathcal{S}$ in $\mathcal{A}$ (i.e, on $\sigma_s$ if $b_s=1$, or $\overline{\sigma}_s$ if $b_s=0$), and (ii) both $f_o$ and $\overline{f}_o$ must have at least one conjunctive clause that depends on the state of one or more nodes in $\mathcal{O}$ and, if this same clause depends on any other states, they must be the states of nodes of $\mathcal{S}$ in $\mathcal{A}$.
\end{myprop}

We now reproduce the three lemmas that allow us to show that the reduction method conserves all attractors. In Lemma \ref{Stablemotifs-reduction1} we construct the set of nodes, for an arbitrary attractor, whose state are identified by our attractor-finding method, $\mathcal{S}_{red} \subset \mathcal{S}$. We also show that there is at least one stable motif composed of the corresponding states in the attractor of the nodes of $\mathcal{S}_{red}$ (as long as $\mathcal{S}_{red}$ is not empty). In Lemma \ref{Stablemotifs-reduction2} we show that the network reduction of these stable motifs can only fix the state of nodes in $\mathcal{S}_{red}$. In Lemma \ref{Oscillatingmotifs-reduction} we show that when no stable motifs are found, which is the exit condition in the attractor-finding algorithm (step 7, Text \hyperref[sec:S2]{S2} subsection \hyperref[sec:S2A6]{A.6})), the fixed-state nodes in an attractor $\mathcal{A}$ must be downstream of an oscillating motif.

For completeness, we reproduce how $\mathcal{S}_{red} \subset \mathcal{S}$ is constructed for a Boolean network attractor $\mathcal{A}$. Without loss of generality we can do a change of variables so that $\sigma_{s}=1$ if $v_s \in \mathcal{S}$. By Proposition \ref{Stablemotifs-inputs}, we can divide $\mathcal{S}$ into the nodes that have at least one conjunctive clause in their rule that depends only on the specific state of nodes of $\mathcal{S}$ in $\mathcal{A}$, and their complement in $\mathcal{S}$. We refer to the former as $\mathcal{S}_{0}$ and to the latter as $\mathcal{S}_{osc}$. Let $\mathcal{S}_{1} \subset \mathcal{S}_{0}$  be the nodes that have at least one conjunctive clause in their rules that depends only on the specific state of the nodes of $\mathcal{S}_{0}$ in $\mathcal{A}$ (i.e, on $\sigma_s$, because of the change of variables). Let $\mathcal{S}_{2} \subset \mathcal{S}_{1}$  be the nodes that have at least one conjunctive clause in their rules that depends only on the specific state of nodes of $\mathcal{S}_{1}$ in $\mathcal{A}$ (note they could depend on the states of nodes in $\mathcal{S}_{0} - \mathcal{S}_{1}$). We do this iteratively until $\mathcal{S}_{i_{max}+1}=\mathcal{S}_{i_{max}}$ and denote $\mathcal{S}_{red}=\mathcal{S}_{i_{max}} \subset \mathcal{S}_{0}$. Since $\mathcal{S}_{red}$ was constructed by first removing the nodes that required nodes in $\mathcal{O}$ to have their states fixed, and then removing the ones that depended on the previously reduced nodes, and so on, then $\mathcal{S}_{red}$ corresponds to the set of nodes in $\mathcal{S}$ that do not depend in any way on nodes of $\mathcal{O}$ to have their node state fixed in their state on $\mathcal{A}$.

\begin{mylemma} \label{Stablemotifs-reduction1}
Let $\mathcal{A}$ be an attractor of the Boolean network $(V,\Sigma,F)$ under the general asynchronous updating scheme, and let $\mathcal{S}$ and $\mathcal{O}$ be the set of the fixed-state and oscillating nodes of $\mathcal{A}$, respectively. There exists a set of nodes $\mathcal{S}_{red} \subset \mathcal{S}$ such that in the expanded network representation of $(V,\Sigma,F)$ there is at least one stable motif composed only of the corresponding states of the nodes of $\mathcal{S}_{red}$ in $\mathcal{A}$, or composite nodes composed of such nodes.
\end{mylemma}

\begin{mylemma} \label{Stablemotifs-reduction2}
Let $\mathcal{S}_{red} \subset \mathcal{S}$ be the constructed set of nodes in Lemma \ref{Stablemotifs-reduction1}. Then (i) $\mathcal{S}_{red}$ is such that the network reduction of any stable motif composed only of the corresponding states of $\mathcal{S}_{red}$ in $\mathcal{A}$ (or composite nodes composed of such nodes) can only fix the state of nodes in $\mathcal{S}_{red}$, and (ii) if any of the states of the nodes in $\mathcal{S}_{red}$ is fixed by network reduction, then it has to be on their corresponding state in $\mathcal{A}$; if they do not have their state fixed, then either their rule (if their fixed state in $\mathcal{A}$ is 1) or the negation of their rule (if their fixed state is 0) in the reduced network have a conjunctive clause that only depends on the specific state of the nodes of $\mathcal{S}_{red}$ in $\mathcal{A}$ (i.e, on $\sigma_s$ if $b_s=1$, or $\overline{\sigma}_s$ if $b_s=0$) that did not have their states fixed during network reduction.
\end{mylemma}

\begin{mylemma} \label{Oscillatingmotifs-reduction}
Let $\mathcal{A}$ be an attractor of the Boolean network $(V,\Sigma,F)$ under the general asynchronous updating scheme, and let $\mathcal{S}$ and $\mathcal{O}$ be the set of the fixed-state and oscillating nodes, respectively. Let $\mathcal{S}_{red} \subset \mathcal{S}$ be the constructed set of nodes in Lemma \ref{Stablemotifs-reduction1} and assume that $\mathcal{S}_{red}$ is empty and that $\mathcal{O}$ is a non empty set. Then the expanded network representation of $(V,\Sigma,F)$ must be such that the normal nodes and complementary nodes of the elements in $\mathcal{O}$, and the nodes corresponding to the state of the nodes of $\mathcal{S}$ in $\mathcal{A}$ must both be downstream of an oscillating motif that contains at least one of the nodes in $\mathcal{O}$.
\end{mylemma}

The following theorem is the main result of this section, and it combines the results of Lemma \ref{Stablemotifs-reduction1}, \ref{Stablemotifs-reduction2}, and \ref{Oscillatingmotifs-reduction}. It shows that for every attractor (under general asynchronous updating, see Methods) in the network, our attractor-finding method finds a corresponding quasi-attractor in which the state of the nodes in $\mathcal{S}_{red}$ is the same as in the attractor, and in which the rest of the nodes are either be part of an oscillating motif or downstream of it.

\begin{mythm} \label{Attractorconservation}
Let $\mathcal{A}$ be an attractor of the Boolean network $(V,\Sigma,F)$ under the general asynchronous updating scheme, and let $\mathcal{S}$ and $\mathcal{O}$ be the set of the fixed-state and oscillating nodes, respectively. Let $\mathcal{S}_{red} \subset \mathcal{S}$ be the set of nodes constructed in Lemma \ref{Stablemotifs-reduction1}. Then, there exists a set of stable motifs such that, by applying network reduction, all the nodes in $\mathcal{S}_{red}$ get fixed in their steady state in $\mathcal{A}$, while the rest of the nodes in $V$ are part of the final reduced network. This resulting final reduced network is such that, in its expanded network representation, all the nodes are either be part of an oscillating motif containing at least one of the nodes in $\mathcal{O}$, or be downstream of an oscillating motif.
\end{mythm}

\subsection{The stable motif control method} \label{sec:S2B}

\subsubsection{A sequence of stable motifs uniquely determines an equivalence class of attractors} \label{sec:S2B1}

The main step in proving that our stable motif control method works is to show that a sequence of stable motifs obtained from the attractor-finding method uniquely determines an attractor. Since the attractor-finding method yields quasi-attractors, we need to be more precise with what ``uniquely determines an attractor'' refers to in this context. Let $(V,\Sigma,F)$ be a Boolean network, and let $\mathbb{A}=\left\{\mathcal{A}_i\right\}, i=1, 2, \ldots, n_{\mathbb{A}}$ be the set of general asynchronous attractors of $(V,\Sigma,F)$. We define $\mathbb{A}^{(red)}=\{\mathcal{A}^{(red)}_j\}, j=1, 2, \ldots, n_{\mathbb{A}^{(red)}}$ as the partition of the attractors $\mathbb{A}$ generated by the equivalence relation $\sim$, where $\mathcal{A}_k \sim \mathcal{A}_l$ if the $\mathcal{S}_{red}$ for $\mathcal{A}_k$ (as defined in Text \hyperref[sec:S2]{S2} subsection \hyperref[sec:S2A7]{A.7}) is the same as the $\mathcal{S}_{red}$ for $\mathcal{A}_l$, and the state of each node $v \in \mathcal{S}_{red}$ is the same in both $\mathcal{A}_k$ and $\mathcal{A}_l$. Consequently, each $\mathcal{A}_i \in \mathbb{A}$ is an element of only one $\mathcal{A}^{(red)}_j \in \mathbb{A}^{(red)}$ (since it is a partition generated by an equivalence relation), and $\forall \mathcal{A}_k, \mathcal{A}_l \in \mathcal{A}^{(red)}_j$, we have $\mathcal{A}_k \sim \mathcal{A}_l$.

Using the above we can now be more precise: For a Boolean network $(V,\Sigma,F)$ under general asynchronous updating, a sequence of stable motifs obtained from the attractor-finding method uniquely determines an equivalence class of attractors $\mathcal{A}^{(red)}$, each of which is the set of all attractors of $(V,\Sigma,F)$ that share the same $\mathcal{S}_{red}$ and the state of each node in $\mathcal{S}_{red}$ ($\mathcal{S}_{red}$ corresponds to the set of fixed-state nodes in an attractor $\mathcal{A}$ that do not depend in any way on the state of nodes whose state oscillates in $\mathcal{A}$ to have their node state fixed in their state on $\mathcal{A}$). We prove this below.

\begin{mylemma} \label{SequenceStableMotifs}
Let $\mathcal{B}=(V,\Sigma,F)$ be a Boolean network, let $\mathbb{A}^{(red)}$ be the set of equivalence classes of attractors defined above, and let $S_{red}$ and $S_{red, \Sigma}$ denote, respectively, the set of nodes and node states which define an equivalence class of attractors $\mathcal{A}^{(red)} \in \mathbb{A}^{(red)}$. Let $\mathcal{S}_{seq}=\left(\mathcal{M}_1, \ldots, \mathcal{M}_L\right)$ be a sequence of stable motifs of the Boolean network obtained by the attractor-finding method (section \hyperref[sec:A]{A}), let $Q$ be its associated quasi-attractor, and let $\mathcal{S}_Q$ and $Q_{\Sigma}$ be the set of fixed-state nodes in $Q$ and the state of the fixed-state nodes in $Q$, respectively. Then $\mathcal{S}_Q$ and $Q_{\Sigma}$ are such that $\mathcal{S}_Q=\mathcal{S}_{red}$ and $Q_{\Sigma}=S_{red, \Sigma}$ for only one $\mathcal{A}^{(red)} \in \mathbb{A}^{(red)}$.
\end{mylemma}
\begin{proof}
Let $\mathcal{B}_i, i=1, 2, \ldots, L$ be the reduced Boolean network obtained from the Boolean network $\mathcal{B}$ after applying network reduction up to and including the stable motif $\mathcal{M}_i$ in the sequence $\mathcal{S}_{seq}$, and define $\mathcal{B}_0\equiv\mathcal{B}$. By construction, one of the stable motifs of the Boolean network $\mathcal{B}_i$ is $\mathcal{M}_{i+1}$. Let $\mathcal{R}_i, i=1, 2, \ldots, L$  be the node state of the nodes in the Boolean network $\mathcal{B}_{i-1}$ whose node state becomes fixed after applying network reduction with motif $\mathcal{M}_i$. By the definition of quasi-attractor $Q$, $\mathcal{S}_Q$ is given by the set of nodes whose state is specified in $\mathcal{M}_i$ or $\mathcal{R}_i$, and $Q_{\Sigma}$ is given by the nodes states specified in $\mathcal{M}_i$ or $\mathcal{R}_i$, that is,
\begin{align*}
  Q_{\Sigma} &\equiv \bigcup_{j=1}^{L} \mathcal{M}_j \cup \mathcal{R}_j = \left\{\sigma_{q_1}=b_{q_1}, \sigma_{q_2}=b_{q_2}, \ldots, \sigma_{q_{n_{Q}}}=b_{q_{n_{Q}}}\right\}, \\
  \mathcal{S}_Q &= \left\{v_{q_1}, v_{q_2}, \ldots, v_{q_{n_{Q}}}\right\} .
\end{align*}
Let $\mathcal{A}^{(red)} \in \mathbb{A}^{(red)}$, and let $\mathcal{S}_{red}$ and $\mathcal{S}_{red, \Sigma}$ be the set of fixed-state nodes and the state of the fixed-states nodes, respectively, which define the equivalence class $\mathcal{A}^{(red)}$. Let $\mathbb{A}^{(red)}_Q \subset \mathbb{A}^{(red)}$ be all the equivalence classes of attractors $\mathcal{A}^{(red)} \in \mathbb{A}^{(red)}$ for which their defining $\mathcal{S}_{red}$ and $\mathcal{S}_{red, \Sigma}$ satisfy $\mathcal{S}_Q\subseteq \mathcal{S}_{red}$ and $Q_{\Sigma} \subseteq \mathcal{S}_{red, \Sigma}$. Note that $\mathbb{A}^{(red)}_Q$ cannot be an empty set; this is a consequence of stable motifs being partial fixed points of the dynamics (Proposition \ref{PartFixedPoints}), and the finite size of the state space spanned by all $\Sigma$ which satisfy $\sigma_{q_i}=b_{q_i} \forall v_{q_i} \in \mathcal{S}$.

Let $\mathcal{A}'^{(red)} \in \mathbb{A}^{(red)}_Q$ and let $\mathcal{S}'_{red}\neq \mathcal{S}_Q$ and $\mathcal{S}'_{red, \Sigma}$ its defining set of fixed-state nodes and node states. We now show, by contradiction, that $\mathcal{S}'_{red} \equiv \mathcal{S}_Q$. Lemma \ref{Stablemotifs-reduction2} (with $\mathcal{S}'_{red}$) guarantees that network reduction of the stable motifs in $\mathcal{S}_{seq}$ can only fix the state of nodes in $\mathcal{S}_Q$ on their corresponding state in $Q_{\Sigma}$, since $\mathcal{S}_Q\subseteq \mathcal{S}'_{red}$ and $Q_{\Sigma} \subseteq \mathcal{S}'_{red, \Sigma}$. Lemma \ref{Stablemotifs-reduction2} also guarantees that each node $v_i \in \mathcal{S}'_{red}-\mathcal{S}_Q$ has a conjunctive clause in its Boolean function in $\mathcal{B}_L$ if $\sigma_i=1$ in $\mathcal{S}'_{red, \Sigma}$, or in the negation of their Boolean function if $\sigma_i=0$ in $\mathcal{S}'_{red, \Sigma}$, that only depends on the specific state of the nodes of $\mathcal{S}'_{red}-\mathcal{S}_Q$ in $\mathcal{S}'_{red, \Sigma}$.

From the above, Lemma \ref{Stablemotifs-reduction2} implies that each node $v_i \in \mathcal{S}'_{red}-\mathcal{S}_Q$ has an associated normal node (if $\sigma_i=1$ in $\mathcal{S}'_{red, \Sigma}$) or complementary node (if $\sigma_i=0$ in $\mathcal{S}'_{red, \Sigma}$) in the expanded network representation of $\mathcal{B}_L$ with, at least, one expanded network input node $v_j$, where $v_j$ is either (a) the associated normal node of a node in $\mathcal{S}'_{red}-\mathcal{S}_Q$ whose node state is $\sigma_j=1$ in $\mathcal{S}'_{red, \Sigma}$, (b) the associated complementary node of a node in $\mathcal{S}'_{red}-\mathcal{S}_Q$ whose node state is $\sigma_j=0$ in $\mathcal{S}'_{red, \Sigma}$, or (c) a composite node with only nodes in (a) and/or (b) as input nodes. Consequently , the expanded network representation of $\mathcal{B}_L$ must have a stable motif composed only of expanded network nodes $v_j$ satisfying either (a), (b) and (c). But this is not possible, since $\mathcal{B}_L$ has no stable motifs (if it had, then it would be part of the sequence $\mathcal{S}_{seq}$). By contradiction, we must have $\mathcal{S}'_{red} \equiv \mathcal{S}_Q$.

From the previous paragraph we have $\mathcal{S}'_{red} \equiv \mathcal{S}_Q$. This implies that $\mathbb{A}^{(red)}_Q \subset \mathbb{A}^{(red)}$ is composed of a single equivalence class of attractors $\mathcal{A}^{(red)}$, which has $\mathcal{S}_{red}=\mathcal{S}_Q$ and $\mathcal{S}_{red, \Sigma}=Q_{\Sigma}$, thus concluding our proof.
\end{proof}

\subsubsection{Fixing the node states specified by a sequence of stable motifs} \label{sec:S2B2}

Lemma \ref{SequenceStableMotifs} shows that a sequence of stable motifs uniquely determines an equivalence class of attractors, but this does not directly show the result of fixing the node states specified by a sequence of stable motifs. For this, we first need to define what we mean with a Boolean network in which a set of node states is fixed. For a Boolean network $\mathcal{B}=(V,\Sigma,F)$ under general asynchronous updating and a set of node states $P=\left\{\sigma_{p_1}=b_{p_1}, \sigma_{p_2}=b_{p_2}, \ldots, \sigma_{p_l}=b_{p_l}\right\}$, we denote $\mathcal{B}_P=(V,\Sigma,F')$, with $f'_i \in F'$ such that $f'_i=f_i|_P$ if $i \not \in \{p_1, \ldots, p_l\}$ or $f'_i=b_i$ if $i \in \{p_1, \ldots, p_l\}$, as the Boolean network in which $P$ is fixed.

Note that, formally, $\sigma_{p_i}\neq b_{p_i}$ for any $\sigma_{p_i} \in P$ is an allowed state of $\Sigma$ in the Boolean network $\mathcal{B}_P$. However, no attractors in $\mathcal{B}_P$ have $\sigma_{p_i} \neq b_{p_i}$ for any $\sigma_{p_i} \in P$ since $f'_i=b_i, \forall i \in \{p_1, \ldots, p_l\}$. Furthermore, if we restrict ourselves to the substate space of $\mathcal{B}_P=(V,\Sigma,F')$ with $\sigma_i=b_i, \forall \sigma_i \in P$, it can be shown that this substate space is equivalent to having taken $\mathcal{B}$, restricting it to $\sigma_i=b_i, \forall \sigma_i \in P$, and removing the transitions from network states with $\sigma_i=b_i, \forall \sigma_i \in P$ to network states with $\sigma_i\neq b_i$ for at least one $\sigma_i \in P$. We choose to work with $\mathcal{B}_P$ instead of a restricted $\mathcal{B}$ because of its similarity with the network reduction process of the attractor-finding method.

We now show that for a sequence of stable motifs $S_{seq}$ with associated quasi-attractor $Q$ and fixed-node states $Q_{\Sigma}$, the Boolean network in which $Q_{\Sigma}$ is fixed has the same attractors as the equivalence class of attractors specified by $S_{seq}$. For this, we use a more general statement for which the above is a special case.

\begin{myprop} \label{ReducedFixedNetworks1}
Let $\mathcal{S}_{seq}=\left(\mathcal{M}_1, \mathcal{M}_2, \ldots, \mathcal{M}_L\right)$ be a sequence of stable motifs of $\mathcal{B}=(V,\Sigma,F)$ obtained by the attractor-finding method, and let $\mathcal{B}_{\lambda}, \lambda\in\{1,2,\ldots,L\}$ be the reduced Boolean network obtained from $\mathcal{B}$ after applying network reduction up to and including the stable motif $\mathcal{M}_{\lambda}$ in the sequence $\mathcal{S}_{seq}$. Let $Q_{\Sigma,\lambda}$ be the following set of node states
\begin{equation*}
  Q_{\Sigma,\lambda} = \bigcup_{j=1}^{\lambda} \mathcal{M}_j \cup \mathcal{R}_j,
\end{equation*}
where $\mathcal{R}_i, i=1, 2, \ldots, L$ is defined in the proof of Lemma \ref{SequenceStableMotifs}. Then, the attractors in the reduced network $\mathcal{B}_{\lambda}$ are the same as the attractors of the Boolean network $\mathcal{B}_{Q_{\Sigma,\lambda}}=(V,\Sigma,F')$ when comparing only the state of nodes present in both $\mathcal{B}_{\lambda}$ and $\mathcal{B}_{Q_{\Sigma,\lambda}}$, and the nodes $v_i$ present only in $\mathcal{B}_{Q_{\Sigma,\lambda}}$ are such that their state in all attractors is given by $\sigma_i=b_i, \sigma_i \in Q_{\Sigma,\lambda}$.
\end{myprop}

Let us sketch the proof for this proposition. By construction, the Boolean function $f'_i$ of node $v_i$ in $\mathcal{B}_{Q_{\Sigma,\lambda}}$ is the same as the Boolean function of node $v_i$ of $\mathcal{B}_{\lambda}$ if $v_i$ is present in both $\mathcal{B}_{Q_{\Sigma,\lambda}}$ and $\mathcal{B}_{\lambda}$; this is because the Boolean functions of a reduced network are given by
\begin{equation*}
f_i|_{\mathcal{M}_1 \cup \mathcal{R}_1}|_{\mathcal{M}_2 \cup \mathcal{R}_2}|\ldots |_{\mathcal{M}_{\lambda} \cup \mathcal{R}_{\lambda}} = f_i|_{Q_{\Sigma,\lambda}} \equiv f'_i.
\end{equation*}
The nodes present in $\mathcal{B}_{Q_{\Sigma,\lambda}}$ but not in $\mathcal{B}_{\lambda}$ are the nodes whose state is specified in $Q_{\Sigma,\lambda}$. The Boolean function of each of these nodes is $f'_i=b_i$, where $b_i$ is specified in $Q_{\Sigma,\lambda}$. Since this implies that there is always a transition from any network state with at least one $\sigma_i \neq b_i, \sigma_i \in Q_{\Sigma,\lambda}$ to a network state with $\sigma_i=b_i, \sigma_i \in Q_{\Sigma,\lambda}$, but not the other way around, an attractor of $\mathcal{B}_{Q_{\Sigma,\lambda}}$ must have $\sigma_i=b_i, \forall \sigma_i \in Q_{\Sigma,\lambda}$. Since the functions of all nodes present in both $\mathcal{B}_{Q_{\Sigma,\lambda}}$ and $\mathcal{B}_{\lambda}$ are the same, and the attractors of $\mathcal{B}_{Q_{\Sigma,\lambda}}$ must have $\sigma_i=b_i, \forall \sigma_i \in Q_{\Sigma,\lambda}$, then the attractors in the reduced network $\mathcal{B}_{\lambda}$ must be the same as the attractors of the Boolean network $\mathcal{B}_{Q_{\Sigma,\lambda}}$ when comparing only the state of nodes present in both $\mathcal{B}_{\lambda}$ and $\mathcal{B}_{Q_{\Sigma,\lambda}}$.

Using Lemma \ref{SequenceStableMotifs}, Proposition \ref{ReducedFixedNetworks1}, and the fact that $Q_{\Sigma,\lambda=L}\equiv Q_{\Sigma}$ (as defined in Proposition \ref{ReducedFixedNetworks1} and Lemma \ref{SequenceStableMotifs}, respectively), we can prove that the attractors of the Boolean network $\mathcal{B}_{Q_{\Sigma}}$ are the same as the attractors in the equivalence class defined by quasi-attractor $Q$.

\begin{myprop} \label{ReducedFixedNetworks2}
Let $\mathcal{S}_{seq}=\left(\mathcal{M}_1, \ldots, \mathcal{M}_L\right)$ be a sequence of stable motifs of $\mathcal{B}=(V,\Sigma,F)$ obtained by the attractor-finding method, let $Q$ be its associated quasi-attractor, and let $\mathcal{S}_Q$ and $Q_{\Sigma}$ be the set of fixed-state nodes in $Q$ and the state of the fixed-state nodes in $Q$, respectively. Let $\mathcal{A}^{(red)}$ be the equivalence class of attractors with $\mathcal{S}_{red}=\mathcal{S}_Q$ and $\mathcal{S}_{red, \Sigma}=Q_{\Sigma}$ given by Lemma \ref{SequenceStableMotifs}. Then, the attractors of the Boolean network $\mathcal{B}_{Q_{\Sigma}}=(V,\Sigma,F')$ are the same as the attractors in $\mathcal{A}^{(red)}$.
\end{myprop}

Before proceeding, we want to emphasize the role of Lemma \ref{SequenceStableMotifs}, Proposition \ref{ReducedFixedNetworks1}, and Proposition \ref{ReducedFixedNetworks2} in proving that the stable motif control algorithm works. Lemma \ref{SequenceStableMotifs} is the main result of section \hyperref[sec:B]{B}, and shows that a sequence of stable motifs $\mathcal{S}_{red}$ uniquely determines an equivalence class of attractors $\mathcal{A}^{(red)}$. Proposition \ref{ReducedFixedNetworks2} shows that the Boolean network obtained by fixing the node states specified by $\mathcal{S}_{red}$ has the attractors in $\mathcal{A}^{(red)}$ as its only attractors, and is a direct consequence of Lemma \ref{SequenceStableMotifs} and the attractor-finding method (section \hyperref[sec:A]{A}). Lemma \ref{SequenceStableMotifs} and Proposition \ref{ReducedFixedNetworks2} guarantee the effectiveness of the stable motif control algorithm: each sequence of stable motifs $\mathcal{S}_{red}$ obtained from the attractor-finding method singles out one equivalent class of attractors $\mathcal{A}^{(red)}$ (Lemma \ref{SequenceStableMotifs}), and Boolean network obtained by fixing the node states specified by $\mathcal{S}_{red}$ has the attractors in $\mathcal{A}^{(red)}$ as its only attractors (Proposition \ref{ReducedFixedNetworks2}).

Proposition \ref{ReducedFixedNetworks2} shows that the attractors of the reduced Boolean networks obtained using a shortened subsequence of $\mathcal{S}_{red}$ are equivalent to the attractors of the Boolean network obtained by fixing the node states specified by this shortened subsequence. This allows us to consider only the attractors of reduced Boolean networks when showing that a sequence simplified by steps 2 and 3 of the stable motif control algorithm (which we reproduce below) has the same effect as the full sequence $\mathcal{S}_{red}$.

\subsubsection{Simplifying the sequences of stable motifs} \label{sec:S2B3}

For completeness, we reproduce the stable motif control algorithm (see Methods and Text \hyperref[sec:S7]{S7} for more details):
\begin{itemize}
\item[-] \textit{Step 1}: Identify the sequences of stable motifs that lead to $\mathcal{A}$. These can be obtained from the stable motif succession diagram (see Fig. \hyperref[fig:ReductionMethod]{2}) by choosing the attractor of interest in the right-most part and selecting all of the attractor's predecessors in the succession diagram.

\item[-] \textit{Step 2}: Shorten each sequence $\mathcal{S} \in Sequences$ by identifying the minimum number of motifs in $\mathcal{S}$ required for reaching $\mathcal{A}$ and removing the remaining motifs from the sequence. This minimum number of motifs can be identified from the stable motif succession diagram (Fig. \hyperref[fig:ReductionMethod]{2}); they are the motifs after which all consequent motif choices lead to the same attractor $\mathcal{A}$.

\item[-] \textit{Step 3}: For each stable motif state $\mathcal{M}=\left(\sigma_{\mathcal{M}_1} = b_{\mathcal{M}_1}, \sigma_{\mathcal{M}_2} = b_{\mathcal{M}_2} , \ldots, \sigma_{\mathcal{M}_m} = b_{\mathcal{M}_m} \right)$ corresponding to node $v$, find the subsets of stable motif's states $O=\left\{M_i\right\}, M_i \subseteq \mathcal{M}$ that, when fixed, are enough to force the state of the whole motif into $\mathcal{M}$. At worst, there will only be one subset, which will equal the whole stable motif's state $\mathcal{M}$. If any of these subsets is fully contained in another subset, remove the larger of the subsets. In each stable motif sequence $\mathcal{S}=\left(\mathcal{M}_1, \ldots, \mathcal{M}_L\right)$, substitute every stable motif $\mathcal{M}_j$ with the subsets of the stable motif states obtained, that is, $\mathcal{S}=\left(O_1, \ldots, O_L\right)$.

\item[-] \textit{Step 4}: For each sequence $\mathcal{S}=\left(O_1, \ldots, O_L\right)$ create a set of states $\mathcal{C}$ by choosing one of the subsets of stable motif's states $M_{k_j}$ in each $O_j$ and taking their union, that is, $\mathcal{C}=M_{k_1} \cup \cdots \cup M_{k_L}, M_{k_j} \in O_j$. The network control set for attractor $\mathcal{A}$ is the set of node states $C_\mathcal{A}=\left\{\mathcal{C}_i\right\}$ obtained from all possible combinations of subsets of stable motif's states $M_{k_j}$'s for every sequence $\mathcal{S}$. To avoid any redundancy, we additionally prune $C_\mathcal{A}$ of duplicates and remove each set of node states $\mathcal{C}_i$ which is a superset of any of the other sets of node states $\mathcal{C}_j$ (i.e. $\mathcal{C}_j \subset \mathcal{C}_i$).
\end{itemize}

To justify that step 2 of the stable motif control algorithm yields a simplified sequence that leads to the same equivalence class of attractors as the original sequence, it suffices to show that a Boolean network with a stable motif succession diagram in which each sequence of stable motifs leads to the same equivalence class of attractors $\mathcal{A}^{(red)}$ has $\mathcal{A}^{(red)}$ as its only equivalence class of attractors.

\begin{myprop} \label{SimplifySequenceStep1}
Let $\mathcal{B}=(V,\Sigma,F)$ be a Boolean network in which all sequences of stable motifs obtained from the attractor-finding method lead to the same equivalence class of attractors $\mathcal{A}^{(red)}$. Then, $\mathcal{A}^{(red)}$ is the only equivalence class of attractors in $\mathcal{B}$.
\end{myprop}
\begin{proof}
By Theorem \ref{Attractorconservation}, every attractor $\mathcal{A}$ has a sequence of stable motifs for which the attractor-finding method fixes all the nodes in the $\mathcal{S}_{red}$ obtained from $\mathcal{A}$ to their fixed state in $\mathcal{A}$. Let $\mathcal{A}'$ be an attractor not in the equivalence class $\mathcal{A}^{(red)}$. Then, Theorem \ref{Attractorconservation} guarantees that there must be a sequence of stable motifs that lead to the equivalence class specified by $\mathcal{A}'$. But, this is a contradiction, since all sequences of stable motifs lead to the equivalence class of attractors $\mathcal{A}^{(red)}$. Hence, $\mathcal{A}^{(red)}$ is the only equivalence class of attractors in $\mathcal{B}$.
\end{proof}

To justify step 3 of the stable motif control algorithm we need to show that, given a Boolean network with a stable motif $\mathcal{M}$, the Boolean network obtained by fixing the state of the nodes specified in the stable motif $\mathcal{M}$ has the same attractors as the the Boolean network obtained by fixing the state of the nodes given by the subsets of $\mathcal{M}$ specified by step 3.

\begin{myprop} \label{SimplifySequenceStep2}
Let $\mathcal{B}=(V,\Sigma,F)$ be a Boolean network, let $\mathcal{M}$ be a stable motif of $\mathcal{B}$, and let $\mathcal{B}_{\mathcal{M}}$ be the Boolean network $\mathcal{B}$ with the node states specified in $\mathcal{M}$ fixed. Let $\mathcal{M}'\subset\mathcal{M}$ be a set of node states such that the network reduction of $\mathcal{B}$ using $\mathcal{M}'$ and network reduction of network reduction of $\mathcal{B}$ using $\mathcal{M}$ yield the same reduced network $\mathcal{B}_{red}$, and let $\mathcal{B}_{\mathcal{M}'}$ be the Boolean network $\mathcal{B}$ with the node states specified in $\mathcal{M}'$ fixed. Then, the attractors in $\mathcal{B}_{\mathcal{M}}$ are the same as the attractors of the Boolean network $\mathcal{B}_{\mathcal{M'}}$.
\end{myprop}
Let us sketch the proof of this proposition. Since $\mathcal{M}'\subset\mathcal{M}$ yields the same reduced network as $\mathcal{M}$ then $f_i|_{\mathcal{M}'}=b_i$ for a set of $\sigma_i$'s, which we denote $\mathcal{M}'_1$, such that $\sigma_i \in \mathcal{M}-\mathcal{M}'$. If $\mathcal{M}-\mathcal{M}'-\mathcal{M}'_1$ is not empty, we can follow the same reasoning, and find $f_i|_{\mathcal{M}'}|_{\mathcal{M}'_1}=b_i$ for a set of $\sigma_i$'s, which we denote $\mathcal{M}'_2$, such that $\sigma_i \in \mathcal{M}-\mathcal{M}'-\mathcal{M}'_1$. If $\mathcal{M}-\mathcal{M}'-\mathcal{M}'_1-\mathcal{M}'_2$ is not empty, and we follow the same reasoning, we can do this iteratively until $\mathcal{M}-\mathcal{M}'-\mathcal{M}'_1-\ldots-\mathcal{M}'_m$ is empty. Defining $\mathcal{M}'_0=\mathcal{M}'$, the result is a group of $\mathcal{M}'_i, i=0,1,\ldots,m$, which are such that
\begin{align*}
\mathcal{M}-\mathcal{M}'_0-\mathcal{M}'_1-\ldots-\mathcal{M}'_{j} =\mathcal{M}''_{j} \neq \varnothing &, \ j=0,1,\ldots,m-1 \\
\mathcal{M}-\mathcal{M}'_0-\mathcal{M}'_1-\ldots-\mathcal{M}'_{m} =\mathcal{M}''_{m} =\varnothing &,\\
f_i|_{\mathcal{M}'_0}|_{\mathcal{M}'_1} | \ldots |_{\mathcal{M}'_k} = b_i, \ \sigma_i \in \mathcal{M}'_{k+1}&, \ k=0, 1,\ldots,m-1.
\end{align*}
The Boolean network $\mathcal{B}_{\mathcal{M}'}$ has Boolean functions given by $f'_i=f_i|_{\mathcal{M}'}, \forall \sigma_i \not\in \mathcal{M}'$ and $f'_i=b_i, \forall \sigma_i \in \mathcal{M}'$. This implies that there is always transition from any network state with at least one $\sigma_j \neq b_j, \sigma_j \in \mathcal{M}'$ to a network state with $\sigma_j=b_j, \sigma_j \in \mathcal{M}'$, but not the other way around. This implies that the attractors in $\mathcal{B}_{\mathcal{M}'}$ must have $\sigma_j=b_j, \forall \sigma_j \in \mathcal{M}'$. Similarly, since $f'_i=f_i|_{\mathcal{M}'}=b_i, \forall \sigma_i \in \mathcal{M}'_1$, then the attractors in $\mathcal{B}_{\mathcal{M}'}$ must have $\sigma_j=b_j, \forall \sigma_j \in \mathcal{M}'_1$.

Since we are interested in the attractors, we can restrict ourselves to network states with $\sigma_j=b_j, \forall \sigma_j \in \mathcal{M}'\cup\mathcal{M}'_{1}$. The Boolean functions evaluated at network states with $\sigma_j=b_j \in \mathcal{M}'\cup\mathcal{M}'_{1}$ are equivalent to $f_i|_{\mathcal{M}_0'}|_{\mathcal{M}_1'}$, which implies $f_i|_{\mathcal{M}_0'}|_{\mathcal{M}_1'}=b_i, \forall \sigma_i \in \mathcal{M}'_2$ and that the attractors in $\mathcal{B}_{\mathcal{M}'}$ must have $\sigma_j=b_j, \forall \sigma_j \in \mathcal{M}'_2$. Doing this iteratively, we get that the attractors in $\mathcal{B}_{\mathcal{M}'}$ must have $\sigma_j=b_j, \forall \sigma_j \in \mathcal{M}'_0\cup\mathcal{M}'_1\cup\ldots\cup\mathcal{M}'_m\equiv\mathcal{M}$, and that the Boolean functions evaluated at the network states where the attractors can be ($\sigma_j=b_j, \forall \sigma_j \in \mathcal{M}$) are given by $f'_i=f_i|_{\mathcal{M}'}, \forall \sigma_i \not\in \mathcal{M}$ and $f'_i=b_i, \forall \sigma_i \in \mathcal{M}$. Since these Boolean functions at the network states where the attractor of $\mathcal{B}_{\mathcal{M}'}$ can be are equivalent to the Boolean functions of $\mathcal{B}_{\mathcal{M}}$, then the attractors of both must be the same.

\clearpage

\section*{Text S3. Time complexity and mitigation techniques for the attractor-finding method and the stable motif control approach} \label{sec:S3}

In this part we discuss the time complexity of our methods, the worst case scenarios, and mitigation techniques for when our method takes a prohibitively long amount of time. This discussion is based on our experience with Boolean models of intracellular networks \cite{Colomoto} and random Boolean networks \cite{KauffmanOriginal,KadanoffReview}. For a detailed description of each step in the attractor-finding method see Text \hyperref[sec:S2]{S2} and ref. \cite{ReductionChaos}. For an algorithmic description of each step in the stable motif control approach see Text \hyperref[sec:S7]{S7}.

In the following we use $V=\left( v_1, v_2, \ldots, v_N \right)$ to represent the $N$ nodes of the Boolean network, $\sigma_i, i=1, 2, \ldots, N$ to represent the state of node $v_i$, $\Sigma=\left(\sigma_1, \sigma_2, \ldots, \sigma_N \right)$ to represent the states of all nodes (also called a network state), $f_i, i=1, 2, \ldots, N$ to represent the Boolean function of node $v_i$, and $F=\left(f_1, f_2 , \ldots, f_N \right)$ to represent all the Boolean functions. We use $f(\Sigma)$ to denote a Boolean function $f$ evaluated at a network state $\Sigma$, and $f|_P$ to denote a Boolean function where only the state of a subset of nodes $P=\left\{\sigma_{p_1}, \sigma_{p_2}, \ldots, \sigma_{p_l}\right\}$ is evaluated. We commonly use $b_i$ to indicate that a specific value for node state $\sigma_i$ is chosen, that is, $\sigma_i=b_i$.

\setcounter{subsection}{0}
\subsection{Expanded network/network reduction attractor-finding method of ref. \cite{ReductionChaos}} \label{sec:S3A}

In this section we restrict ourselves to the steps of the attractor-finding method which we have found to be the most time consuming, namely, the creation of the expanded network representation and the identification of stable motifs from the expanded network representation.

\subsubsection{The expanded network representation} \label{sec:S3A1}

The expanded network representation requires each Boolean function $f$ to be written in a disjunctive normal form:
  \begin{align*}
  f =&\left(s_{1} \hbox{ AND } s_{2} \hbox{ AND } \cdots \hbox{ AND } s_{k} \right)\hbox{ OR }\left(s_{k+1} \hbox{ AND } s_{k+2} \hbox{ AND } \cdots s_{l} \right) \\
  & \hbox{ OR } \cdots  \hbox{ OR } \left(s_{m} \hbox{ AND } s_{m+1} \hbox{ AND } \cdots \hbox{ AND } s_{n} \right),
    \end{align*}
where the $s_{j}$'s are either the states of one of the input nodes of $f$, or one of these states' negations. Additionally, we require that if for $M$, denoting a state of a subset of the inputs of $f$, one has $f|_M=1$ (regardless of the states of the remaining inputs), then the disjunctive form of $f_i$ must have at least one of its conjunctive clauses equal to 1 when evaluated at the state $M$ of this subset of nodes. In logic minimization terms, this is equivalent to requiring that $f$ is written as the sum of all of its prime implicants \cite{McCluskey,Quine1,Quine2,LogicMinRev}.

The number of prime implicants of a Boolean function $f$ of $K$ inputs is known to be at most $O(3^K/\sqrt{K})$ \cite{NumPrimeImplicants1}. For a Boolean function that can be expressed as a disjunctive normal form of $m$ conjunctive clauses, the number of prime implicants is bounded by $2^m-1$ \cite{NumPrimeImplicants2}. The Boolean functions in the models we use have a relatively low number of inputs ($K\leq10$), thus the simplest algorithm (the Quine–McCluskey algorithm \cite{McCluskey,Quine1,Quine2}) is sufficient to find the disjunctive normal form. If one requires Boolean functions with a larger number of inputs, more sophisticated algorithms should be used (e.g. \cite{LogicMinRev,PrimeImplPolyn} and references within).

\subsubsection{Identifying stable motifs from the expanded network} \label{sec:S3A2}

A stable motif $M$ in the expanded network is any of the smallest strongly connected components (SCCs) in the expanded network representation which satisfy these two properties:

\begin{enumerate}
\item If $M$ contains a normal node $v_i$ (complementary node $\overline{v}_i$) then $M$ does not contain its corresponding complementary node $\overline{v}_i$ (normal node $v_i$).
\item If $M$ contains a composite node $v^{(comp)}$, then all input nodes of $v^{(comp)}$ are elements of $M$.
\end{enumerate}

The identification of stable motifs tends to be the most time-consuming part of our method. Specifically, it is the identification of stable motifs of the full network model that we have found to be the most time-consuming.

We are interested in identifying the smallest strongly connected components subject to two restrictions. Since most algorithms to identify stable connected components restrict themselves to the largest strongly connected components, we need to devise a way to enumerate all strongly connected components. To do this, we use the fact that any strongly connected component is composed of cycles. We first identify all directed cycles in the expanded network that do not contain both a normal node and its complementary node, which guarantees that property 1 of stable motifs is satisfied. For each of these cycles, we check whether it satisfies property 2; if it does, then the cycle is a stable motif.

For the cycles that do not satisfy property 2, we form all unions of cycles which share composite nodes, while discarding unions that do not satisfy property 1, until they satisfy property 2 for the composite node being considered. The result is a set of SCCs that satisfy property 1 and property 2 for at least one composite node inside the SCC. For each of these SCCs, we check whether it satisfies property 2 for all composite nodes; if it does, then this SCC is a stable motif candidate. For the SCCs that do not satisfy property 2 for all composite nodes, we repeat the same step as before, that is, we form all unions of SCCs which share composite nodes, while discarding the unions that do not satisfy property 1, until they satisfy property 2 for the composite node being considered. We do this iteratively, until only SCCs that satisfy property 1 and property 2 for all composite nodes are left. Finally, we take all SCCs obtained during the process that satisfy property 1 and 2, and leave only SCCs which do not have other SCCs as a subset, which yields the smallest SCCs that satisfy property 1 and 2, i.e., the stable motifs.

In the following we discuss the complexity of each step of the stable motif identification algorithm.

\subsubsection{Complexity of enumerating all cycles in the expanded network} \label{sec:S3A3}

To identify stable motifs, we first search the expanded network $G_{exp}=(V_{exp},E_{exp},F_{exp})$ for all directed cycles that do not contain both a normal node and its complementary node using a modified version of Johnson's cycle algorithm \cite{CycleAlgorithm}. The original Johnson's cycle algorithm for a graph $G=(V,E)$ has a time complexity $O\left((|V|+|E|)(C+1)\right)$, where $|V|$ is the number of nodes, $|E|$ is the number of edges, and $C$ is the total number of directed cycles in $G$. In our modified version of Johnson's algorithm, a normal/complementary node can only be added to the stack from which cycles are obtained if its respective complementary/normal node is not already part of the stack. Consequently, this modified version of Johnson's cycle algorithm has a similar time complexity as the original algorithm, specifically, $O\left((|V_{exp}|+|E_{exp}|)(c+1+|V_{exp}|+|E_{exp}|)\right)$, where $|V_{exp}|$ is the number of nodes in $G_{exp}$, $|E_{exp}|$ is the number of edges in $G_{exp}$, and $c$ is the number of directed cycles in $G_{exp}$ that satisfy property 1 of stable motifs. The number of directed cycles that satisfy property 1, $c$, is typically much smaller than total the number of directed cycles, $C$. The limiting factor of the algorithm is $c$, which is typically smaller than the limiting factor of a brute-force search to attractor-finding, which is limited by the size of the network state space $2^N$. For our test cases, the resulting number of cycles is $c_{TLGL}=18,241$ and $c_{Th}=63$ for the T-LGL leukemia network model and the helper T cell network model, respectively.

To illustrate how the number of directed cycles $C$ varies for different networks, we consider the worst-case scenario and the typical behavior of $C$ in terms of directed Erd\H{o}s-R\'{e}nyi graphs. Let $G=(V,E)$ be a directed network with $|V|=n$. The worst case scenario for the number of cycles is a fully connected network, for which $C$ grows as $O((n-1)!)$. For a directed Erd\H{o}s-R\'{e}nyi graph $G(|V|,p=K/n)$, where $K$ is the average degree of $G$, the probability to find a cycle of length $L$ is $(K/n)^L$, while the number of such cycles is $n!/[2L(n-L)!]$. The average number cycles of length $L$ is then
\begin{equation}
  \overline{C}_L=\frac{1}{2L}\frac{n!}{(n-L)!}\left(\frac{K}{n}\right)^L,
\end{equation}
which scales as $O(K^L)$ for small $L$ and large $n$, and as $O(K^n/\sqrt{n})$ for large $L$ and $n$. We note that even though the average number of large cycles for an Erd\H{o}s-R\'{e}nyi graph grows exponentially with $n$, this average is dominated by networks with large $L\sim n$, which have a low probability of appearing ($\sim(K/n)^n$) but have a large number of ways to have such cycles ($\sim(n-1)!/2$).

As the last paragraph illustrates, in some cases the number of directed cycles $c$ of the expanded network is too large to enumerate. For these cases, we propose using a cutoff $L_{max}$ in the modified Johnson's cycle algorithm for the maximum number of nodes in the stack from which cycles are obtained. The result of this cutoff is that only cycles with $L < L_{max}$ are output by the cycle algorithm. Unfortunately, this can result in overlooking cycles that are required to find all stable motifs, something which one should take into consideration. We should emphasize that even if a cutoff is used, we can still identify control interventions using a modified version of the stable motif control algorithm as long as the attractor of interest appears in the stable motif succession diagram obtained with the cutoff (see Text \hyperref[sec:S3]{S3} section \hyperref[sec:S3B3]{B.3}).

\subsubsection{Complexity of forming strongly connected components by the union of cycles} \label{sec:S3A4}

For each of the directed cycles in the expanded network that satisfy property 1, we check whether it satisfies property 2, and if it does, then the cycle is a stable motif. If it doesn't, we check that for each composite node $v^{(comp)}$ in the cycle, the complement of the inputs of $v^{(comp)}$ are not part of the cycle. The cycles that do not satisfy this condition are discarded, since they cannot form part of an SCC satisfying both property 1 and property 2. For the resulting set of cycles, we check whether for each composite node $v^{(comp)}$ the union of all cycles that have $v^{(comp)}$ as an element contains all inputs of $v^{(comp)}$. For the composite nodes for which this is not true, we discard all cycles containing such composite nodes, since SCCs containing such cycles cannot satisfy property 2. The result is a set of directed cycles in the expanded network that satisfy property 1 and that have the potential of being able to satisfy property 2. We call this set of directed cycles $S_{simp}$, and denote $c_{simp}=|S_{simp}|$ as the number of such cycles. For our test cases we have $c_{simp,TLGL}=670$ and $c_{simp,Th}=63$ for the T-LGL leukemia network model and the helper T cell network model, respectively.

At worse, all combinations of the directed cycles in $S_{simp}$ can form an SCC that satisfies property 1 and 2, which makes the worst-case scenario $O(2^{C_{simp}})$. In practice, we find that the restriction of satisfying both property 1 and property 2 makes the number of possible combinations much smaller.

To illustrate this, let us describe in detail the process we use to find the stable motifs from $S_{simp}$. Given a composite node $v^{(comp)}$, we form all unions of cycles in $S_{simp}$ which have the first input of $v^{(comp)}$ but do not have the second input of $v^{(comp)}$, and the cycles which have the second input of $v^{(comp)}$ but do not have the first input of $v^{(comp)}$, while discarding unions that do not satisfy property 1. The result is a group of SCCs that have both the first and second input of $v^{(comp)}$. We then form all unions of cycles/SCCs that have both the first and second input of $v^{(comp)}$ and the cycles which have the third input of $v^{(comp)}$ but do not have the first and second input of $v^{(comp)}$, while discarding the unions that do not satisfy property 1. We do this iteratively for all the inputs of $v^{(comp)}$. The result is $S_{simp,1}$, a set of SCCs that satisfy property 1 and property 2 for at least one composite node inside the SCC.

For each SCC in $S_{simp,1}$, we check whether it satisfies property 2 for all composite nodes; if it does, then this SCC is a stable motif candidate. If it does not, we check that for each composite node $v^{(comp)}$ in the SCC, the complement of the inputs of $v^{(comp)}$ are not part of the SCC. The SCCs that do not satisfy this condition are discarded from $S_{simp,1}$, since they cannot form part of an SCC satisfying both property 1 and property 2. For the resulting $S_{simp,1}$, we check whether for each composite node $v^{(comp)}$ the union of all SCCs that have $v^{(comp)}$ as an element contains all inputs of $v^{(comp)}$. For the composite nodes for which this is not true, we discard all SCCs in $S_{simp,1}$ containing such composite nodes, since SCCs containing such SCCs cannot satisfy property 2.

For the SCCs in $S_{simp,1}$, we repeat the same steps as before, that is, we form all unions of SCCs which share composite nodes, while discarding unions that do not satisfy property 1, until they satisfy property 2 for the composite node being considered. The result is a set of SCCs which we call $S_{simp,2}$, which we prune of SCCs that cannot satisfy both property 1 and property 2. We do this iteratively, until only SCCs that satisfy property 1 and property 2 for all composite nodes are left. Finally, we take all SCCs obtained during the process that satisfy property 1 and 2, and leave only SCCs which do not have other SCCs as a subset, which yields the smallest SCCs that satisfy property 1 and 2, i.e., the stable motifs.

As the worst case scenario illustrates, the number of combinations that can potentially become stable motifs can be too large to enumerate. For these cases, we propose using a cutoff $L_{max}$ for the max number of nodes allowed in an SCC when taking the unions of SCCs/cycles. Unfortunately, like in the case of the cutoff for cycles, this can result in overlooking SCCs that are stable motifs or required to find all stable motifs, something which one should take into consideration. We should emphasize that even if a cutoff is used, we can still identify control interventions using a modified version of the stable motif control algorithm as long as the attractor of interest appears in the stable motif succession diagram obtained with the cutoff (see Text \hyperref[sec:S3]{S3} section \hyperref[sec:S3B3]{B.3}).

\subsection{The stable motif control method} \label{sec:S3B}

A stable motif succession diagram can be represented as a directed graph $G_{diag}=(V_{diag},E_{diag})$ together with a dictionary $L$. The nodes $V_{diag}=\left(v_{diag,1}, v_{diag,2}, \ldots, v_{diag,n}\right)$ denote either stable motifs $\mathcal{M}_i$ (if the node has at least one outgoing edge) or attractors $\mathcal{A}_i$ (if the node has no outgoing edges). The dictionary $L$ stores the type of object (stable motif or attractor) of each node in $V_{diag}$ denotes. Each edge in $E_{diag}$ connects a stable motif with the stable motifs or attractor that can be obtained from the reduced network associated to it; if network reduction leads to a simplified network with at least one stable motif, then the edges point from the stable motif being considered to the stable motifs of the simplified network, otherwise, the edge points towards an attractor. It should be noted that the same stable motif/attractor may be assigned to more than one node in $V_{diag}$.

For completeness, we reproduce the stable motif control algorithm (see Methods and Text \hyperref[sec:S7]{S7} for more details):
\begin{itemize}
\item[-] \textit{Step 1}: Identify the sequences of stable motifs that lead to $\mathcal{A}$. These can be obtained from the stable motif succession diagram (see Fig. \hyperref[fig:ReductionMethod]{2}) by choosing the attractor of interest in the right-most part and selecting all of the attractor's predecessors in the succession diagram.

\item[-] \textit{Step 2}: Shorten each sequence $\mathcal{S} \in Sequences$ by identifying the minimum number of motifs in $\mathcal{S}$ required for reaching $\mathcal{A}$ and removing the remaining motifs from the sequence. This minimum number of motifs can be identified from the stable motif succession diagram (Fig. \hyperref[fig:ReductionMethod]{2}); they are the motifs after which all consequent motif choices lead to the same attractor $\mathcal{A}$.

\item[-] \textit{Step 3}: For each stable motif state $\mathcal{M}=\left(\sigma_{\mathcal{M}_1} = b_{\mathcal{M}_1}, \sigma_{\mathcal{M}_2} = b_{\mathcal{M}_2} , \ldots, \sigma_{\mathcal{M}_m} = b_{\mathcal{M}_m} \right)$ corresponding to node $v$, find the subsets of stable motif's states $O=\left\{M_i\right\}, M_i \subseteq \mathcal{M}$ that, when fixed, are enough to force the state of the whole motif into $\mathcal{M}$. At worst, there will only be one subset, which will equal the whole stable motif's state $\mathcal{M}$. If any of these subsets is fully contained in another subset, remove the larger of the subsets. In each stable motif sequence $\mathcal{S}=\left(\mathcal{M}_1, \ldots, \mathcal{M}_L\right)$, substitute every stable motif $\mathcal{M}_j$ with the subsets of the stable motif states obtained, that is, $\mathcal{S}=\left(O_1, \ldots, O_L\right)$.

\item[-] \textit{Step 4}: For each sequence $\mathcal{S}=\left(O_1, \ldots, O_L\right)$ create a set of states $\mathcal{C}$ by choosing one of the subsets of stable motif's states $M_{k_j}$ in each $O_j$ and taking their union, that is, $\mathcal{C}=M_{k_1} \cup \cdots \cup M_{k_L}, M_{k_j} \in O_j$. The network control set for attractor $\mathcal{A}$ is the set of node states $C_\mathcal{A}=\left\{\mathcal{C}_i\right\}$ obtained from all possible combinations of subsets of stable motif's states $M_{k_j}$'s for every sequence $\mathcal{S}$. To avoid any redundancy, we additionally prune $C_\mathcal{A}$ of duplicates and remove each set of node states $\mathcal{C}_i$ which is a superset of any of the other sets of node states $\mathcal{C}_j$ (i.e. $\mathcal{C}_j \subset \mathcal{C}_i$).
\end{itemize}

In this section we restrict ourselves to the steps of the stable motif control method which we have found to be the most time consuming, namely, identifying the sequences of stable motifs that lead to an attractor (step 1) and finding the subsets of stable motif's states that fix the state of the whole stable motif (step 3). We also consider the identification of stable motif control interventions for the case when only part of the stable motif succession diagram is known.

\subsubsection{Complexity of identifying the sequences of stable motifs that lead to an attractor} \label{sec:S3B1}

For a stable motif succession diagram $G_{diag}$ with $n_{sm}$ stable motifs, the worst case scenario in terms of the number sequences $n_{seq}$ is when all permutations of the $n_{sm}$ stable motifs form a sequence. In this case, the number of sequences $n_{seq}$ is $O(n_{sm}!)$. In practice, the number of motif we find tends to be small and/or many of the motifs are not independent from each other, in which case the effective $n_{sm}$ is much smaller. This allows us to manage even the worst case scenario as long as $n_{sm}$ or the effective $n_{sm}$ is smaller than ten. For example, $n_{sm,TLGL}=7$ and $n_{seq,TLGL}=144$ for the T-LGL leukemia network model, and $n_{sm,Th}=17$ and $n_{seq,Th}=697$ for the helper T cell network model. Note that the succession diagram of the T-LGL leukemia model shown in Fig. \hyperref[fig:ReductionTLGL]{4} does not include the part of the diagram associated with node P2, as motifs with P2 do not give rise to new motifs and do not influence the resulting attractor of any sequences in the succession diagram.

For the case of a network in which the number of sequences becomes computationally intractable, we suggest the following technique to simplify the stable motif succession diagram as the attractor-finding method is being applied. This technique takes into account which stable motifs are independent of each other and should significantly cut down the combinatorial explosion caused by allowing all possible permutations of independent motifs. Let $SM1=\left\{\mathcal{M}_{SM1,1}, \mathcal{M}_{SM1,2}, \ldots, \mathcal{M}_{SM1,n_{SM1}}\right\}$ be all the stable motifs of a Boolean network, which could be the full Boolean model or one of the reduced networks obtained during the attractor-finding method (see Text \hyperref[sec:S1]{S1} or Text \hyperref[sec:S2]{S2} for details). Let $SM2 \subset SM1$ be the motifs $\mathcal{M}\in SM1$ such that
\begin{itemize}
\item[(a)] $\mathcal{M}$ is still a stable motif after the first branching (i.e. if $\mathcal{M}, \mathcal{M}' \in SM2$ and $\mathcal{M}\neq\mathcal{M}'$, then $\mathcal{M}\rightarrow\mathcal{M}'$), and
\item[(b)] all successor motifs of $\mathcal{M}$ are still there after the second branching (i.e. if $\mathcal{M}\rightarrow\mathcal{M}'$, $\mathcal{M} \in SM2, \mathcal{M}' \not\in SM2$, then $\mathcal{M}\rightarrow\mathcal{M}''\rightarrow\mathcal{M}'$, $\forall \mathcal{M}''\in SM2$),
\end{itemize}
then we can simplify the succession diagram by compressing the separate motifs of $SM2$ into a single group of motifs.

\subsubsection{Complexity of finding the subsets of stable motif's states that fix the state of the whole stable motif} \label{sec:S3B2}

For a stable motif $\mathcal{M}$ composed of $m$ nodes and their state, step 3 of the stable motif control algorithm has the objective to find the subsets of stable motif's states $M \subseteq \mathcal{M}$ such that, when fixed, are enough to force the state of the whole motif into $\mathcal{M}$, restricted to the condition that the resulting subsets are not fully contained in another subset of the result. At worst, there is only one such subset, namely, the whole stable motif state $\mathcal{M}$.

In algorithm 5 in Text \hyperref[sec:S7]{S7} we give a procedure to find this subsets of motif's states $M \subseteq \mathcal{M}$. In principle, almost every possible combination of the $m$ node states in $\mathcal{M}$ could be obtained during the process we proposed, making the worst-case scenario have a complexity of $O(2^m)$. For our test cases, even the worst case scenario wasn't a problem, since the maximum motif size was seven for both the T-LGL leukemia network model and the helper T cell network model.

For the cases where the number of node states $m$ in $M$ is too large, we propose using a modified version of algorithm 5 where a cutoff $L_{max,1}$ is introduced for the maximum subset size considered (i.e., changing ``\textbf{for} $subsetSize\leftarrow1$ to length of list $\mathcal{M} - 1$'' to ``\textbf{for} $subsetSize\leftarrow1$ to $L_{max,1}$''). Additionally, we propose searching for subsets while starting from a subset size of $m-1$ and ending at a subset size of $L_{max,2}>L_{max,1}$ (i.e., repeat the instructions in the loop ``\textbf{for} $subsetSize\leftarrow1$ to length of list $\mathcal{M} - 1$'' just after the it ends, but starting the loop with ``\textbf{for} $subsetSize\leftarrow\mathcal{M} - 1$ to length of list $L_{max,2}$''). Unlike the other cases where a cutoff is introduced, the cutoffs $L_{max,1}$ and $L_{max,2}$ do not require a modification to the control method to guarantee the method's effectiveness (see Text \hyperref[sec:S3]{S3} section \hyperref[sec:S3B3]{B.3}), their only effect is that the control interventions obtained with these cutoffs may involve more nodes or have more redundancy than the interventions obtained without the cutoffs.

\subsubsection{Identifying stable motif control interventions with partial knowledge of the stable motif succession diagram} \label{sec:S3B3}

In Text \hyperref[sec:S2]{S2} sections \hyperref[sec:S2A3]{A.3} and \hyperref[sec:S2A4]{A.4} we considered the scenario in which the number of cycles and/or SCCs becomes computationally intractable, which prompted us to introduce a cutoff in the maximum length of cycles allowed and/or the maximum SCC size allowed. In both of these cases, the introduction of this cutoff may cause the overlooking of some stable motifs. If this is the case, the result of the attractor-finding method is not the full stable motif succession diagram, but only a part of it. Here we look at how to identify stable motif control interventions when only a part of the stable motif succession diagram is known.

Our results in Text \hyperref[sec:S2]{S2} section \hyperref[sec:S2B]{B} show that a sequence of stable motifs in the stable motif succession diagram uniquely determines an attractor, and that fixing the subset of the states determining a motif specified in step 3 has the same effect as fixing all node states in the stable motif  (see Text \hyperref[sec:S2]{S2} for more details). This implies that step 1, 3 and 4 are still applicable even if only a part of the stable motif succession diagram is known. On the other hand, step 2 requires knowing whether all consequent motif choices after a certain motif lead to the same attractor and, thus, needs the knowledge of all the stable motif decision diagram. The result is that the modified stable motif control algorithm is the same as the original algorithm except for step 2, which is skipped.

\clearpage

\section*{Text S4. Logical rules and classification of attractors in the T-LGL leukemia network model} \label{sec:S4}
\setcounter{subsection}{0}
\subsection{Logical rules of the T-LGL leukemia network model}  \label{sec:S4A}
These rules dictate the dynamics of the T-LGL leukemia survival signaling network depicted in Fig. \hyperref[fig:TLGLnetwork]{3}. For simplicity, the node states are represented by the node names. The Boolean rules were constructed based on experimental results of the corresponding intracellular components in normal and leukemic cytotoxic T cells, in such a way that that the model reproduces the known experimental behavior. The interested reader is referred to \cite{TLGLPNAS} for the detailed explanation of the rules. For transparency of interpretation we slightly diverge from \cite{TLGLPNAS} by not allowing a single transient activation of the Apoptosis node to drive cell death. For this reason these rules do not include the "AND NOT Apoptosis" clause on each node, and the auto-activation of Apoptosis that \cite{TLGLPNAS} has. This slight change, also used in \cite{AssiehPCB}, does not change the results qualitatively.
\ \\
\ \\
$f_{CTLA4}$ = TCR \\
$f_{TCR}$ = Stimuli AND NOT CTLA4 \\
$f_{PDGFR}$ = S1P OR PDGF \\
$f_{FYN}$ = TCR OR IL2RB \\
$f_{Cytoskeleton\hbox{ }signaling}$ = FYN \\
$f_{LCK}$ = CD45 OR ((TCR OR IL2RB) AND NOT ZAP70) \\
$f_{ZAP70}$ = LCK AND NOT FYN \\
$f_{GRB2}$ = IL2RB OR ZAP70 \\
$f_{PLCG1}$ = GRB2 OR PDGFR\\
$f_{RAS}$ = (GRB2 OR PLCG1) AND NOT GAP \\
$f_{GAP}$ = (RAS OR (PDGFR AND GAP)) AND NOT (IL15 OR IL2) \\
$f_{MEK}$ = RAS \\
$f_{ERK}$ = MEK AND PI3K \\
$f_{PI3K}$ = PDGFR OR RAS \\
$f_{NFKB}$ = (TPL2 OR PI3K) OR (FLIP AND TRADD AND IAP) \\
$f_{NFAT}$ = PI3K \\
$f_{RANTES}$ = NFKB \\
$f_{IL2}$ = (NFKB OR STAT3 OR NFAT) AND NOT TBET \\
$f_{IL2RBT}$ = ERK AND TBET \\
$f_{IL2RB}$ = IL2RBT AND (IL2 OR IL15) \\
$f_{IL2RAT}$ = IL2 AND (STAT3 OR NFKB) \\
$f_{IL2RA}$ = IL2 AND IL2RAT AND NOT IL2RA \\
$f_{JAK}$ = (IL2RA OR IL2RB OR RANTES OR IFNG) AND NOT (SOCS OR CD45) \\
$f_{SOCS}$ = JAK AND NOT (IL2 OR IL15) \\
$f_{STAT3}$ = JAK \\
$f_{P27}$ = STAT3 \\
$f_{Proliferation}$ = STAT3 AND NOT P27 \\
$f_{TBET}$ = JAK OR TBET \\
$f_{CREB}$ = ERK AND IFNG \\
$f_{IFNGT}$ = TBET OR STAT3 OR NFAT \\
$f_{IFNG}$ = ((IL2 OR IL15 OR Stimuli) AND IFNGT) AND NOT (SMAD OR P2) \\
$f_{P2}$ = (IFNG OR P2) AND NOT Stimuli2 \\
$f_{GZMB}$ = (CREB AND IFNG) OR TBET \\
$f_{TPL2}$ = TAX OR (PI3K AND TNF) \\
$f_{TNF}$ = NFKB \\
$f_{TRADD}$ = TNF AND NOT (IAP OR A20) \\
$f_{FasL}$ = STAT3 OR NFKB OR NFAT OR ERK\\
$f_{FasT}$ = NFKB \\
$f_{Fas}$ = FasT AND FasL AND NOT sFas  \\
$f_{sFas}$ = FasT AND S1P AND NOT Apoptosis \\
$f_{Ceramide}$ = Fas AND NOT S1P \\
$f_{DISC}$ = FasT AND ((Fas AND IL2) OR Ceramide OR (Fas AND NOT FLIP)) \\
$f_{Caspase}$ = (((TRADD OR GZMB) AND BID) AND NOT IAP) OR DISC \\
$f_{FLIP}$ = (NFKB OR (CREB AND IFNG)) AND NOT DISC \\
$f_{A20}$ = NFKB \\
$f_{BID}$ = (Caspase OR GZMB) AND NOT (BclxL OR MCL1) \\
$f_{IAP}$ = NFKB AND NOT BID \\
$f_{BclxL}$ = (NFKB OR STAT3) AND NOT (BID OR GZMB OR DISC) \\
$f_{MCL1}$ = (IL2RB AND STAT3 AND NFKB AND PI3K) AND NOT DISC \\
$f_{Apoptosis}$ = Caspase \\
$f_{GPCR}$ = S1P \\
$f_{SMAD}$ = GPCR \\
$f_{SPHK1}$ = PDGFR \\
$f_{S1P}$ = SPHK1 AND NOT Ceramide\\

\subsection{Classification of attractors in the T-LGL leukemia network model} \label{sec:ClassTLGL} 

To classify the attractors in the T-LGL leukemia network we use the state of the node Apoptosis; ON for apoptosis and OFF for T-LGL leukemia. This is the same criterion used by Saadatpour et al. \cite{AssiehPCB}. This criterion groups several attractors into the T-LGL leukemia attractor class and several others into the apoptosis attractor class. Thus, stable motif blocking is not successful by default.

The attractor states classified as T-LGL leukemia attractors differ from one another in the activity of some nodes (e.g. IL2RB, IL2RBT, IL2, and IL2RA), but most of them are characterized by the inhibition of Fas-induced apoptosis pathway elements (e.g. Caspase=OFF, DISC=OFF, TRADD=OFF, Fas=OFF, FasL=ON, FasT=ON and Ceramide=OFF), and the activation of transcription factors (e.g. NFKB=ON, TPL2=ON and IFNGT=ON), receptors (e.g. PDGFR=ON and GPCR=ON), or kinases (e.g. S1P=ON, SPHK1=ON, and PI3K=ON). The attractor states classified as Apoptosis attractors are characterized by the activation of Caspase (Caspase=ON) by Fas-induced apoptosis pathway elements such as DISC=ON, Ceramide=ON, Fas=ON, IAP=OFF, GZMB=ON, and BID=ON.

\clearpage

\section*{Text S5. Logical rules, classification of attractors, and analysis of the stable motif decision diagram in the helper T cell differentiation network model} \label{sec:S5}
\setcounter{subsection}{0}
\subsection{Logical rules of the helper T cell differentiation network model developed by Naldi et al. \cite{ThCellDifferentiation}} \label{sec:S5A}
For our study we use one of environmental conditions studied by Naldi et al. \cite{ThCellDifferentiation}, namely, the presence of antigen presenting cells (APC=ON), external TGF$\beta$ (TGFB\_e=ON), external IL2 (IL2\_e=ON), and the absence of other external cytokines (IFNB\_e=OFF, IFNG\_e=OFF, IL4\_e=OFF, IL6\_e=OFF, IL10\_e=OFF, IL12\_e=OFF, IL15\_e=OFF, IL21\_e =OFF, IL23\_e=OFF, and IL27\_e=OFF). The helper T cell differentiation network with only the considered input signals is shown in Fig. \hyperref[fig:Thnetwork]{5}. For simplicity, the node states are represented by the node names. For the nodes that have three states (0, 1 and 2), we created an extra node (denoted by Nodename\_2) that represents the third state (2) and adapted the rules accordingly. The interested reader is referred to the work of Naldi et al. \cite{ThCellDifferentiation} for a detailed justification of the logical rules.
\ \\
\ \\
$f_{CD28}$ = APC \\
$f_{IFNBR}$ = IFNB\_e \\
$f_{IFNGR}$ = IFNGR1 AND IFNGR2 AND (IFNG OR IFNG\_e) \\
$f_{IL2R}$ = CGC AND IL2RB AND (NOT IL2RA) AND (IL2 OR IL2\_e) \\
$f_{IL2R\_2}$ =CGC AND IL2RB AND IL2RA AND (IL2 OR IL2\_e) \\
$f_{IL4R}$ = CGC AND NOT IL4RA AND (IL4 OR IL4\_e) \\
$f_{IL4R\_2}$ =CGC AND IL4RA AND (IL4 OR IL4\_e) \\
$f_{IL6R}$= GP130 AND IL6RA AND IL6\_e \\
$f_{IL10R}$ = IL10RA AND IL10RB AND (IL10 OR IL10\_e) \\
$f_{IL12R}$ = IL12RB1 AND IL12RB2 AND IL12\_e \\
$f_{IL15R}$ = CGC AND IL15RA AND IL2RB AND IL15\_e \\
$f_{IL21R}$ = GP130 AND CGC AND (IL21 OR IL21\_e) \\
$f_{IL23R}$ = GP130 AND (IL23 OR IL23\_e) AND STAT3 AND RORGT \\
$f_{IL27R}$ = GP130 AND IL27RA AND IL27\_e \\
$f_{TCR}$ = APC \\
$f_{TGFBR}$ = TGFB OR TGFB\_e \\
$f_{IL12RB1}$ = IRF1 \\
$f_{IL4RA}$ = STAT5\_2 \\
$f_{IL2RA}$ =(SMAD3 OR FOXP3 OR (STAT5 OR STAT5\_2) OR NFKB) AND NFAT \\
$f_{IFNG}$ = Proliferation AND NOT FOXP3 AND NOT STAT3 AND NFAT AND ((TBET AND RUNX3) OR STAT4) \\
$f_{IL2}$ = ((NFAT AND NOT FOXP3) OR NFKB) AND (NOT (STAT5 OR STAT5\_2) OR NOT STAT6) AND (NOT NFKB OR NOT TBET) \\
$f_{IL4}$ = NFAT AND Proliferation AND GATA3 AND NOT FOXP3 AND NOT ((TBET AND RUNX3) OR IRF1) \\
$f_{IL10}$ = (GATA3 OR STAT3) AND NFAT AND Proliferation \\
$f_{IL21}$ = NFAT AND Proliferation AND STAT3 \\
$f_{IL23}$ = NFAT AND Proliferation AND STAT3 \\
$f_{TGFB}$ = NFAT AND Proliferation AND FOXP3 \\
$f_{TBET}$ = (TBET OR STAT1) AND NOT GATA3 \\
$f_{GATA3}$ = (GATA3 OR STAT6) AND NOT TBET \\
$f_{FOXP3}$ = (STAT5 OR STAT5\_2) AND NFAT AND (FOXP3 OR (SMAD3 AND NOT STAT1 AND NOT (STAT3 AND RORGT))) \\
$f_{NFAT}$ = CD28 AND TCR \\
$f_{STAT1}$ = IFNBR OR IFNGR OR IL27R \\
$f_{STAT3}$ = IL6R OR IL10R OR IL21R OR IL23R OR IL27R \\
$f_{STAT4}$ = IL12R AND NOT GATA3 \\
$f_{STAT5}$ = (IL4R OR IL2R OR IL15R) AND NOT IL2R\_2 AND NOT IL4R\_2 \\
$f_{STAT5\_2}$ = IL2R\_2 OR IL4R\_2 \\
$f_{STAT6}$ = IL4R OR IL4R\_2 \\
$f_{SMAD3}$ = TGFBR \\
$f_{IRF1}$ = STAT1 \\
$f_{RUNX3}$ = TBET \\
$f_{Proliferation}$ = STAT5\_2 OR Proliferation \\
$f_{NFKB}$ = NOT IKB AND NOT FOXP3 \\
$f_{IKB}$ = NOT TCR \\
$f_{RORGT}$ = (TGFBR AND STAT3) OR (RORGT AND (TGFBR OR STAT3)) \\
$f_{IL17}$ = NFAT AND Proliferation AND RORGT AND NOT FOXP3 AND NFKB AND STAT3 AND NOT ((STAT5 OR STAT5\_2) OR STAT1 OR STAT6) \\
$f_{IFNGR1}$ = ON \\
$f_{IFNGR2}$ = ON \\
$f_{GP130}$ = ON \\
$f_{IL6RA}$ = ON \\
$f_{CGC}$ = ON \\
$f_{IL12RB2}$ = ON \\
$f_{IL10RB}$ = ON \\
$f_{IL10RA}$ = ON \\
$f_{IL15RA}$ = ON \\
$f_{IL2RB}$ = ON \\
$f_{IL27RA}$ = ON \\

\subsection{Classification of attractors in the helper T cell differentiation network model} \label{sec:ClassTh} 

To classify the attractors in the helper T cell differentiation network we use the same criteria used by Naldi et al. \cite{ThCellDifferentiation}: TBET=ON for Th1; TBET=OFF and GATA3=ON for Th2; TBET=OFF, GATA3=OFF, and FOXP3=ON for Treg; and TBET=OFF, GATA3=OFF, FOXP3=OFF, and RORGT=ON for Th17. These criteria group several attractors into each attractor class. Consequently, stable motif blocking is not successful by default.

As explained in the work by Naldi et al., the attractor states in an attractor class share the expression of many nodes apart from their master regulator (TBET, GATA3, FOXP3, and/or RORGT), but can also be very similar to attractor states of other attractor classes. This gives rise to hybrid cells types co-expressing markers of more than one canonical cell type which are classified according to the above criteria. These criteria are used because they are consistent with the differentiation of naive helper T cells into the different helper T cell subtypes under the appropriate environmental signals.

\subsection{Analysis of the stable motif decision diagram for helper T cell differentiation network} \label{sec:S5C}

Using the attractor-finding method on the helper T cell differentiation network we obtain 17 stable motifs and a stable motif decision diagram composed of 697 sequences. Despite the size of the decision diagram, a closer look at it suggests a simple explanation: the stable motifs associated with each attractor regulate the characteristic transcription factor of each helper T cell subtype. To check this, we look at the minimal subsets of stable motifs that are sufficient for a sequence to lead to a single differentiated helper T cell subtype. Each subset is minimal because removing any stable motif allows sequences with that subset to lead to more than one helper T cell subtype.

The minimal subsets of stable motifs associated to each helper T cell subtype are shown in Fig. \hyperref[fig:ReductionTh]{6}. Most of these subsets contain a motif with the defining transcription factor of each helper T cell subtype: TBET=ON for Th1, GATA3=ON for Th2, RORGT=ON for Th17, and FOXP3=ON for Treg. The motif subsets that do not contain a subtype's characteristic transcription factor can be shown to depend on the stabilization of a motif that includes this transcription factor (e.g., the motif IFNGR=IFNG=STAT=ON, FOXP3=OFF in Fig. \hyperref[fig:ReductionTh]{6}(a) requires a stable motif with TBET=ON to stabilize) or that causes the differential expression of this transcription factor (e.g., the stable motif subsets in Fig. \hyperref[fig:ReductionTh]{6}(c) are sufficient to cause RORGT=ON and a Th17 subtype when taken together with the motifs they depend on, despite RORGT not being part of any of the motifs). This shows that the minimal subsets of stable motifs in the decision diagram regulate the characteristic transcription factor of each helper T cell subtype.

\clearpage

\section*{Text S6. Translating the logical network models into ordinary differential equation models, and intervention target validation for the ordinary differential equation models} \label{sec:S6}
\setcounter{subsection}{0}
\subsection{Translating the logical network models into ordinary differential equation models} \label{sec:S6A}

We use the method described by Wittmann et al. \cite{BooleantoODE} and its MATLAB implementation \cite{Odefy} to translate the studied Boolean network models into ordinary differential equation (ODE) models. In the ODE models, the node state variables $\widetilde\sigma_i$ take continuous values in the range $\left[0,1\right]$ and their time evolution is given by
\begin{equation} \label{eq:ODE}
  \frac{d\widetilde\sigma_i}{dt} = \frac{1}{\tau_i} \left[\widetilde{f}_i(\widetilde\sigma_{i_1}, \widetilde\sigma_{i_2}, \dots, \widetilde\sigma_{i_{k_i}})-\widetilde\sigma_i\right], \ i=1, 2, \ldots, N,
\end{equation}
where $\widetilde\sigma_{i_1}, \dots, \widetilde\sigma_{i_{k_i}}$ are the state variables of the inputs of node $i$, $\widetilde{f}_i$ is a smooth Hill-type function parameterized by a set of Hill coefficient $\{n_{i_1}, \ldots, n_{i_{k_i}}\}$ and threshold parameters $\{\theta_{i_1}, \ldots, \theta_{i_{k_i}}\}$ for each input, and $\tau_i$ is a time-scale parameter. The function $\widetilde{f}_i$, a so-called normalized Hillcube \cite{BooleantoODE,Odefy}, is a continuous analogue of the Boolean function $f_i$ and is such that it matches $f_i$ whenever all its input variables are either 0 or 1.

This type of logical-to-ODE model conversion is far from trivial, as evidenced by the great number of studies in this topic \cite{GlassKauffman,ThomasReview,deJong} and the fact that some aspects of this type of conversion are still not fully understood \cite{Snoussi,Casey,Farcot,ChavesChaos}. For example, even though the fixed point attractors of the Boolean model are guaranteed to be preserved in the ODE model, this is not necessarily the case for complex attractors. The ODE model may also have attractors that have no Boolean equivalent. Some of these not fully understood aspects, namely, the appearance of attractors of the ODE system with no Boolean equivalent and the mapping of Boolean complex attractors to the ODE system, play a role on our test cases; for example, the T-LGL leukemia network model has an ODE attractor with no Boolean equivalent, and the T-LGL leukemia attractor is a complex attractor (albeit with only two oscillating nodes). These types of ODE attractors can still be classified using the criteria explained in subsection \hyperref[sec:ClassTLGL]{B} of Text \hyperref[sec:S4]{S4} and subsection \hyperref[sec:ClassTh]{B} of Text \hyperref[sec:S5]{S5}. Despite this, we find that the control interventions are remarkably effective in the ODE version of the model, though not always 100\% effective.

\subsection{Intervention target validation for the ordinary differential equation models} \label{sec:S6B}
To validate an intervention target in the ordinary differential equation model, we fix the node states in the continuous equivalent of the states in the logical model interventions (1 for ON and 0 for OFF), choose a random uniformly chosen initial condition in the continuous interval $[0,1]^N$, and evolve the system using Eq. \ref{eq:ODE}. The system of ordinary differential equations is solved using MATLAB's ode45 function, based on the Runge-Kutta method by Dormand and Prince \cite{RungeKutta}. The error tolerances in the ode45 function are chosen between $10^{-2}$ and $10^{-3}$, while the rest of the function's parameters are left at their default value.

Each initial condition is evolved from $t=0$ to $t=15$ or until it reaches an attractor. We repeat this for a large number of initial conditions (25,000, unless otherwise specified) and calculate the probability of reaching each attractor from a uniformly chosen initial condition. For the case when the intervention is not permanent, we evolve the system with the intervention from $t=0$ to $t=15$, remove the intervention, and evolve the system from $t=15$ to $t=30$. The number of initial conditions we use is enough to estimate the probabilities $p_{Attr}$ of reaching the attractor of interest with an error (standard deviation of the estimated probability $p_{Attr}$) of $6\cdot10^{-3}\left[p_{Attr} (1-p_{Attr})\right]^{1/2}$. Equivalently, if $p_{Attr}$ is expressed as a percentage (which we denote as  $\%p_{Attr}$ for clarity), the error in it is estimated as $6\cdot10^{-3}\left[\%p_{Attr}(100\%-\%p_{Attr})\right]^{1/2}\%$ (e.g. 0.06\% for a $\%p_{Attr}$ of 1\%, and 0.3\% for a $\%p_{Attr}$ of 50\%). The number of time steps we use is enough to show no changes in $p_{Attr}$ beyond what is expected from the standard deviation of the estimated probability $p_{Attr}$, and is also found to be enough for the initial conditions to reach the attractors when no interventions are applied.

For the results in Tables S3 and S4 we use the default parameters: a Hill coefficient of $n=3$, a threshold parameter of $\theta=0.5$, and a time-scale parameter $\tau=1$ for all nodes. For Table S5 we use the Hill coefficient specified in the table ($n=1, 1.5, 2,$ or $2.5$), a threshold parameter of $\theta=0.5$, and a time-scale parameter $\tau=1$ for all nodes. For Table S7 we fix the intervention at the values specified (0.9, 0.8, 0.7 or 0.6 if the Boolean intervention was 1, or 0.1, 0.2, 0.3, or 0.4 if the Boolean intervention was 0), and use a Hill coefficient of $n=3$, a threshold parameter of $\theta=0.5$, and a time-scale parameter $\tau=1$ for all nodes. Finally, for the results in Table S6 we use 1,000 initial conditions and 40 different networks in which the Hill coefficients take the values specified in the table ($n=1, 1.5, 2, 2.5,$ or $3$), and the thresholds $\theta_i$ and time-scale parameters $\tau_i$ for each node are chosen uniformly from the interval $[0, 1]$.
\clearpage

\section*{Text S7. Pseudocode for the stable motif control algorithm and the stable motif blocking algorithm} \label{sec:S7}
\setcounter{subsection}{0}
For the pseudocodes we assume that one starts with a target attractor $\mathcal{A}$, the logical functions $F=\left(f_1, f_2, \ldots, f_N\right)$ for the logical network model of interest, and the stable motif succession diagram for the logical network model of interest (see Fig. \hyperref[fig:ReductionMethod]{2}). A stable motif succession diagram can be represented as a directed graph $G_{diag}=(V_{diag},E_{diag})$ together with a dictionary $L$. The nodes $V_{diag}=\left(v_{diag,1}, v_{diag,2}, \ldots, v_{diag,n}\right)$ denote either stable motifs $\mathcal{M}_i$ (if the node has at least one outgoing edge) or attractors $\mathcal{A}_i$ (if the node has no outgoing edges). The dictionary $L$ stores the type of object (stable motif or attractor) each node in $V_{diag}$ denotes. Each edge in $E_{diag}$ connects a stable motif with the stable motifs or attractors that can be obtained from the reduced network associated to it; if network reduction leads to a simplified network with at least one stable motif, then the edges points from the stable motif being considered to the stable motifs of the simplified network, otherwise, the edges point towards an attractor. It should be noted that stable motifs/attractors may be assigned to more than one node in $V_{diag}$. For example, in Fig. \hyperref[fig:ReductionMethod]{2} there are three nodes that denote the motif $\{A=0\}$, and two nodes that denote the attractor $\mathcal{A}_2$.

\subsection{Pseudocode for the stable motif control algorithm} \label{sec:S7A}

\textit{Step 1}: Identify the sequences of stable motifs that lead to $\mathcal{A}$. These can be obtained from the stable motif succession diagram (see Fig. \hyperref[fig:ReductionMethod]{2}) by choosing the attractor of interest in the right-most part and selecting all of the attractor's predecessors in the succession diagram. The stable motif diagram is represented by the directed graph $G_{diag}=(V_{diag},E_{diag})$ together with the list $L$.

\begin{pseudocode}[ruled]{GetSequences}{G,L,\mathcal{A}}
\label{algorithm1}
\mbox{\textbf{comment:} Sequences, SequencesLeft, and NewSequences are sets.} \\
\phantom{\mbox{\textbf{comment:} }}\mbox{$\mathcal{S}$ is a sequence (ordered list).} \\
Sequences \GETS \mbox{empty set}\\
SequencesLeft \GETS \mbox{empty set}\\
\FOREACH v \in \mbox{sink nodes of $G$} \DO
    \BEGIN
        \COMMENT{$L(v)$ gives the motif or attractor denoted by $v$.} \\
        \IF \mbox{$L(v)$ equals $\mathcal{A}$} \THEN
        \BEGIN
        \mathcal{S} \GETS \mbox{empty sequence} \\
        \mbox{add $v$ to the beginning of $\mathcal{S}$} \\
        \mbox{add $\mathcal{S}$ to $SequencesLeft$}
        \END
    \END
    \\
\REPEAT
\BEGIN
NewSequences \GETS \mbox{empty set} \\
\FOREACH \mathcal{S} \in SequencesLeft \DO
\BEGIN
v \GETS \mbox{first item of $\mathcal{S}$} \\
\IF \mbox{$v$ has input nodes}
    \THEN
    \BEGIN
        \FOREACH v' \in \mbox{input nodes of $v$} \DO
        \BEGIN
        \mathcal{S}' \GETS \mbox{copy $\mathcal{S}$} \\
        \mbox{add $v'$ to the beginning of $\mathcal{S}'$}\\
        \mbox{add $\mathcal{S}'$ to $NewSequences$}
        \END
    \END
\ELSE
     \mbox{add $\mathcal{S}$ to $Sequences$}\\
\mbox{remove $\mathcal{S}$ from $SequencesLeft$} \\
\END
\\
\FOREACH \mathcal{S}' \in NewSequences \DO
    \mbox{add $\mathcal{S}'$ to $SequencesLeft$}
\END
\UNTIL \mbox{$NewSequences$ is empty} \\
\RETURN{Sequences}
\end{pseudocode}

\clearpage

\textit{Step 2}: Shorten each sequence $\mathcal{S} \in Sequences$ by identifying the minimum number of motifs in $\mathcal{S}$ required for reaching $\mathcal{A}$ and removing the remaining motifs from the sequence. This minimum number of motifs can be identified from the stable motif succession diagram (Fig. \hyperref[fig:ReductionMethod]{2}); they are the motifs after which all consequent motif choices lead to the same attractor $\mathcal{A}$.

\begin{pseudocode}[ruled]{ShortenSequences1}{G,L,\mathcal{A},Sequences}
\mbox{\textbf{comment:} $ShortenedSequences1$ is a set.} \\
\phantom{\mbox{\textbf{comment:} }}\mbox{$\mathcal{S}'$ is a sequence (ordered list).}\\
\phantom{\mbox{\textbf{comment:} }}\mbox{$pathFound$ is a Boolean variable}\\
ShortenedSequences1 \GETS \mbox{empty set} \\
\FOREACH \mathcal{S} \in Sequences \DO
    \BEGIN
        \mathcal{S}' \GETS \mbox{copy $\mathcal{S}$}\\
        \FOR \textbf{each } v \in \mbox{$\mathcal{S}$ in reverse order} \DO
            \BEGIN
                pathFound \GETS \FALSE \\
                \FOR v' \in \mbox{sink nodes of $G$} \DO
                \BEGIN
                    \COMMENT{$L(v')$ gives the motif or attractor denoted by $v'$.} \\
                    \IF L(v') \mbox{ is not } \mathcal{A} \THEN
                        \BEGIN
                        \IF \mbox{there exists a directed path from $v$ to $v'$} \THEN
                            \BEGIN
                            pathFound \GETS \TRUE  \\
                            \mbox{exit \textbf{for} loop}
                            \END
                        \END
                \END
                \\
                \IF pathFound
                \THEN \mbox{exit \textbf{for} loop}
                \ELSE \mbox{remove $v$ from $\mathcal{S}'$}
            \END
            \\
            \IF \mbox{ShortenedSequences1 does not contain $\mathcal{S}'$} \THEN
                    \mbox{add $\mathcal{S}'$ to $ShortenedSequences1$}
    \END
    \\
\RETURN{ShortenedSequences1}
\end{pseudocode}

\clearpage

\textit{Step 3}: For each stable motif state $\mathcal{M}=\left(\sigma_{\mathcal{M}_1} = b_{\mathcal{M}_1}, \sigma_{\mathcal{M}_2} = b_{\mathcal{M}_2} , \ldots, \sigma_{\mathcal{M}_m} \right)$ corresponding to node $v$, find the subsets of stable motif's states $O=\left\{M_i\right\}, M_i \subseteq \mathcal{M}$ that, when fixed, are enough to force the state of the whole motif into $\mathcal{M}$. At worst, there will only be one subset, which will equal the whole stable motif state $\mathcal{M}$. If any of these subsets is fully contained in another subset, remove the larger of the subsets. In each stable motif sequence $\mathcal{S}=\left(\mathcal{M}_1, \ldots, \mathcal{M}_L\right)$, substitute every stable motif $\mathcal{M}_j$ with the subsets of the stable motif states obtained, that is, $\mathcal{S}=\left(O_1, \ldots, O_L\right)$.

\begin{pseudocode}[ruled]{SequencesWithMotifControlSets}{ShortenedSequences1,SequenceDictionary,F,L}
\mbox{\textbf{comment:} $F=\left(f_1, f_2, \ldots, f_N\right)$ contains the Boolean functions of the logical model.} \\
\phantom{\COMMENT{}}\mbox{$ShortenedSequences2$ is a set.} \\
\phantom{\COMMENT{}}\mbox{$O$ and $Subsequence$ are sequences (ordered lists).} \\
ShortenedSequences2 \GETS \mbox{empty set} \\
\FOREACH \mathcal{S} \in ShortenedSequences1 \DO
\BEGIN
    \mbox{\textbf{comment:} $index$ is an integer. It stores the index of the first element of $\mathcal{S}'$}\\
    \phantom{\mbox{\textbf{comment:} }}\mbox{that will be visited in the \textbf{for} loop below.}\\
    \phantom{\mbox{\textbf{comment:} }}\mbox{$\mathcal{S}'$ and $\mathcal{S}''$ are sequences (ordered lists).}\\
    \phantom{\mbox{\textbf{comment:} }}\mbox{$F'$ is a sequence (ordered list) of Boolean functions.}\\
    index \GETS 0 \\
    \mathcal{S}' \GETS \mbox{sequence assigned to $\mathcal{S}$ in $SequenceDictionary$}\\
    \mathcal{S}'' \GETS \mbox{empty sequence}{}\\
    F' \GETS \mbox{copy $F$} \\
    \FOREACH v \in \mathcal{S} \DO
        \BEGIN
        \mbox{\textbf{comment:} $\mathcal{S}'$ has more motifs than $\mathcal{S}$,}\\
        \phantom{\mbox{\textbf{comment:} }}\mbox{we need the extra motifs to find the reduced network from which the motif}\\
        \phantom{\mbox{\textbf{comment:} }}\mbox{$L(v)$ was obtained. These extra motifs are stored in $Subsequence$} \\
        Subsequence \GETS \mbox{empty sequence} \\
        \FOR i \GETS index \TO \mbox{length of list $\mathcal{S}'$}-1 \DO
            \BEGIN
            v' \GETS \mbox{get element of $\mathcal{S}'$ in position $i$} \\
            \IF \mbox{$v'$ equals $v$} \THEN
                \BEGIN
                index \GETS i+1\\
                \mbox{exit \textbf{for} loop}
                \END
                \\
                \mbox{add $v'$ to the end of $Subsequence$}\\
            \END
            \\
            \mbox{\textbf{comment:} $\CALL{DownstreamEffect}{L(v'),F'}$ is described in Algorithm \ref{algorithm5}.}\\
            \phantom{\mbox{\textbf{comment:} }}\mbox{$\CALL{DownstreamEffect}{L(v'),F'}$ evaluates the states of motif $L(v')$ into $F'$.}\\
            \phantom{\mbox{\textbf{comment:} }}\mbox{If any $f \in F'$  becomes a constant Boolean function after the evaluation,}\\
            \phantom{\mbox{\textbf{comment:} }}\mbox{it evaluates the resulting Boolean state of the node corresponding to}\\
            \phantom{\mbox{\textbf{comment:} }}\mbox{$f$ in every $F'$. This is done iteratively until no new constant Boolean} \\
            \phantom{\mbox{\textbf{comment:} }}\mbox{functions are found, at which point the resulting $F'$ is returned.}\\
            \FOREACH v' \in Subsequence
                \DO F' \GETS \CALL{DownstreamEffect}{L(v'),F'}\\
            \mbox{\textbf{comment:} $\CALL{MotifControlSet}{L(v),F'}$ is described in Algorithm \ref{algorithm6}}\\
            \phantom{\mbox{\textbf{comment:} }}\mbox{$\CALL{MotifControlSet}{L(v),F'}$ finds the subsets of stable motif's states}\\
            \phantom{\mbox{\textbf{comment:} }}\mbox{of $L(v)$ that, when fixed, are enough to force the state of the whole motif}\\
            \phantom{\mbox{\textbf{comment:} }}\mbox{into $L(v)$.}\\
            O \GETS \CALL{MotifControlSet}{L(v),F'} \\
            \mbox{add $O$ to end of $\mathcal{S}''$}\\
            F' \GETS \CALL{DownstreamEffect}{L(v),F'}\\
        \END
    \\
    \mbox{add $S''$ to $ShortenedSequences2$}
\END
\\
\RETURN{ShortenedSequences2}
\end{pseudocode}

\clearpage

\begin{pseudocode}[ruled]{DownstreamEffect}{\mathcal{M},F'}
\label{algorithm5}
\mbox{\textbf{comment:} $\CALL{DownstreamEffect}{\mathcal{M},F'}$ evaluates the states of motif $\mathcal{M}$ into $F'$.}\\
\phantom{\mbox{\textbf{comment:} }}\mbox{If any $f \in F'$  becomes a constant Boolean function after the evaluation,}\\
\phantom{\mbox{\textbf{comment:} }}\mbox{it evaluates the resulting Boolean state of the node corresponding to}\\
\phantom{\mbox{\textbf{comment:} }}\mbox{$f$ in every $F'$. This is done iteratively until no new constant Boolean} \\
\phantom{\mbox{\textbf{comment:} }}\mbox{functions are found, at which point the resulting $F'$ is returned.}\\
\phantom{\mbox{\textbf{comment:} }}\mbox{$M'$ and $M''$ are sets containing nodes in the logical model together with a}\\
\phantom{\mbox{\textbf{comment:} }}\mbox{Boolean variable with their state.}\\
\phantom{\mbox{\textbf{comment:} }}\mbox{$F''$ is a sequence (ordered lists) of Boolean functions.}\\
M' \GETS \mbox{empty set};M'' \GETS \mbox{copy $M$};F'' \GETS \mbox{copy $F'$}\\
\REPEAT
    \BEGIN
        \FOREACH f \in F'' \DO
            \BEGIN
                \IF \mbox{$f$ is not a constant Boolean function} \THEN
                \BEGIN
                    f \GETS \mbox{$f$ with the states in $M'$ evaluated on it} \\
                    \IF \mbox{$f$ is a constant Boolean function} \THEN
                    \BEGIN
                        \mbox{\textbf{comment:} $\sigma$ is a node in the logical model together}\\
                        \phantom{\mbox{\textbf{comment:} }}\mbox{with a Boolean variable with its state.}\\
                        \sigma \GETS \mbox{node in the logical model whose function is $f$ and}\\
                        \phantom{\sigma \GETS }\mbox{ the value of $f$ as its state.}\\
                        \mbox{add $\sigma$ to $M''$}
                    \END
                \END
            \END
            \\
            \mbox{$M' \leftarrow$ copy $M''$}\\
            \mbox{$M'' \leftarrow$ empty set}\\
    \END
\UNTIL \mbox{$M'$ is empty}\\
\RETURN{F''}
\end{pseudocode}

\begin{pseudocode}[ruled]{MotifControlSet}{\mathcal{M},F'}
\label{algorithm6}
\mbox{\textbf{comment:} $\CALL{MotifControlSet}{\mathcal{M},F'}$ finds the subsets of stable motif's states of $\mathcal{M}$ that,}\\
\phantom{\mbox{\textbf{comment:} }}\mbox{when fixed, are enough to force the state of the whole motif into $\mathcal{M}$.}\\
\phantom{\mbox{\textbf{comment:} }}\mbox{$F'$ and $F''$ are sequences (ordered lists) of Boolean functions.}\\
\phantom{\mbox{\textbf{comment:} }}\mbox{$F'$ are the logical functions of the nodes in the model whose states are specified in $\mathcal{M}$.}\\
\phantom{\mbox{\textbf{comment:} }}\mbox{$O$ is a sequence (ordered list).} \\
\phantom{\mbox{\textbf{comment:} }}\mbox{$isMotifControlSet$ and $validSubset$ are Boolean variables.}\\
O \GETS \mbox{empty sequence} \\
\FOR subsetSize \GETS 1 \TO \mbox{length of list $\mathcal{M}$} -1 \DO
\BEGIN
    \FOREACH M \in \mbox{subsets of size $subsetSize$ in $\mathcal{M}$} \DO
        \BEGIN
        validSubset \GETS \TRUE \\
        \FOREACH M' \in O \DO
            \BEGIN
                \IF \mbox{$M'$ is a subset of $M$} \THEN
                \BEGIN
                validSubset \GETS \FALSE \\
                \mbox{exit \textbf{for} loop}
                \END
            \END
        \\
        \IF \NOT validSubset \THEN
        \mbox{exit \textbf{for} loop} \\
        \mbox{\textbf{comment:} $\CALL{DownstreamEffect}{\mathcal{M},F'}$ is described in Algorithm \ref{algorithm5}.}\\
        F'' \GETS \CALL{DownstreamEffect}{\mathcal{M},F'}\\
        isMotifControlSet \GETS \TRUE \\
        \FOREACH f \in F'' \DO
        \BEGIN
            \IF \mbox{$f$ is not a constant Boolean function} \THEN
            \BEGIN
                isMotifControlSet \GETS \FALSE \\
                \mbox{exit $\FOR$ loop}
            \END
        \END
        \\
        \IF isMotifControlSet
        \THEN \mbox{add $M$ to $O$}
        \END
\END
\\
\IF \mbox{O is empty} \THEN
\mbox{add $\mathcal{M}$ to $O$}\\
\RETURN{O}
\end{pseudocode}

\clearpage

\textit{Step 4}: For each sequence $\mathcal{S}=\left(O_1, \ldots, O_L\right)$ create a set of states $\mathcal{C}$ by choosing one of the subsets of stable motif states $M_{k_j}$ in each $O_j$ and taking their union, that is, $\mathcal{C}=M_{k_1} \cup \cdots \cup M_{k_L}, M_{k_j} \in O_j$. The network control set for attractor $\mathcal{A}$ is the set of states $C_\mathcal{A}=\left\{\mathcal{C}_i\right\}$ obtained from all possible combinations of $M_{k_j}$'s for every sequence $\mathcal{S}$. To avoid any redundancy, we additionally prune $C_\mathcal{A}$ of duplicates and remove the states $\mathcal{C}_i$ which are supersets of any of the other states $\mathcal{C}_j$ (i.e. $\mathcal{C}_j \subset \mathcal{C}_i$).

\begin{pseudocode}[ruled]{StableMotifControlSets}{ShortenedSequences2}
\mbox{\textbf{comment:} $ControlSets$, $ControlSet$, and $M$ are sets} \\
\phantom{\mbox{\textbf{comment:} }}\mbox{$O$ is a sequence (ordered list).} \\
\phantom{\mbox{\textbf{comment:} }}\mbox{$L$ and $index$ are integers.} \\
\phantom{\mbox{\textbf{comment:} }}\mbox{$countArray$ and $countArrayMax$ are arrays of integers.}\\
ControlSets \GETS \mbox{empty set} \\
\FOREACH \mathcal{S} \in ShortenedSequences2 \DO
\BEGIN
    L \GETS \mbox{length of list $\mathcal{S}$} \\
    \mbox{\textbf{comment:} $countArray$ and $countArrayMax$ keep track of the combinations of motifs}\\
    \phantom{\mbox{\textbf{comment:} }}\mbox{in $\mathcal{S}$ that we have tried and that we have left.}\\
    countArray \GETS \mbox{array of integers of length $L$}\\
    countArrayMax \GETS \mbox{array of integers of length $L$}\\
    \FOR i \GETS 0 \TO L-1 \DO
        \BEGIN
        O \GETS \mbox{get element of $\mathcal{S}$ in position $i$}\\
        countArrayMax[i] \GETS \mbox{length of list $O$}\\
        countArray[i] \GETS 0
        \END
    \\
    \REPEAT
    \BEGIN
        ControlSet \GETS \mbox{empty set} \\
        \FOR i \GETS 0 \TO L-1 \DO
        \BEGIN
            O \GETS \mbox{get element of $\mathcal{S}$ in position $i$}\\
            M \GETS \mbox{get element of $O$ in position $countArray[i]$} \\
            \FOREACH \sigma \in M \DO
            \mbox{add $\sigma$ to $ControlSet$}
        \END
        \\
        \mbox{add $ControlSets$ to $ControlSets$}\\
        \mbox{\textbf{comment:} $index$ gets increased whenever $countArray[index]$ reaches its}\\
        \phantom{\mbox{\textbf{comment:} }}\mbox{max value, $countArrayMax[index]$.}\\
        index \GETS 0\\
        \REPEAT
            \BEGIN
                \COMMENT{$increasedIndex$ breaks the \textbf{repeat} loop.}\\
                increasedIndex \GETS \FALSE \\
                countArray[index] \GETS countArray[index]+1 \\
                \IF \mbox{$countArray[index]$ equals $countArrayMax[index]$}
                \THEN
                \BEGIN
                    countArray[index] \GETS 0 \\
                    index \GETS index+1 \\
                    increasedIndex \GETS \TRUE
                \END
                \\
                \IF \mbox{$index$ equals $L$}
                \THEN \mbox{exit \textbf{repeat} loop}
            \END
        \UNTIL \NOT increasedIndex
    \END
    \UNTIL \mbox{$index$ equals $L$}
\END
\\
\RETURN{ControlSets}
\end{pseudocode}

\clearpage

\begin{pseudocode}[ruled]{PruneControlSets}{ControlSets}
\COMMENT{$PrunedControlSets$ is a set} \\
PrunedControlSets \GETS \mbox{copy $ControlSets$} \\
\FOREACH ControlSet \in ControlSets \DO
    \BEGIN
    \FOREACH ControlSet' \in ControlSets \DO
        \BEGIN
        \IF \mbox{$ControlSet'$ is not $ControlSet$} \THEN
            \BEGIN
            \IF \mbox{$ControlSet'$ is a subset of $ControlSet$} \THEN
                \BEGIN
                \mbox{remove $ControlSet$ from $PrunedControlSets$}\\
                \mbox{exit \textbf{for} loop}
                \END
            \END
        \END
    \END
\\
\RETURN{PrunedControlSets}
\end{pseudocode}

\subsection{Pseudocode for the stable motif blocking algorithm} \label{sec:S7B}

\textit{Step 1}: Identify the sequences of stable motifs that lead to $\mathcal{A}$. This step is the same as the first step in the stable motif control algorithm (Algorithm \ref{algorithm1}), and can be obtained from the stable motif succession diagram (Fig. \hyperref[fig:ReductionMethod]{2}).
\\
\\
\indent \textit{Step 2}: Take each stable motif's state $\mathcal{M}_i$ in the sequences obtained in the previous step ($Sequences$). Create a new set $\mathbf{M}_{\mathcal{A}}$ with all of these stable motif states, $\mathbf{M}_{\mathcal{A}}=\left\{\mathcal{M}_i\right\}$.

\begin{pseudocode}[ruled]{MotifStates}{Sequences,L}
\COMMENT{$\mathbf{M}_{\mathcal{A}}$ and $\mathcal{M}$ are sets.} \\
\mathbf{M}_{\mathcal{A}} \GETS \mbox{empty set} \\
\FOREACH \mathcal{S} \in Sequences \DO
    \BEGIN
    \FOREACH v \in \mathcal{S} \mbox{ s.t. $v$ is not a sink node} \DO
        \BEGIN
            \COMMENT{$\mathcal{M}$ stores the states of the motif $L(v)$.}\\
            \mathcal{M} \GETS L(v)\\
            \mbox{add $\mathcal{M}$ to $\mathbf{M}_{\mathcal{A}}$}
        \END
    \END
\\
\RETURN{\mathbf{M}_{\mathcal{A}}}
\end{pseudocode}

\textit{Step 3}: Take each node state $\sigma_j \subset \mathcal{M}_i$ of the stable motif's states $\mathcal{M}_i$ in $\mathbf{M}_{\mathcal{A}}$. Create a new set $\mathcal{B}_{\mathcal{A}}$ with the negation of each node state, $\mathcal{B}_{\mathcal{A}}=\left\{\overline{\sigma}_j\right\}$. The node states in $\mathcal{B}_{\mathcal{A}}$ and any combination of them are identified as potential interventions to block attractor $\mathcal{A}$.

\begin{pseudocode}[ruled]{StableMotifBlocking}{\mathbf{M}_{\mathcal{A}}}
\COMMENT{$\mathcal{B}_{\mathcal{A}}$ is a set.} \\
\mathcal{B}_{\mathcal{A}} \GETS \mbox{empty set} \\
\FOREACH \sigma \in \mathbf{M}_{\mathcal{A}} \DO
    \BEGIN
    \COMMENT{$\sigma'$ is a node in the logical model together with a Boolean variable with its state.}\\
    \sigma' \GETS \mbox{reverse node state of $\sigma$}\\
    \mbox{add $\sigma'$ to $\mathcal{B}_{\mathcal{A}}$}
    \END
\\
\RETURN{\mathcal{B}_{\mathcal{A}}}
\end{pseudocode}

\clearpage
\renewcommand{\thefigure}{S\arabic{figure}}
\setcounter{figure}{0}

\section*{Supporting Information Figures} \label{sec:SFigs}
\ \\
\begin{figure}[!h]
\centerline{\includegraphics[width=\textwidth]{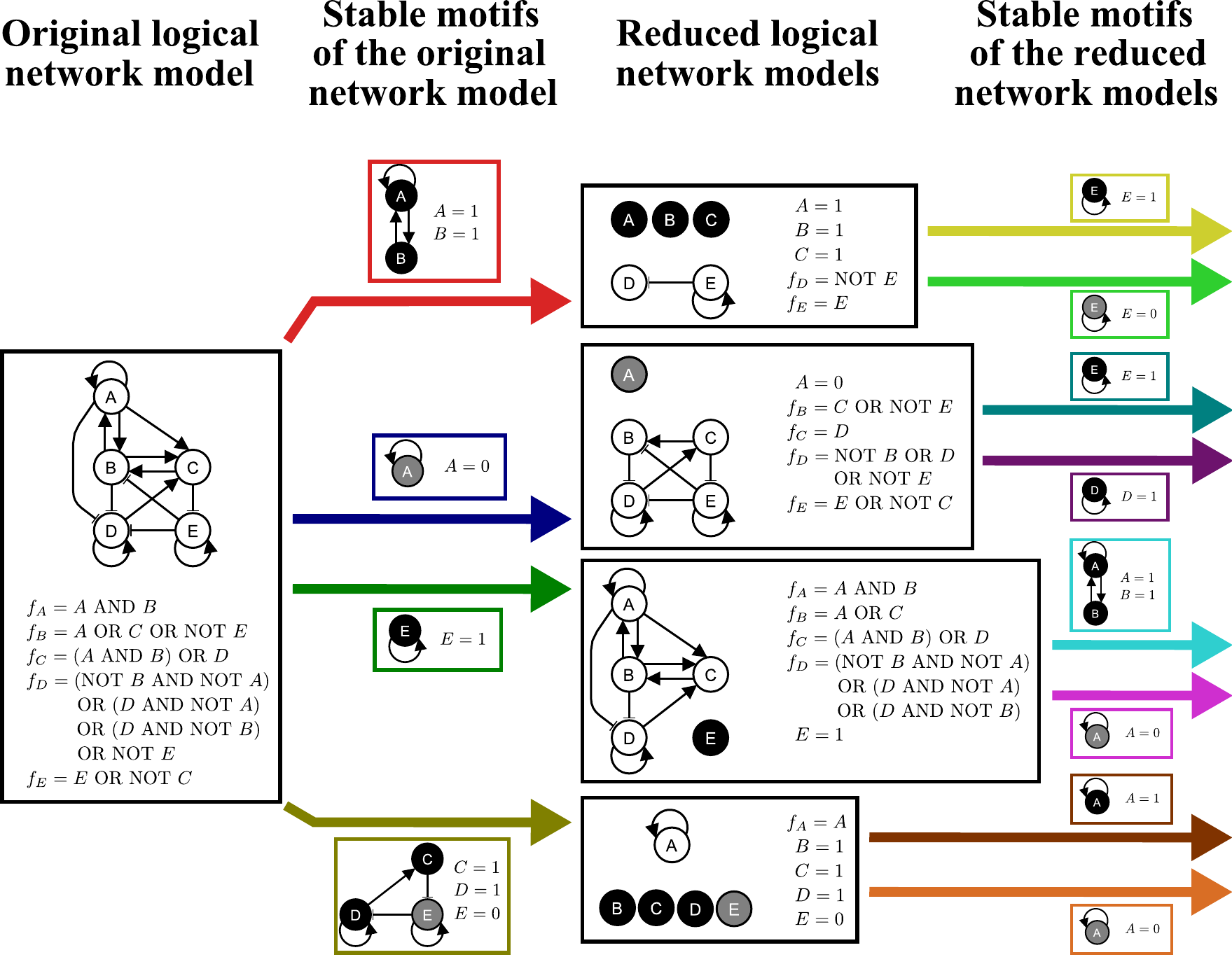}}
\caption{Stable motifs and simplified logical networks for the logical network in Fig. \ref{fig:NetworkExample}. Read from left to right, the figure shows the logical network in Fig. \ref{fig:NetworkExample}, the stable motifs of this logical network, the simplified networks obtained from tracing the downstream effect of each of the original logical network's stable motifs, and the stable motifs obtained from these simplified networks. Nodes are colored based on their respective node state: gray for 0, black for 1, and white for nodes whose state is not yet determined. Each large arrow has an associated stable motif sharing the arrow's color. These large arrows stand for the use of a network reduction technique on the network they start from by tracing the downstream effect of their associated stable motifs on this network.}
\label{fig:FigS1}
\end{figure}

\begin{figure}[!h]
\centerline{\includegraphics[width=0.6\textwidth]{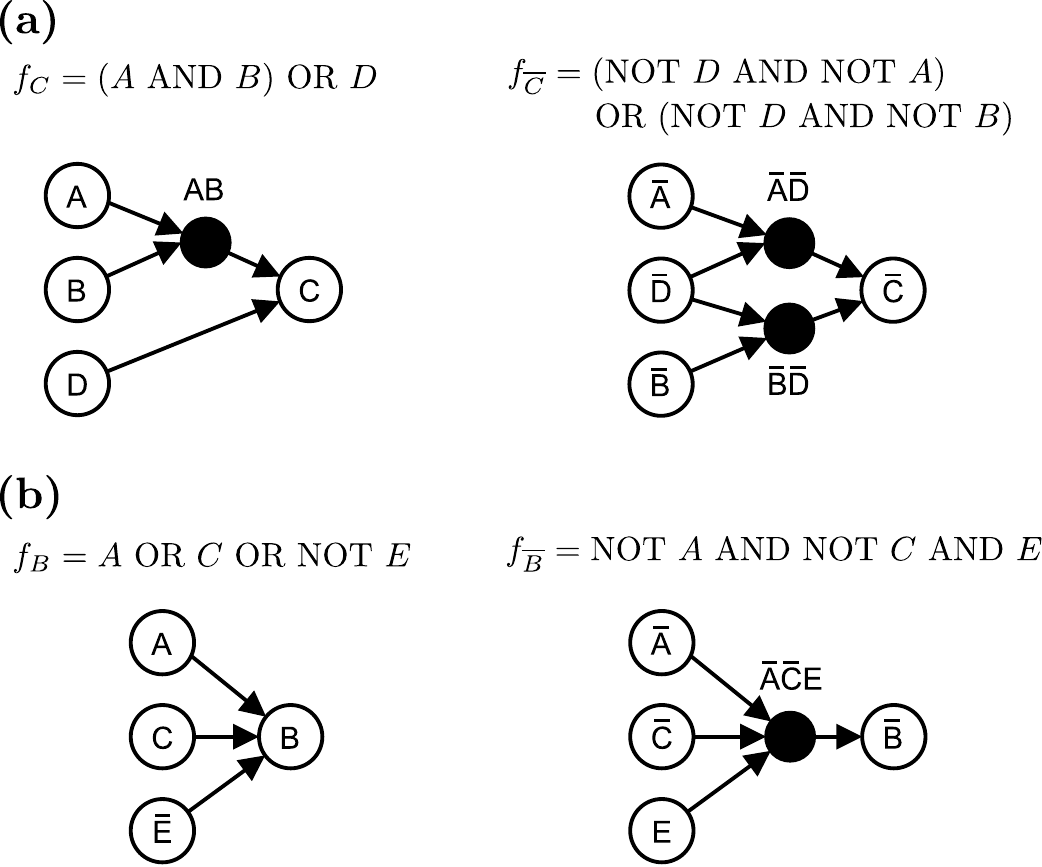}}
\caption{Example of the expanded network representation of selected nodes of the logical network in Fig. \ref{fig:NetworkExample}. The logical function of each example node is shown above its expanded network representation. Nodes are colored white if they denote normal nodes or complementary node (complementary nodes have a bar above their name, while normal nodes do not), and colored black if they denote composite nodes. For more details see Text \hyperref[sec:S1]{S1} and Text \hyperref[sec:S2]{S2}. (a) Expanded network representation for normal node $C$, complementary node $\overline C$, and their inputs. (b) Expanded network representation for normal node $B$, complementary node $\overline B$, and their inputs.}
\label{fig:FigS2}
\end{figure}

\begin{figure}[!h]
\centerline{\includegraphics[width=0.6\textwidth]{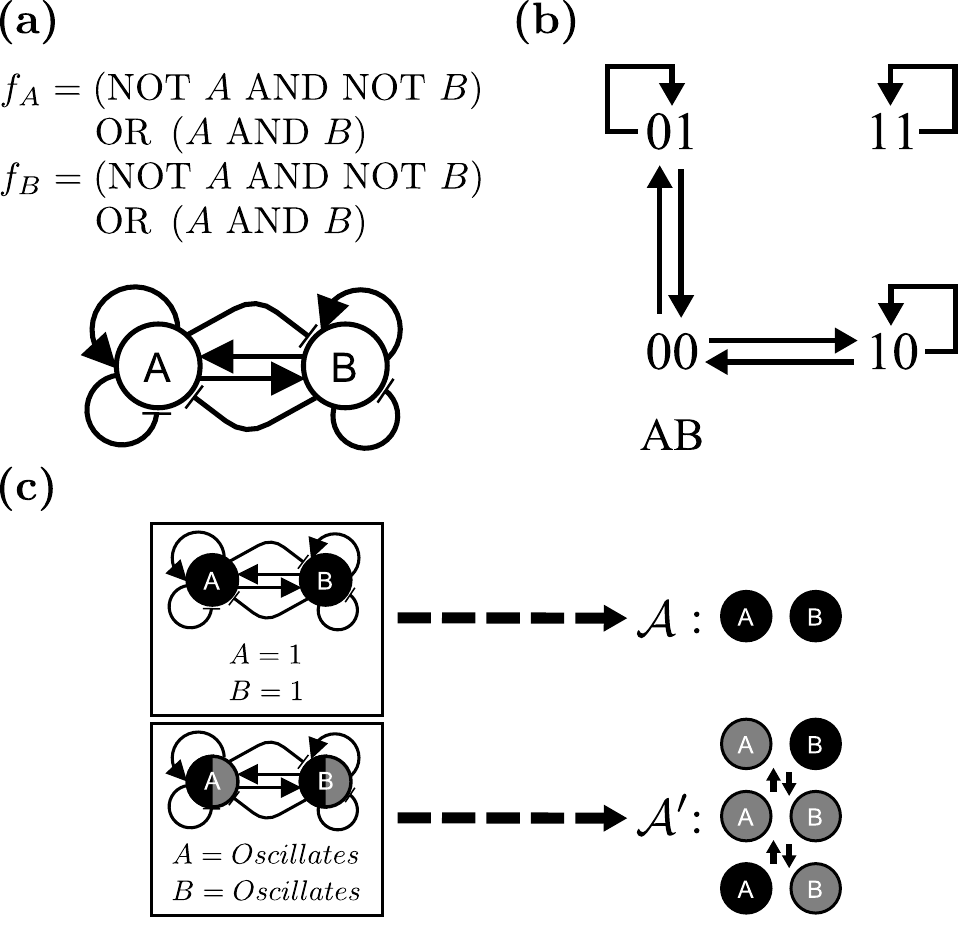}}
\caption{Example logical network displaying unstable oscillations. The figure shows (a) a two node Boolean network whose logical functions are given by an XOR function, (b) the network's state transition graph, i.e., all combinations of network states and the allowed transitions between them under the general asynchronous updating scheme, and (c) the network's stable motif succession diagram. This Boolean network is the simplest example (up to a relabeling of node states) of so-called unstable oscillations. Unstable oscillations refer to a subset of nodes whose node states oscillate in an attractor while their node states are fixed in a different attractor, even though both attractors are the same except for the state of this subset of nodes. In the example Boolean network shown in this figure, we have the states of nodes $A$ and $B$ oscillate between three network states in attractor $\mathcal{A}'=\left\{(A=1, B=0), (A=0, B=0), (A=0, B=1)\right\}$, while they are fixed in attractor $\mathcal{A}=\left\{A=1, B=1\right\}$. Unstable oscillations are treated with special care when using our attractor-finding method, since ignoring them can lead to missing attractors displaying this behavior; for more details see Text \hyperref[sec:S1]{S1} and Text \hyperref[sec:S2]{S2}.}
\label{fig:FigS3}
\end{figure}

\begin{figure}[!h]
\centerline{\includegraphics[width=0.6\textwidth]{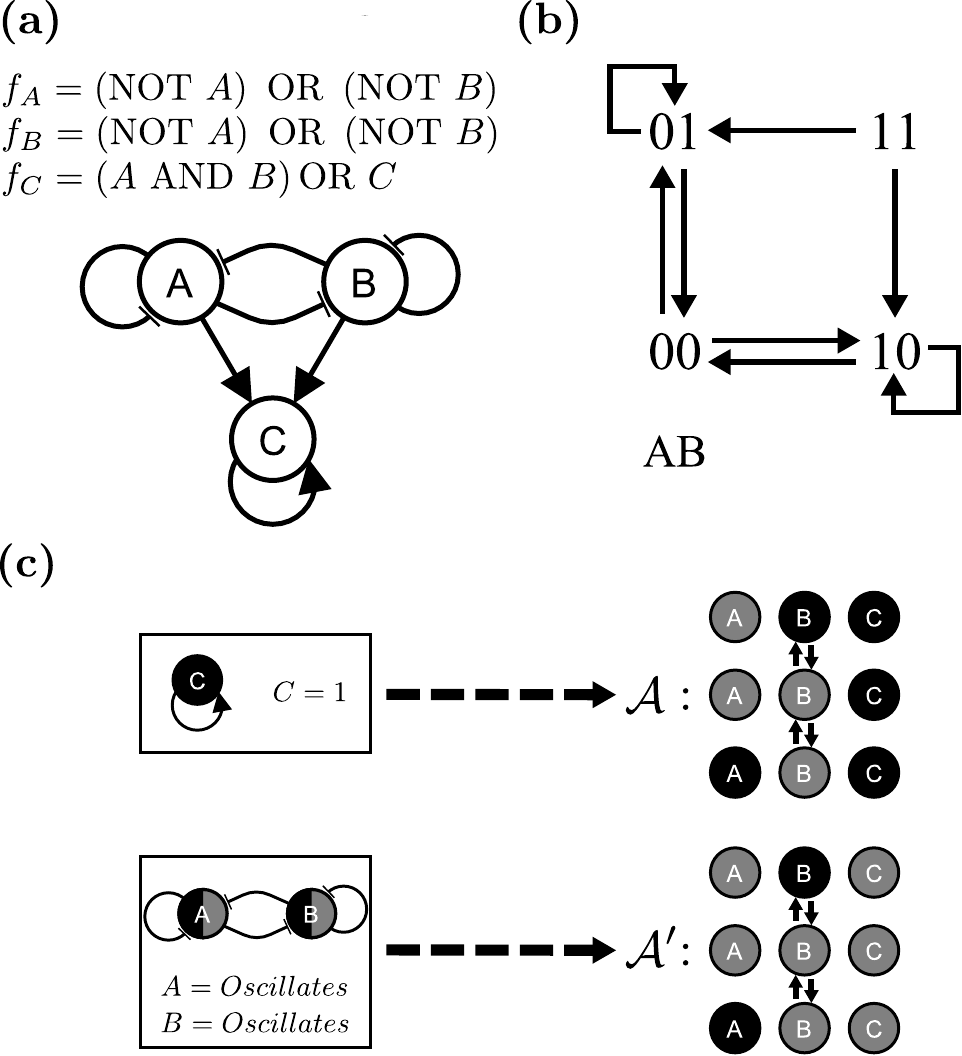}}
\caption{Example logical network displaying incomplete oscillations. The figure shows (a) a three node Boolean network that displays incomplete oscillations, (b) the sub-state-space of nodes $A$ and $B$ in the network's state transition graph (i.e., all combinations of network states and the allowed transitions between them) under the general asynchronous updating scheme, and (c) the network's stable motif succession diagram. Incomplete oscillations refer to a subset of nodes whose node states oscillate in an attractor but do not visit all possible states of their sub-state-space in the attractor. In the example Boolean network shown in this figure, we have the states of nodes $A$ and $B$ oscillate between three subnetwork states $\left\{(A=1, B=0), (A=0, B=0), (A=0, B=1)\right\}$ in the attractors $\mathcal{A}$ and $\mathcal{A}'$. Incomplete oscillations are treated with special care when using our attractor-finding method, since ignoring them can lead to missing attractors displaying this behavior; for more details see Text \hyperref[sec:S1]{S1} and Text \hyperref[sec:S2]{S2}.}
\label{fig:FigS4}
\end{figure}

\begin{figure}[!h]
\centerline{\includegraphics[width=\textwidth]{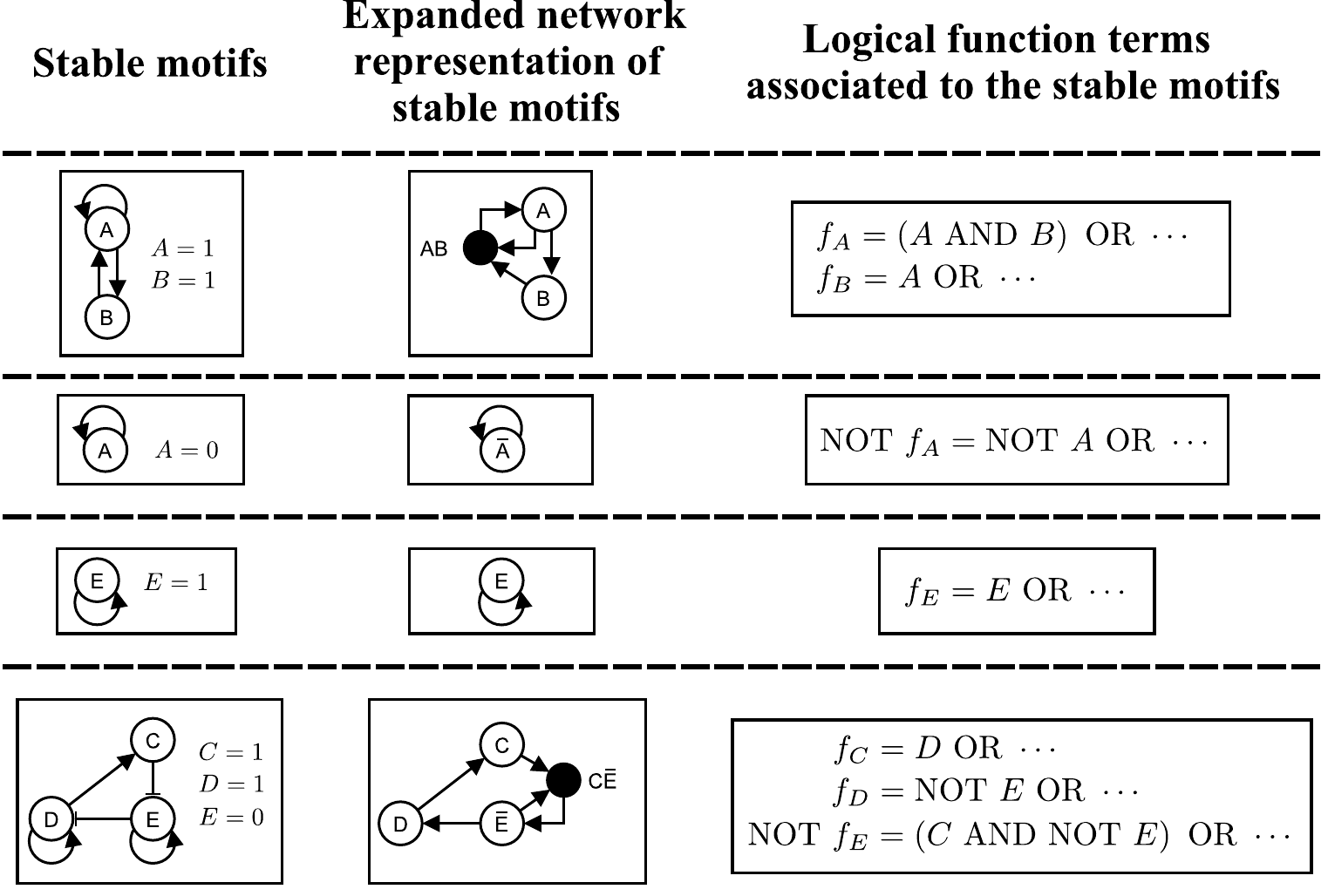}}
\caption{Stable motifs and simplified logical networks for the logical network in Fig. \ref{fig:NetworkExample}. Read from left to right, the figure shows the logical network in Fig. \ref{fig:NetworkExample}, the stable motifs of this logical network, the simplified networks obtained from tracing the downstream effect of each of the original logical network's stable motifs, and the stable motifs obtained from these simplified networks. Nodes are colored based on their respective node state: gray for 0, black for 1, and white for nodes whose state is not yet determined. Each large arrow has an associated stable motif sharing the arrow's color. These large arrows stand for the use of a network reduction technique on the network they start from by tracing the downstream effect of their associated stable motifs on this network.}
\label{fig:FigS5}
\end{figure}

\clearpage

\section*{Supporting Information Tables}

\begin{table*}[!h]
\caption{Validation of the intervention targets in Table 1 and single interventions from control sets with more than one node in Table 1 for the T-LGL leukemia network model. The relative apoptosis \% change is defined as $(\hbox{Apoptosis \%}-\hbox{Normal apoptosis \%})/(\hbox{Normal apoptosis \%})$, where Normal apoptosis \% = 62.1 \% is the percentage of initial conditions that go to apoptosis when no intervention is applied. Interventions marked with $\dag$ appear in more than one control strategy or target attractor in Table 1. The percentages are significant in the digits shown and have an estimated absolute error (standard deviation of the mean) of $3\cdot10^{-3}[\%p_{Attr}(100\%-\%p_{Attr})]^{1/2}\ \%$, where $\%p_{Attr}$ is the percentage shown (e.g. 0.03\% for a $\%p_{Attr}$ of 1\%, and 0.15\% for a $\%p_{Attr}$ of 50\%). }
\label{tab:SimulationsTLGL}
\begin{tabular*}{\hsize}{@{\extracolsep{\fill}}ccccccc}
\hline
Intervention & Successful? & Long-term? & Apoptosis & Relative apoptosis & Apoptosis & Relative apoptosis \\
     &  &  & \%  & \% change& \% & \% change \\
& & & (permanent & (permanent  & (nonpermanent & (nonpermanent \\
& & & intervention) & intervention) & intervention) & intervention) \\
\hline
\multicolumn7c{T-LGL stable motif control interventions ($C_{TLGL}$)} \\
\hline
\{S1P=ON\}$^\dag$ & Yes & Yes & 0.0 & -100 & 0.0 & -100 \\
\{Ceramide=OFF, & Yes & Yes &0.0 & -100 & 0.0 & -100 \\
SPHK1=ON\} &  &  & &  &  &  \\
\{Ceramide=OFF, & Yes & Yes &0.0 & -100 & 0.0 & -100 \\
PDGFR=ON\} &  &  & &  &  &  \\
\hline
\multicolumn7c{Apoptosis stable motif control interventions ($C_{Apoptosis}$)} \\
\hline
\{S1P=OFF\}$^\dag$ & Yes & Yes &100.0 & 61 & 100.0 & 61 \\
\{SPHK1=OFF\}$^\dag$ & Yes & Yes &100.0 & 61 & 100.0 & 61 \\
\{PDGFR=OFF\}$^\dag$ & Yes & Yes &100.0 & 61 & 100.0 & 61 \\
\{TBET=ON, & Yes & Yes &100.0 & 61 & 100.0 & 61 \\
Ceramide=ON, &  &  & &  &  &  \\
RAS=ON\} &  &  & &  &  &  \\
\{TBET=ON, & Yes & Yes &100.0 & 61 & 100.0 & 61 \\
Ceramide=ON, &  &  & &  &  &  \\
GRB2=ON\} &  &  & &  &  &  \\
\{TBET=ON, & Yes & Yes &100.0 & 61 & 100.0 & 61 \\
Ceramide=ON, &  &  & &  &  &  \\
IL2RB=ON\} &  &  & &  &  &  \\
\{TBET=ON, & Yes & Yes &100.0 & 61 & 100.0 & 61 \\
Ceramide=ON,  &  &  & &  &  &  \\
IL2RBT=ON\} &  &  & &  &  &  \\
\{TBET=ON, & Yes & Yes &100.0 & 61 & 100.0 & 61 \\
Ceramide=ON, &  &  & &  &  &  \\
ERK=ON\} &  &  & &  &  &  \\
\{TBET=ON, & Yes & Yes &100.0 & 61 & 100.0 & 61 \\
Ceramide=ON, &  &  & &  &  &  \\
MEK=ON, &  &  & &  &  &  \\
PI3K=ON\} &  &  & &  &  &  \\
\hline
\multicolumn7c{T-LGL stable motif blocking interventions ($B_{TLGL}$)} \\
\hline
\{Ceramide=ON\} & Yes & Yes &100.0 & 61 & 100.0 & 61 \\
\{PI3K=OFF\}$^\dag$ & Yes & No & 89.0 & 43 & 61.1 & 2 \\
\{RAS=OFF\}$^\dag$ & Yes & No & 95.2 & 53 & 62.0 & 0 \\
\{GRB2=OFF\}$^\dag$ & No & No & 58.5 & -6 & 62.1 & 0 \\
\{MEK=OFF\}$^\dag$ & Yes & No & 100.0 & 61 & 62.4 & 1 \\
\{ERK=OFF\}$^\dag$ & Yes & No & 100.0 & 61 & 62.1 & 0 \\
\{IL2RBT=OFF\}$^\dag$ & Yes & No & 100.0 & 61 & 62.1 & 0 \\
\{IL2RB=OFF\}$^\dag$ & Yes & No & 100.0 & 61 & 62.1 & 0 \\
\hline
\multicolumn7c{Apoptosis stable motif blocking interventions ($B_{Apoptosis}$)} \\
\hline
\{SPHK1=ON\} & Yes & Yes &12.4 & -80 & 12.3 & -80 \\
\{PDGFR=ON\} & Yes & Yes &23.6 & -62 & 23.8 & -62 \\
\{Ceramide=OFF\} & Yes & Partial &10.2 & -84 & 50.0 & -20 \\
\{sFas=ON\} & Yes & No &0.0 & -100 & 59.7 & -4 \\
\{Fas=OFF\} & Yes & No &0.0 & -100 & 56.9 & -9 \\
\{TBET=OFF\}$^\dag$ & Yes & No & 9.7 & -85 & 61.9 & 0 \\ \hline
\end{tabular*}
\end{table*}

\begin{table}[t]
\begin{tabular*}{\hsize}{@{\extracolsep{\fill}}ccccccc} \hline
Intervention & Successful? & Long-term? & Apoptosis & Relative apoptosis & Apoptosis & Relative apoptosis \\
     &  &  & \%  & \% change & \% & \% change \\
& & & (permanent & (permanent  & (nonpermanent & (nonpermanent \\
& & & intervention) & intervention) & intervention) & intervention) \\
\hline
\multicolumn7c{Single interventions of T-LGL stable motif control sets} \\
\hline
\{SPHK1=ON\} & Yes & Yes &8.2 & -87 & 12 & -80 \\
\{PDGFR=ON\} & Yes & Yes &23.9 & -62 & 23.8 & -62 \\
\{Ceramide=OFF\} & Yes & Partial &9.4 & -84 & 50.0 & -20 \\
\hline
\multicolumn7c{Single interventions of apoptosis stable motif control sets} \\ \hline
\{TBET=ON\} & No & No &62.2 & 0 & 62.3 & 0 \\
\{Ceramide=ON\} & Yes & Yes &100.0 & 61 & 100.0 & 61 \\
\{RAS=ON\} & No & No &62.4 & 0 & 62.6 & 1 \\
\{GRB2=ON\} & No & No &62.2 & 0 & 62.3 & 0 \\
\{IL2RB=ON\} & No & No &62.1 & 0 & 62.2 & 0 \\
\{IL2RBT=ON\} & No & No &62.1 & 0 & 62.3 & 0 \\
\{ERK=ON\} & No & No &62.1 & 0 & 62.3 & 0 \\
\{MEK=ON\} & No & No &62.2 & 0 & 62.0  & 0 \\
\{PI3K=ON\} & No & No &62.3 & 0 & 62.6 & 1 \\ \hline
\end{tabular*}
\end{table}

\clearpage

\begin{table*}[t]
\caption{Validation of the intervention targets in Table 2 and single interventions from control sets with more than one node in Table 2 for the helper T cell network. The relative attractor \% change is defined as $(\hbox{attractor \%}-\hbox{normal attractor \%})/(\hbox{normal attractor \%})$, where the normal attractor \%  is the percentage of initial conditions that go to the attractor of interest when no intervention is applied. The normal attractor percentages are 48.6 \%, 47.5 \%, 1.3 \%, and 2.6 \% for the Th1, Th2, Th17, and Treg helper T cell subtypes, respectively. Interventions marked with $\dag$ appear in more than one control strategy or target attractor in Table 2. The percentages are significant in the digits shown and have an estimated absolute error (standard deviation of the mean) of $3\cdot10^{-3}[\%p_{Attr}(100\%-\%p_{Attr})]^{1/2}\ \%$, where $\%p_{Attr}$ is the percentage shown (e.g. 0.03\% for a $\%p_{Attr}$ of 1\%, and 0.15\% for a $\%p_{Attr}$ of 50\%).}
\label{tab:SimulationsTh}
\begin{tabular*}{\hsize}{@{\extracolsep{\fill}}ccccccc} \hline
Intervention & Successful? & Long-term? & Attractor & Relative attractor & Attractor & Relative attractor \\
     &  &  & \%  & \% change& \% & \% change \\
& & & (permanent & (permanent  & (nonpermanent & (nonpermanent \\
& & & intervention) & intervention) & intervention) & intervention) \\
\hline
\multicolumn7c{Th1 stable motif control interventions ($C_{Th1}$)} \\ \hline
\{TBET=ON\} & Yes & Yes & 100.0 & 106 & 100.0 & 106 \\
\hline
\multicolumn7c{Th2 stable motif control interventions ($C_{Th2}$)} \\
\hline
\{GATA3=ON\} & Yes & Yes & 100.0 & 111 & 100.0 & 111 \\
\hline
\multicolumn7c{Th17 stable motif control interventions ($C_{Th17}$)} \\
\hline
\{GATA3=OFF, & Yes & Yes &100.0 & 7357 & 100.0 & 7357 \\
FOXP3=OFF, &  &  & &  &  &  \\
TBET=OFF, &  &  & &  &  &  \\
STAT3=ON\} &  &  & &  &  &  \\
\{GATA3=OFF, & Yes & Yes &100.0 & 7357 & 100.0 & 7357 \\
FOXP3=OFF, &  &  & &  &  &  \\
TBET=OFF, &  &  & &  &  &  \\
IL10=ON\} &  &  & &  &  &  \\
\{GATA3=OFF, & Yes & Yes &100.0 & 7357 & 100.0 & 7357 \\
FOXP3=OFF, &  &  & &  &  &  \\
TBET=OFF, &  &  & &  &  &  \\
IL10R=ON\} &  &  & &  &  &  \\
\{GATA3=OFF, & Yes & Yes &100.0 & 7357 & 100.0 & 7357 \\
FOXP3=OFF, &  &  & &  &  &  \\
TBET=OFF, &  &  & &  &  &  \\
IL21=ON\} &  &  & &  &  &  \\
\{GATA3=OFF, & Yes & Yes &100.0 & 7357 & 100.0 & 7357 \\
FOXP3=OFF, &  &  & &  &  &  \\
TBET=OFF, &  &  & &  &  &  \\
IL21R=ON\} &  &  & &  &  &  \\
\{GATA3=OFF, & Yes & Yes &100.0 & 7357 & 100.0 & 7357 \\
FOXP3=OFF, &  &  & &  &  &  \\
TBET=OFF, &  &  & &  &  &  \\
IL23R=ON, &  &  & &  &  &  \\
RORGT=ON\} &  &  & &  &  &  \\
\hline
\multicolumn7c{Treg stable motif control interventions ($C_{Treg}$)} \\
\hline
\{GATA3=OFF, & Yes & Yes &100.0 & 3706 & 100.0 & 3706 \\
FOXP3=ON, &  &  & &  &  &  \\
TBET=OFF\} &  &  & &  &  &  \\
\{GATA3=OFF, & Yes & Yes &100.0 & 3706 & 100.0 & 3706 \\
TBET=OFF, &  &  & &  &  &  \\
STAT3=OFF\} &  &  & &  &  &  \\
\{GATA3=OFF, & Yes & Yes &100.0 & 3706 & 100.0 & 3706 \\
TBET=OFF, &  &  & &  &  &  \\
IL23R=OFF, &  &  & &  &  &  \\
IL10R=OFF, &  &  & &  &  &  \\
IL21R=OFF\} &  &  & &  &  &  \\
\{GATA3=OFF, & Yes & Yes &100.0 & 3706 & 100.0 & 3706 \\
TBET=OFF, &  &  & &  &  &  \\
IL23R=OFF, &  &  & &  &  &  \\
IL10=OFF, &  &  & &  &  &  \\
IL21R=OFF\} &  &  & &  &  &  \\
\{GATA3=OFF, & Yes & Yes &100.0 & 3706 & 100.0 & 3706 \\
TBET=OFF, &  &  & &  &  &  \\
IL23R=OFF, &  &  & &  &  &  \\
IL10R=OFF, &  &  & &  &  &  \\
IL21=OFF\} &  &  & &  &  &  \\
\hline
\end{tabular*}
\end{table*}

\begin{table*}[t]
\begin{tabular*}{\hsize}{@{\extracolsep{\fill}}ccccccc} \hline
Intervention & Successful? & Long-term? & Attractor & Relative attractor & Attractor & Relative attractor \\
     &  &  & \%  & \% change& \% & \% change \\
& & & (permanent & (permanent  & (nonpermanent & (nonpermanent \\
& & & intervention) & intervention) & intervention) & intervention) \\
\hline
\multicolumn7c{Treg stable motif control interventions ($C_{Treg}$) (continuation)} \\
\hline
\{GATA3=OFF, & Yes & Yes &100.0 & 3706 & 100.0 & 3706 \\
TBET=OFF, &  &  & &  &  &  \\
IL23R=OFF, &  &  & &  &  &  \\
IL10=OFF, &  &  & &  &  &  \\
IL21=OFF\} &  &  & &  &  &  \\
\hline
\multicolumn7c{Th1 stable motif blocking interventions ($B_{Th1}$)} \\
\hline
\{GATA3=ON\}$^\dag$  &  Yes  &  Yes  & 0.0 & -100 & 0.0 & -100\\
\{IL4=ON\}$^\dag$  &  No  &  No  & 48.2 & -1 & 48.1 & -1\\
\{IL4R 2=ON\}$^\dag$  &  No  &  No  & 47.2 & -3 & 47.4 & -2\\
\{STAT6=ON\}$^\dag$  &  No  &  No  & 45.3 & -7 & 45.0 & -7\\
\{STAT1=OFF\}  &  Yes  &  Yes  & 37.2 & -23 & 37.5 & -23\\
\{IFNG=OFF\}  &  No  &  No  & 48.2 & -1 & 48.0 & -1\\
\{IFNGR=OFF\}  &  No  &  No  & 47.0 & -3 & 46.8 & -4\\
\{IL23=OFF\}$^\dag$  &  No  &  No  & 48.7 & 0 & 48.9 & 1\\
\{IL10=ON\}$^\dag$  &  No  &  No  & 48.6 & 0 & 48.8 & 1\\
\{IL10=OFF\}$^\dag$  &  No  &  No  & 48.9 & 1 & 48.7 & 0\\
\{IL10R=ON\}$^\dag$  &  No  &  No  & 48.8 & 1 & 48.6 & 0\\
\{IL10R=OFF\}$^\dag$  &  No  &  No  & 48.6 & 0 & 48.9 & 1\\
\{IL21=ON\}$^\dag$  &  No  &  No  & 48.0 & -1 & 48.8 & 0\\
\{IL21=OFF\}$^\dag$  &  No  &  No  & 48.7 & 0 & 48.4 & 0\\
\{IL21R=ON\}$^\dag$  &  No  &  No  & 48.8 & 0 & 48.6 & 0\\
\{IL21R=OFF\}$^\dag$  &  No  &  No  & 48.6 & 0 & 48.7 & 0\\
\{STAT3=ON\}$^\dag$  &  No  &  No  & 48.8 & 1 & 48.8 & 0\\
\{IL23R=ON\}$^\dag$  &  No  &  No  & 48.6 & 0 & 48.6 & 0\\
\{IL23R=OFF\}$^\dag$  &  No  &  No  & 48.7 & 0 & 49.1 & 1\\
\{RORGT=ON\}$^\dag$  &  No  &  No  & 48.7 & 0 & 48.9 & 1\\
\{RORGT=OFF\}$^\dag$  &  No  &  No  & 48.7 & 0 & 48.6 & 0\\
\{FOXP3=ON\}$^\dag$  &  No  &  No  & 48.6 & 0 & 48.4 & 0\\
\{FOXP3=OFF\}$^\dag$  &  No  &  No  & 48.7 & 0 & 48.7 & 0\\
\hline
\multicolumn7c{Th2 stable motif blocking interventions ($B_{Th2}$)} \\
\hline
\{TBET=ON\}$^\dag$  &  Yes  &  Yes  & 0.0 & -100 & 0.0 & -100\\
\{GATA3=OFF\}  &  Yes  &  Yes  & 0.0 & -100 & 0.0 & -100\\
\{STAT1=ON\}$^\dag$  &  No  &  No  & 44.7 & -6 & 44.7 & -6\\
\{IFNG=ON\}$^\dag$  &  No  &  No  & 47.3 & 0 & 47.0 & -1\\
\{IFNGR=ON\}$^\dag$  &  No  &  No  & 46.6 & -2 & 46.6 & -2\\
\{IL23=OFF\}$^\dag$  &  No  &  No  & 47.5 & 0 & 47.3 & 0\\
\{IL23R=OFF\}$^\dag$  &  No  &  No  & 47.5 & 0 & 47.1 & -1\\
\{STAT3=OFF\}$^\dag$  &  No  &  No  & 47.3 & 0 & 47.3 & 0\\
\{IL10=OFF\}$^\dag$  &  No  &  No  & 47.3 & 0 & 47.5 & 0\\
\{IL10R=OFF\}$^\dag$  &  No  &  No  & 47.6 & 0 & 47.3 & 0\\
\{RORGT=ON\}$^\dag$  &  No  &  No  & 47.5 & 0 & 47.3 & 0\\
\{FOXP3=ON\}$^\dag$  &  No  &  No  & 47.5 & 0 & 47.7 & 0\\
\{FOXP3=OFF\}$^\dag$  &  No  &  No  & 47.6 & 0 & 47.5 & 0\\
\hline
\end{tabular*}
\end{table*}

\begin{table*}[t]
\begin{tabular*}{\hsize}{@{\extracolsep{\fill}}ccccccc} \hline
Intervention & Successful? & Long-term? & Attractor & Relative attractor & Attractor & Relative attractor \\
     &  &  & \%  & \% change& \% & \% change \\
& & & (permanent & (permanent  & (nonpermanent & (nonpermanent \\
& & & intervention) & intervention) & intervention) & intervention) \\
\hline
\multicolumn7c{Th17 stable motif blocking interventions ($B_{Th17}$)} \\
\hline
\{GATA3=ON\}$^\dag$  &  Yes  &  Yes  & 0.0 & -100 & 0.0 & -100\\
\{TBET=ON\}$^\dag$  &  Yes  &  Yes  & 0.0 & -100 & 0.0 & -100\\
\{IL4=ON\}$^\dag$  &  Yes  &  Yes  & 0.0 & -100 & 0.0 & -100\\
\{IL4R\_2=ON\}$^\dag$  &  Yes  &  Yes  & 0.0 & -100 & 0.0 & -100\\
\{STAT6=ON\}$^\dag$  &  Yes  &  Yes  & 0.0 & -100 & 0.0 & -100\\
\{STAT1=ON\}$^\dag$  &  Yes  &  Yes  & 0.0 & -100 & 0.0 & -100\\
\{IFNG=ON\}$^\dag$  &  Yes  &  Yes  & 0.0 & -100 & 0.0 & -100\\
\{IFNGR=ON\}$^\dag$  &  Yes  &  Yes  & 0.0 & -100 & 0.0 & -100\\
\{STAT3=OFF\}$^\dag$  &  Yes  &  Yes  & 0.0 & -100 & 0.0 & -100\\
\{FOXP3=ON\}$^\dag$  &  Yes  &  Yes  & 0.0 & -100 & 0.0 & -100\\
\{RORGT=OFF\}$^\dag$  &  Yes  &  Yes  & 0.0 & -100 & 0.0 & -100\\
\{IL21=OFF\}$^\dag$  &  Yes  &  Yes  & 1.1 & -20 & 1.1 & -20\\
\{IL21R=OFF\}$^\dag$  &  Yes  &  Yes  & 1.0 & -28 & 1.0 & -23\\
\{IL23=OFF\}$^\dag$  &  Partial  &  Partial  & 1.2 & -11 & 1.2 & -12\\
\{IL23R=OFF\}$^\dag$  &  Yes  &  Yes  & 1.1 & -19 & 1.1 & -19\\
\{IL10=OFF\}$^\dag$  &  Yes  &  Yes  & 1.1 & -20 & 1.1 & -20\\
\{IL10R=OFF\}$^\dag$  &  Yes  &  Yes  & 1.0 & -28 & 1.0 & -29\\
\hline
\multicolumn7c{Treg stable motif blocking interventions ($B_{Treg}$)} \\
\hline
\{GATA3=ON\}$^\dag$  &  Yes  &  Yes  & 0.0 & -100 & 0.0 & -100\\
\{TBET=ON\}$^\dag$  &  Yes  &  Yes  & 0.0 & -100 & 0.0 & -100\\
\{IL4=ON\}$^\dag$  &  Yes  &  Yes  & 0.0 & -100 & 0.0 & -100\\
\{IL4R\_2=ON\}$^\dag$  &  Yes  &  Yes  & 0.0 & -100 & 0.0 & -100\\
\{STAT6=ON\}$^\dag$  &  Yes  &  Yes  & 0.0 & -100 & 0.0 & -100\\
\{STAT1=ON\}$^\dag$  &  Yes  &  Yes  & 0.0 & -100 & 0.0 & -100\\
\{IFNG=ON\}$^\dag$  &  Yes  &  Yes  & 0.0 & -100 & 0.0 & -100\\
\{IFNGR=ON\}$^\dag$  &  Yes  &  Yes  & 0.0 & -100 & 0.0 & -100\\
\{STAT3=ON\}$^\dag$  &  Yes  &  Yes  & 0.6 & -76 & 0.6 & -76\\
\{STAT3=OFF\}$^\dag$  &  No  &  No  & 3.7 & 41 & 3.7 & 42\\
\{FOXP3=OFF\}$^\dag$  &  Yes  &  Yes  & 0.0 & -100 & 2.0 & -23\\
\{RORGT=ON\}$^\dag$  &  No  &  No  & 2.4 & -10 & 2.4 & -9\\
\{RORGT=OFF\}$^\dag$  &  No  &  No  & 3.9 & 48 & 3.9 & 50\\
\{IL21=ON\}$^\dag$  &  Yes  &  Yes  & 1.1 & -60 & 1.0 & -61\\
\{IL21=OFF\}$^\dag$  &  No  &  No  & 2.7 & 2 & 2.7 & 3\\
\{IL21R=ON\}$^\dag$  &  Yes  &  Yes  & 0.8 & -70 & 0.8 & -71\\
\{IL21R=OFF\}$^\dag$  &  No  &  No  & 2.9 & 12 & 2.8 & 7\\
\{IL23=OFF\}$^\dag$  &  No  &  No  & 2.6 & -2 & 2.7 & 2\\
\{IL23R=ON\}$^\dag$  &  Yes  &  Yes  & 0.8 & -70 & 0.8 & -69\\
\{IL23R=OFF\}$^\dag$  &  No  &  No  & 2.8 & 6 & 2.7 & 3\\
\{IL10=ON\}$^\dag$  &  Yes  &  Yes  & 1.1 & -60 & 1.0 & -60\\
\{IL10=OFF\}$^\dag$  &  No  &  No  & 2.7 & 4 & 2.7 & 4\\
\{IL10R=ON\}$^\dag$  &  Yes  &  Yes  & 0.7 & -72 & 0.8 & -70\\
\{IL10R=OFF\}$^\dag$  &  No  &  No  & 2.9 & 10 & 2.8 & 8\\
\hline
\end{tabular*}
\end{table*}

\begin{table*}[t]
\begin{tabular*}{\hsize}{@{\extracolsep{\fill}}ccccccc} \hline
Intervention & Successful? & Long-term? & Attractor & Relative attractor & Attractor & Relative attractor \\
     &  &  & \%  & \% change& \% & \% change \\
& & & (permanent & (permanent  & (nonpermanent & (nonpermanent \\
& & & intervention) & intervention) & intervention) & intervention) \\
\hline
\multicolumn7c{Single interventions of Th17 stable motif control sets} \\
\hline
\{GATA3=OFF\}  &  Yes  &  Yes  & 6.3 & 369 & 6.2 & 359 \\
\{FOXP3=OFF\}  &  Partial  &  Partial  & 1.7 & 25 & 1.8 & 31 \\
\{TBET=OFF\}  &  Yes  &  Yes  & 7.5 & 461 & 7.6 & 468 \\
\{STAT3=ON\}  &  Yes  &  Yes  & 3.3 & 146 & 3.2 & 142 \\
\{IL10=ON\}  &  Yes  &  Yes  & 2.8 & 110 & 2.9 & 114 \\
\{IL10R=ON\}  &  Yes  &  Yes  & 3.1 & 132 & 3.1 & 130 \\
\{IL21=ON\}  &  Yes  &  Yes  & 3.0 & 120 & 2.8 & 107 \\
\{IL21R=ON\}  &  Yes  &  Yes  & 3.1 & 127 & 3.1 & 133 \\
\{IL23R=ON\}  &  Yes  &  Yes  & 3.1 & 134 & 3.1 & 130 \\
\{RORGT=ON\}  &  No  &  No  & 1.5 & 9 & 1.4 & 6 \\
\hline
\multicolumn7c{Single interventions of Treg stable motif control sets} \\
\hline
\{GATA3=OFF\}  &  Yes  &  Yes  & 12.0 & 358 & 11.9 & 354\\
\{FOXP3=ON\}  &  Partial  &  Partial  & 3.9 & 49 & 3.9 & 49\\
\{TBET=OFF\}  &  Yes  &  Yes  & 13.5 & 415 & 13.7 & 423\\
\{STAT3=OFF\}  &  Partial  &  Partial  & 3.7 & 41 & 3.7 & 42\\
\{IL21=OFF\}  &  No  &  No  & 2.3 & -13 & 2.7 & 3\\
\{IL21R=OFF\}  &  No  &  No  & 2.6 & -2 & 2.8 & 7\\
\{IL23R=OFF\}  &  No  &  No  & 2.8 & 6 & 2.7 & 3\\
\hline
\end{tabular*}
\end{table*}

\clearpage

\begin{table*}[t]
\caption{Validation of the intervention targets in Table 1 for the T-LGL leukemia differential equation network model and single interventions from control sets with more than one node in Table 1 for the T-LGL leukemia differential equation network model. The relative apoptosis \% change is defined as $(\hbox{Apoptosis \%}-\hbox{Normal apoptosis \%})/(\hbox{Normal apoptosis \%})$, where Normal apoptosis \% = 54.7 \% is the percentage of initial conditions that go to apoptosis when no intervention is applied. Interventions marked with $\dag$ appear in more than one control strategy or target attractor in Table 1. The percentages are significant in the digits shown and have an estimated absolute error (standard deviation of the mean) of $6\cdot10^{-3}[\%p_{Attr}(100\%-\%p_{Attr})]^{1/2}\ \%$, where $\%p_{Attr}$ is the percentage shown (e.g. 0.06\% for a $\%p_{Attr}$ of 1\%, and 0.3\% for a $\%p_{Attr}$ of 50\%).}
\label{tab:SimulationsTLGLODE}
\begin{tabular*}{\hsize}{@{\extracolsep{\fill}}ccccccc}
\hline
Intervention & Successful? & Long-term? & Apoptosis & Relative apoptosis & Apoptosis & Relative apoptosis \\
     &  &  & \%  & \% change& \% & \% change \\
& & & (permanent & (permanent  & (nonpermanent & (nonpermanent \\
& & & intervention) & intervention) & intervention) & intervention) \\
\hline
\multicolumn7c{T-LGL stable motif control interventions ($C_{TLGL}$)} \\
\hline
\{S1P=ON\}$^\dag$  &  Yes  &  Yes  & 0.0 & -100 & 0.0 & -100\\
\{Ceramide=OFF,  &  Yes  &  Yes  & 0.0 & -100 & 0.0 & -100\\
SPHK1=ON\}  &    &    &   &    &    & \\
\{Ceramide=OFF,  &  Yes  &  Yes  & 0.0 & -100 & 0.0 & -100\\
PDGFR=ON\}  &    &    &   &    &    & \\
\hline
\multicolumn7c{Apoptosis stable motif control interventions ($C_{Apoptosis}$)} \\
\hline
\{S1P=OFF\}$^\dag$  &  Yes  &  Yes  & 99.9 & 83 & 99.9 & 83\\
\{SPHK1=OFF\}$^\dag$  &  Yes  &  Yes  & 99.9 & 83 & 99.9 & 83\\
\{PDGFR=OFF\}$^\dag$  &  Yes  &  Yes  & 99.8 & 83 & 99.8 & 83\\
\{TBET=ON,  &  Yes  &  Yes  & 100.0 & 83 & 100.0 & 83\\
Ceramide=ON,  &    &    &   &    &    & \\
RAS=ON\}  &    &    &   &    &    & \\
\{TBET=ON,  &  Yes  &  Yes  & 100.0 & 83 & 100.0 & 83\\
Ceramide=ON,  &    &    &   &    &    & \\
GRB2=ON\}  &    &    &   &    &    & \\
\{TBET=ON,  &  Yes  &  Yes  & 100.0 & 83 & 100.0 & 83\\
Ceramide=ON,  &    &    &   &    &    & \\
IL2RB=ON\}  &    &    &   &    &    & \\
\{TBET=ON,  &  Yes  &  Yes  & 100.0 & 83 & 100.0 & 83\\
Ceramide=ON,   &    &    &   &    &    & \\
IL2RBT=ON\}  &    &    &   &    &    & \\
\{TBET=ON,  &  Yes  &  Yes  & 100.0 & 83 & 100.0 & 83\\
Ceramide=ON,  &    &    &   &    &    & \\
ERK=ON\}  &    &    &   &    &    & \\
\{TBET=ON,  &  Yes  &  Yes  & 100.0 & 83 & 100.0 & 83\\
Ceramide=ON,  &    &    &   &    &    & \\
MEK=ON,  &    &    &   &    &    & \\
PI3K=ON\}  &    &    &   &    &    & \\
\hline
\multicolumn7c{T-LGL stable motif blocking interventions ($B_{TLGL}$)} \\
\hline
\{Ceramide=ON\}  &  Yes  &  Yes  & 99.9 & 83 & 99.9 & 83\\
\{PI3K=OFF\}$^\dag$  &  Yes  &  No  & 98.0 & 79 & 50.5 & -8\\
\{RAS=OFF\}$^\dag$  &  Yes  &  No  & 99.6 & 82 & 53.7 & -2\\
\{GRB2=OFF\}$^\dag$  &  No  &  No  & 54.6 & 0 & 54.3 & -1\\
\{MEK=OFF\}$^\dag$  &  Yes  &  No  & 100.0 & 83 & 54.4 & 0\\
\{ERK=OFF\}$^\dag$  &  Yes  &  No  & 100.0 & 83 & 54.3 & -1\\
\{IL2RBT=OFF\}$^\dag$  &  Yes  &  No  & 99.9 & 83 & 54.4 & -1\\
\{IL2RB=OFF\}$^\dag$  &  Yes  &  No  & 99.9 & 83 & 54.3 & -1\\
\hline
\multicolumn7c{Apoptosis stable motif blocking interventions ($B_{Apoptosis}$)} \\
\hline
\{SPHK1=ON\}  &  Yes  &  Yes  & 8.8 & -84 & 7.7 & -86\\
\{PDGFR=ON\}  &  Yes  &  Yes  & 13.9 & -75 & 13.7 & -75\\
\{Ceramide=OFF\}  &  Yes  &  Partial  & 11.2 & -79 & 33.5 & -39\\
\{sFas=ON\}  &  Yes  &  No  & 11.3 & -79 & 48.5 & -11\\
\{Fas=OFF\}  &  Yes  &  No  & 10.1 & -81 & 43.4 & -21\\
\{TBET=OFF\}$^\dag$  &  Yes  &  Yes  & 0.0 & -100 & 0.5 & -99\\
\hline
\end{tabular*}
\end{table*}

\begin{table*}[t]
\begin{tabular*}{\hsize}{@{\extracolsep{\fill}}ccccccc} \hline
Intervention & Successful? & Long-term? & Apoptosis & Relative apoptosis & Apoptosis & Relative apoptosis \\
     &  &  & \%  & \% change& \% & \% change \\
& & & (permanent & (permanent  & (nonpermanent & (nonpermanent \\
& & & intervention) & intervention) & intervention) & intervention) \\
\hline
\multicolumn7c{Single interventions of T-LGL stable motif control sets} \\
\hline
\{SPHK1=ON\}  &  Yes  &  Yes  & 8.8 & -84 & 7.7 & -86\\
\{PDGFR=ON\}  &  Yes  &  Yes  & 13.9 & -75 & 13.7 & -75\\
\{Ceramide=OFF\}  &  Yes  &  Partial  & 11.2 & -79 & 33.5 & -39\\
\hline
\multicolumn7c{Single interventions of apoptosis stable motif control sets} \\
\hline
\{TBET=ON\}  &  No  &  No  & 54.9 & 0 & 54.4 & 0\\
\{Ceramide=ON\}  &  Yes  &  Yes  & 99.9 & 83 & 99.9 & 83\\
\{RAS=ON\}  &  No  &  No  & 55.0 & 1 & 54.6 & 0\\
\{GRB2=ON\}  &  No  &  No  & 54.9 & 0 & 54.5 & 0\\
\{IL2RB=ON\}  &  No  &  No  & 54.8 & 0 & 54.4 & 0\\
\{IL2RBT=ON\}  &  No  &  No  & 54.8 & 0 & 54.4 & 0\\
\{ERK=ON\}  &  No  &  No  & 54.9 & 0 & 54.5 & 0\\
\{MEK=ON\}  &  No  &  No  & 54.8 & 0 & 54.4 & 0\\
\{PI3K=ON\}  &  No  &  No  & 55.3 & 1 & 54.9 & 0\\
\hline
\end{tabular*}
\end{table*}

\clearpage

\begin{table*}[t]
\caption{Validation of the stable motif control intervention targets in Table 2 for the helper T cell differential equation network model. The relative attractor \% change is defined as $(\hbox{attractor \%}-\hbox{normal attractor \%})/(\hbox{normal attractor \%})$, where the normal attractor \%  is the percentage of initial conditions that go to the attractor of interest when no intervention is applied. The normal attractor percentages are 50.0 \%, 45.4 \%, 2.8 \%, and 1.8 \% for the Th1, Th2, Th17, and Treg helper T cell subtypes, respectively. Interventions marked with $\dag$ appear in more than one control strategy or target attractor in Table 2. The percentages are significant in the digits shown and have an estimated absolute error (standard deviation of the mean) of $6\cdot10^{-3}[\%p_{Attr}(100\%-\%p_{Attr})]^{1/2}\ \%$, where $\%p_{Attr}$ is the percentage shown (e.g. 0.06\% for a $\%p_{Attr}$ of 1\%, and 0.3\% for a $\%p_{Attr}$ of 50\%).}
\label{tab:SimulationsThODE}
\begin{tabular*}{\hsize}{@{\extracolsep{\fill}}ccccccc} \hline
Intervention & Successful? & Long-term? & Attractor & Relative attractor & Attractor & Relative attractor \\
     &  &  & \%  & \% change& \% & \% change \\
& & & (permanent & (permanent  & (nonpermanent & (nonpermanent \\
& & & intervention) & intervention) & intervention) & intervention) \\
\hline
\multicolumn7c{Th1 stable motif control interventions ($C_{Th1}$)} \\ \hline
\{TBET=ON\} & Yes & Yes & 100.0 & 100 & 100.0 & 100 \\
\hline
\multicolumn7c{Th2 stable motif control interventions ($C_{Th2}$)} \\
\hline
\{GATA3=ON\} & Yes & Yes & 100.0 & 120 & 100.0 & 120 \\
\hline
\multicolumn7c{Th17 stable motif control interventions ($C_{Th17}$)} \\
\hline
\{GATA3=OFF, & Yes & Yes &100.0 & 3437 & 100.0 & 3437 \\
FOXP3=OFF, &  &  & &  &  &  \\
TBET=OFF, &  &  & &  &  &  \\
STAT3=ON\} &  &  & &  &  &  \\
\{GATA3=OFF, & Yes & Yes &100.0 & 3437 & 100.0 & 3437 \\
FOXP3=OFF, &  &  & &  &  &  \\
TBET=OFF, &  &  & &  &  &  \\
IL10=ON\} &  &  & &  &  &  \\
\{GATA3=OFF, & Yes & Yes &100.0 & 3437 & 100.0 & 3437 \\
FOXP3=OFF, &  &  & &  &  &  \\
TBET=OFF, &  &  & &  &  &  \\
IL10R=ON\} &  &  & &  &  &  \\
\{GATA3=OFF, & Yes & Yes &100.0 & 3437 & 100.0 & 3437 \\
FOXP3=OFF, &  &  & &  &  &  \\
TBET=OFF, &  &  & &  &  &  \\
IL21=ON\} &  &  & &  &  &  \\
\{GATA3=OFF, & Yes & Yes &100.0 & 3437 & 100.0 & 3437 \\
FOXP3=OFF, &  &  & &  &  &  \\
TBET=OFF, &  &  & &  &  &  \\
IL21R=ON\} &  &  & &  &  &  \\
\{GATA3=OFF, & Yes & Yes &100 & 3437 & 100 & 3437 \\
FOXP3=OFF, &  &  & &  &  &  \\
TBET=OFF, &  &  & &  &  &  \\
IL23R=ON, &  &  & &  &  &  \\
RORGT=ON\} &  &  & &  &  &  \\
\hline
\multicolumn7c{Treg stable motif control interventions ($C_{Treg}$)} \\
\hline
\{GATA3=OFF, & Yes & Yes &100.0 & 5613 & 100.0 & 5613 \\
FOXP3=ON, &  &  & &  &  &  \\
TBET=OFF\} &  &  & &  &  &  \\
\{GATA3=OFF, & Yes & Yes &100.0 & 5613 & 100.0 & 5613 \\
TBET=OFF, &  &  & &  &  &  \\
STAT3=OFF\} &  &  & &  &  &  \\
\{GATA3=OFF, & Yes & Yes &100.0 & 5613 & 100.0 & 5613 \\
TBET=OFF, &  &  & &  &  &  \\
IL23R=OFF, &  &  & &  &  &  \\
IL10R=OFF, &  &  & &  &  &  \\
IL21R=OFF\} &  &  & &  &  &  \\
\{GATA3=OFF, & Yes & Yes &100.0 & 5613 & 100.0 & 5613 \\
TBET=OFF, &  &  & &  &  &  \\
IL23R=OFF, &  &  & &  &  &  \\
IL10=OFF, &  &  & &  &  &  \\
IL21R=OFF\} &  &  & &  &  &  \\
\{GATA3=OFF, & Yes & Yes &100.0 & 5613 & 100.0 & 5613 \\
TBET=OFF, &  &  & &  &  &  \\
IL23R=OFF, &  &  & &  &  &  \\
IL10R=OFF, &  &  & &  &  &  \\
IL21=OFF\} &  &  & &  &  &  \\
\hline
\end{tabular*}
\end{table*}

\clearpage

\begin{table*}[t]
\begin{tabular*}{\hsize}{@{\extracolsep{\fill}}ccccccc} \hline
Intervention & Successful? & Long-term? & Attractor & Relative attractor & Attractor & Relative attractor \\
     &  &  & \%  & \% change& \% & \% change \\
& & & (permanent & (permanent  & (nonpermanent & (nonpermanent \\
& & & intervention) & intervention) & intervention) & intervention) \\
\hline
\multicolumn7c{Treg stable motif control interventions ($C_{Treg}$) (continuation)} \\
\hline
\{GATA3=OFF, & Yes & Yes &100.0 & 5613 & 100.0 & 5613 \\
TBET=OFF, &  &  & &  &  &  \\
IL23R=OFF, &  &  & &  &  &  \\
IL10=OFF, &  &  & &  &  &  \\
IL21=OFF\} &  &  & &  &  &  \\
\hline
\multicolumn7c{Single interventions of Th17 stable motif control sets} \\
\hline
\{GATA3=OFF\}  &  No  &  No  & 1.7 & -41 & 1.7 & -40\\
\{FOXP3=OFF\}  &  Partial  &  Partial  & 4.3 & 54 & 3.7 & 30\\
\{TBET=OFF\}  &  Partial  &  Partial   & 3.8 & 35 & 3.8 & 34\\
\{STAT3=ON\}  &  Partial  &  Partial   & 4.0 & 41 & 4.0 & 41\\
\{IL10=ON\}  &  Partial  &  Partial  & 3.8 & 33 & 3.8 & 34\\
\{IL10R=ON\}  &  Partial  &  Partial   & 3.8 & 34 & 3.8 & 35\\
\{IL21=ON\}  &  Partial  &  Partial   & 3.8 & 33 & 3.8 & 34\\
\{IL21R=ON\}  &  Partial  &  Partial  & 3.8 & 34 & 3.8 & 35\\
\{IL23R=ON\}  &  Partial  &  Partial  & 3.8 & 34 & 3.8 & 35\\
\{RORGT=ON\}  &  No  &  No  & 3.0 & 7 & 3.0 & 6\\
\hline
\multicolumn7c{Single interventions of Treg stable motif control sets} \\
\hline
\{GATA3=OFF\}  &  No  &  No  & 1.4 & -18 & 1.6 & -10\\
\{FOXP3=ON\}  &  Yes  &  Yes  & 4.8 & 172 & 4.7 & 167\\
\{TBET=OFF\}  &  Partial  &  Partial  & 2.6 & 47 & 2.6 & 49\\
\{STAT3=OFF\}  &  Yes  &  Yes  & 4.1 & 137 & 4.2 & 137\\
\{IL21=OFF\}  &  Partial  &  Partial  & 2.4 & 35 & 2.4 & 39\\
\{IL21R=OFF\}  &  Partial  &  Partial  & 2.5 & 40 & 2.6 & 46\\
\{IL23R=OFF\}  &  Yes  &  Yes  & 2.0 & 14 & 2.0 & 13\\
\hline
\end{tabular*}
\end{table*}

\clearpage

\begin{table*}
\caption{Validation of some stable motif control intervention targets in Table 1 for different Hill coefficients ($n$) in the T-LGL leukemia differential equation network model. The percentages are significant in the digits shown and have an estimated absolute error (standard deviation of the mean) of $6\cdot10^{-3}[\%p_{Attr}(100\%-\%p_{Attr})]^{1/2}\ \%$, where $\%p_{Attr}$ is the percentage shown (e.g. 0.06\% for a $\%p_{Attr}$ of 1\%, and 0.3\% for a $\%p_{Attr}$ of 50\%).}
\label{tab:SimulationsTLGLODEHill}
\begin{tabular*}{\hsize}{@{\extracolsep{\fill}}ccccccccc}
\hline
Intervention & \multicolumn{4}{c}{Apoptosis \% (permanent intervention)} & \multicolumn{4}{c}{Apoptosis \% (nonpermanent intervention)} \\
\hline
\multicolumn{9}{c}{Different Hill coefficients ($n$)} \\
\hline
& \multicolumn{8}{c}{$n$}  \\ \cline{2-9}
 & 1 & 1.5 & 2 & 2.5 & 1 & 1.5 & 2 & 2.5  \\
\hline
\multicolumn{9}{c}{T-LGL stable motif control interventions ($C_{TLGL}$)} \\
\hline
\{S1P=ON\}  & 0.0 & 0.0 & 0.0 & 0.0& 0.0 & 0.0 & 0.0 & 0.0 \\
\{Ceramide=OFF,  & 0.0 & 0.0 & 0.0 & 0.0& 0.0 & 0.0 & 0.0 & 0.0 \\
PDGFR=ON\}  & &  &  & & & & &  \\
\hline
\multicolumn{9}{c}{Apoptosis stable motif control interventions ($C_{Apoptosis}$)} \\
\hline
\{SPHK1=OFF\} & 100.0 & 100.0 & 100.0 & 100.0 & 100.0 & 100.0 & 100.0 & 100.0  \\
\{TBET=ON,  & 100.0 & 100.0 & 100.0 & 100.0 & 100.0 & 100.0 & 100.0 & 100.0 \\
Ceramide=ON,    & &  &  & & & & & \\
IL2RB=ON\}    & &  &  & & & & & \\
\hline
\end{tabular*}
\end{table*}

\begin{table*}[!h]
\caption{Validation of some stable motif control intervention targets in Table 1 for different Hill coefficients ($n$) in the T-LGL leukemia differential equation network model with randomly chosen $\tau_i$ and $\theta_i$. The percentages are significant in the digits shown and have an estimated absolute error (standard deviation of the mean) of $5\cdot10^{-3}[\%p_{Attr}(100\%-\%p_{Attr})]^{1/2}\ \%$, where $\%p_{Attr}$ is the percentage shown (e.g. 0.05\% for a $\%p_{Attr}$ of 1\%, and 0.25\% for a $\%p_{Attr}$ of 50\%).}
\label{tab:SimulationsTLGLODEParams}
\begin{tabular*}{\hsize}{@{\extracolsep{\fill}}ccccccccccc}
\hline
Intervention & \multicolumn{5}{c}{Apoptosis \% (permanent intervention)} & \multicolumn{5}{c}{Apoptosis \% (nonpermanent intervention)} \\
\cline{2-11}
& \multicolumn{10}{c}{$n$} \\ \cline{2-11}
 & 1 & 1.5 & 2 & 2.5 & 3 & 1 & 1.5 & 2 & 2.5 & 3 \\
\hline
\multicolumn{11}{c}{T-LGL stable motif control interventions ($C_{TLGL}$)} \\
\hline
\{S1P=ON\}  & 0.0 & 0.0 & 0.0 & 0.0& 0.0 & 0.0 & 0.0 & 0.0 & 0.0& 0.0 \\
\{Ceramide=OFF,  & 0.0 & 0.0 & 0.0 & 0.0& 0.0 & 0.0 & 0.0 & 0.0 & 0.0& 0.0  \\
PDGFR=ON\}  & &  &  & & & & & & & \\
\hline
\multicolumn{11}{c}{Apoptosis stable motif control interventions ($C_{Apoptosis}$)} \\
\hline
\{SPHK1=OFF\} & 100.0 & 100.0 & 100.0 & 99.9 & 99.5 & 100.0 & 100.0 & 100.0 & 99.9 & 99.5 \\
\{TBET=ON,  & 100.0 & 100.0 & 100.0 & 100.0 & 100.0 & 100.0 & 100.0 & 100.0 & 100.0 & 100.0 \\
Ceramide=ON,    & &  &  & & & & & & &  \\
IL2RB=ON\}    & &  &  & & & & & & & \\
\hline
\end{tabular*}
\end{table*}

\begin{table*}[!h]
\caption{Validation of some stable motif control intervention targets in Table 1 when fixing the intervened node variables close to the intervention-prescribed values in the T-LGL leukemia differential equation network model. If the intervention is 0 (1), the node variable is fixed at 0.1 (0.9), 0.8 (0.2), 0.7 (0.3), or 0.6 (0.4). The percentages are significant in the digits shown and have an estimated absolute error (standard deviation of the mean) of $6\cdot10^{-3}[\%p_{Attr}(100\%-\%p_{Attr})]^{1/2}\ \%$, where $\%p_{Attr}$ is the percentage shown (e.g. 0.06\% for a $\%p_{Attr}$ of 1\%, and 0.3\% for a $\%p_{Attr}$ of 50\%).}
\label{tab:SimulationsTLGLODEConcentr}
\begin{tabular*}{\hsize}{@{\extracolsep{\fill}}ccccccccc}
\hline
& \multicolumn{8}{c}{Fixed value of intervened node variable} \\ \cline{2-9}
 & 0.1/0.9 & 0.2/0.8 & 0.3/0.7 & 0.4/0.6 & 0.1/0.9 & 0.2/0.8 & 0.3/0.7 & 0.4/0.6   \\
\hline
\multicolumn{9}{c}{T-LGL stable motif control interventions ($C_{TLGL}$)} \\
\hline
\{S1P=ON\}  & 0.0 & 0.0 & 0.0 & 0.0& 0.0 & 0.0 & 0.0 & 0.0  \\
\{Ceramide=OFF,  & 0.0 & 0.0 & 0.0 & 0.0& 0.0 & 0.0 & 0.0 & 0.0  \\
PDGFR=ON\}  & &  &  & & & & & \\
\hline
\multicolumn{9}{c}{Apoptosis stable motif control interventions ($C_{Apoptosis}$)} \\
\hline
\{SPHK1=OFF\} & 99.9 & 99.9 & 99.9 & 100.0 & 99.9 & 99.9 & 99.9 & 100.0  \\
\{TBET=ON,  & 100.0 & 100.0 & 100.0 & 100.0 & 100.0 & 100.0 & 100.0 & 100.0  \\
Ceramide=ON,    & &  &  & & & & &   \\
IL2RB=ON\}    & &  &  & & & & & \\
\hline
\end{tabular*}
\end{table*}

\end{document}